\documentclass[12pt]{article}
\usepackage{graphicx,amsfonts,amsmath,amssymb,amsthm,mathrsfs,url,hyperref}
\usepackage[usenames,dvipsnames]{xcolor}

\title{Consistency Proof for Multi-Time Schr\"odinger Equations with Particle Creation\\ and Ultraviolet Cut-Off}
\author{
Sascha Lill,\footnote{Mathematisches Institut,
     Eberhard-Karls-Universit\"at, Auf der Morgenstelle 10, 72076
     T\"ubingen, Germany}\ \footnote{E-Mail: sascha.lill@uni-tuebingen.de}~
Lukas Nickel,\footnote{Studienstiftung des deutschen Volkes,
Ahrstra\ss e 41, 53175 Bonn, Germany} \ and
Roderich Tumulka$^*$\footnote{E-mail:
     roderich.tumulka@uni-tuebingen.de}
} 
\date{October 28, 2020}

\usepackage[utf8]{inputenc}
\usepackage{tikz}

\usepackage[section]{placeins}
\usepackage{siunitx}
\usepackage{capt-of}%for captions in minipages
\usepackage{paralist}
\usepackage{booktabs}%setting distances in tables manually
\usepackage{hhline}
\usepackage{color}                  %Colors

\newtheorem{lemma}{Lemma}
\newtheorem{theorem}{Theorem}
\theoremstyle{definition}\newtheorem{definition}{Definition}

\newcommand{\be}{\begin{equation}}
\newcommand{\ee}{\end{equation}}
\newcommand{\ba}{\boldsymbol{a}}
\newcommand{\bj}{\boldsymbol{j}}
\newcommand{\bk}{\boldsymbol{k}}
\newcommand{\bn}{\boldsymbol{n}}
\newcommand{\br}{\boldsymbol{r}}
\newcommand{\bs}{\boldsymbol{s}}
\newcommand{\bx}{\boldsymbol{x}}
\newcommand{\by}{\boldsymbol{y}}
\newcommand{\bz}{\boldsymbol{z}}
\newcommand{\bH}{\boldsymbol{H}}
\newcommand{\bN}{\boldsymbol{N}}

\newcommand{\bU}{\boldsymbol{U}}
\newcommand{\bX}{\boldsymbol{X}}
\newcommand{\bY}{\boldsymbol{Y}}
\newcommand{\bZ}{\boldsymbol{Z}}
\newcommand{\balpha}{\boldsymbol{\alpha}}
\newcommand{\bDelta}{\boldsymbol{\Delta}}
\newcommand{\bnabla}{\boldsymbol{\nabla}}
\newcommand{\bzero}{\boldsymbol{0}}
\newcommand{\bone}{\boldsymbol{1}}
\newcommand{\cP}{\mathcal{P}}
\newcommand{\cQ}{\mathcal{Q}}
\newcommand{\sH}{\mathscr{H}}
\newcommand{\sM}{\mathscr{M}}
\newcommand{\sS}{\mathscr{S}}
\newcommand{\sX}{\mathscr{X}}
\newcommand{\sY}{\mathscr{Y}}
\newcommand{\CCC}{\mathbb{C}}
\newcommand{\HHH}{\mathbb{H}}
\newcommand{\NNN}{\mathbb{N}}
\newcommand{\RRR}{\mathbb{R}}
\newcommand{\G}{G} % symbol for Green's function
\newcommand{\free}{\mathrm{free}}
\newcommand{\inter}{\mathrm{int}}
\newcommand{\cutoff}{\varphi}
\newcommand{\supp}{\mathrm{supp}}
\newcommand{\Gr}{\mathrm{Gr}}
\newcommand{\comp}{\boldsymbol{R}} %comparison operator

\renewcommand{\Im}{\mathrm{Im}}

\newcommand{\x}[1]{{#1}}

\newcounter{remarks}
\begin{document}
\maketitle
\begin{abstract}
For multi-time wave functions, which naturally arise as the relativistic particle-position representation of the quantum state vector, the analog of the Schr\"odinger equation consists of several equations, one for each time variable. This leads to the question of how to prove the consistency of such a system of PDEs. The question becomes more difficult for theories with particle creation, as then different sectors of the wave function have different numbers of time variables. Petrat and Tumulka (2014) gave an example of such a model and a non-rigorous argument for its consistency. We give here a rigorous version of the argument after introducing an ultraviolet cut-off into the creation and annihilation terms of the multi-time evolution equations. These equations form an infinite system of coupled PDEs; they are based on the Dirac equation but are not fully relativistic (in part because of the cut-off). We prove the existence and uniqueness of a smooth solution to this system for every initial wave function from a certain class that corresponds to a dense subspace in the appropriate Hilbert space.

\medskip

\end{abstract}

\newpage
\tableofcontents

\section{Introduction}	
\label{sec:intro}

The quantum state of $N$ particles is usually described by means of a 
wave function
\be
\Psi(t,\bx_1,\ldots,\bx_N)
\ee
that is a function of time $t\in \RRR$ and the positions 
$\bx_j\in\RRR^3$ of the particles. The obvious relativistic 
generalization is a wave function
\be
\Phi\bigl( (t_1,\bx_1),\ldots,(t_N,\bx_N) \bigr)
\ee
of $N$ space-time points $(t_j,\bx_j)\in\sM=\RRR^4$ (setting $c=1$), 
called a \emph{multi-time wave function} 
\x{\cite{dirac:1932,bloch:1934,LPT:2020}}. The relation between $\Psi$ and 
$\Phi$ is straightforward from the fact that $\Psi$ also refers to $N$ 
space-time points, $(t,\bx_1), \ldots,(t,\bx_N)$, which are simultaneous 
relative to the chosen Lorentz frame; that is, the single-time wave 
function $\Psi$ is recovered from $\Phi$ by setting all times equal,
\be\label{psiphit}
\Psi(t,\bx_1,\ldots,\bx_N) = \Phi\bigl( (t,\bx_1),\ldots,(t,\bx_N) \bigr)\,.
\ee
The domain of $\Phi$ is usually the set of spacelike configurations,
\be\label{SNdef}
\sS^{(N)} := \Bigl\{ (x_1,...,x_N)\in \sM^N ~:~ \forall j,k: 
(x_j-x_k)^\mu(x_j-x_k)_\mu <0 \text{ or }x_j=x_k\Bigr\} \,,
\ee
relative to the Minkowski metric $\mathrm{diag}(1,-1,-1,-1)$. The 
advantage of using $\Phi$ is that it is a covariant object, defined 
without reference to any hypersurface---nor in fact to coordinates, if 
we think of $\Phi$ as a function of $N$ points in the space-time manifold $\sM$.

The concept of a multi-time wave function $\Phi$ is closely related to 
that of associating with every spacelike hypersurface $\Sigma$ a wave 
function $\Psi_\Sigma$, as used by Tomonaga \cite{tomonaga:1946} and 
Schwinger \cite{schwinger:1948}: given $\Phi$, we can define 
$\Psi_\Sigma$ on $\Sigma^N$ by simply setting
\be\label{psiSigmaphi}
\Psi_\Sigma(x_1,\ldots,x_N) := \Phi(x_1,\ldots,x_N) ~~~\text{for } 
(x_1,\ldots,x_N)\in\Sigma^N\,.
\ee
Still, $\Phi$ is a more elementary concept than $\Psi_\Sigma$, as $\Phi$ 
is simply a function of $4N$ variables. Moreover, as we will see, the multi-time wave 
function $\Phi$ can be adapted to the situation with an ultraviolet (UV) 
cut-off, while $\Psi_\Sigma$ cannot.

In a quantum field theory (QFT), the particle-position representation of 
the quantum state vector $|\Psi\rangle$ in Fock space, whenever that 
exists, naturally yields a multi-time wave function
\be
\Phi(x_1,\ldots,x_N) := \langle \emptyset|a(x_1)\cdots a(x_N)|\Psi\rangle\,,
\ee
where $|\emptyset\rangle$ is the Fock vacuum and $a(x)$ is the 
annihilation operator in the position representation at $x=(x^0,\bx)\in 
\sM$. Since $N$ is now variable, this $\Phi$ is a multi-time 
Fock function, i.e., a function on
\be
\sS := \bigcup_{N=0}^\infty \sS^{(N)}\,.
\ee
We take the approach of \emph{defining} a QFT model in a manifestly 
covariant way (without ever choosing a Lorentz frame or spacelike 
hypersurface) by starting from $\Phi$ as a solution of a suitable multi-time 
variant of the Schr\"odinger equation. 
However, while this procedure is physically the ultimate goal, we have to introduce an UV cut-off for the sake of mathematical rigor, which breaks Lorentz invariance. That is why we here formulate a multi-time evolution law for $\Phi$ corresponding to a simplified QFT model with cut-off in a fixed Lorentz frame. In that setting, we can then prove existence and uniqueness of solutions for sufficiently regular initial conditions, which demonstrates that the approach makes sense.

\subsection{Multi-Time Evolution Laws}

In order to discuss multi-time evolution laws, let us begin again with 
the case of a fixed number $N$ of particles. As $\Phi$ depends on $N$ 
time variables, its evolution can be governed by a system of $ N $ 
Schr\"odinger equations (setting $\hbar=1$):
\begin{align}
     i \partial_{t_1} \Phi &= H_1 \Phi\nonumber\\
     &\vdots \label{eq:multisys}\\
     i \partial_{t_N} \Phi &= H_N \Phi\nonumber
\end{align}
with ``partial Hamiltonians'' $ H_k $. By \eqref{psiphit} and the chain 
rule, the sum of the $H_k$ should be the Hamiltonian for $\Psi$ at every 
simultaneous configuration.

As a consequence of 
$\partial_{t_j}\partial_{t_k}\Phi=\partial_{t_k}\partial_{t_j}\Phi$, we 
obtain a condition on the $H_k$,
\be
     \bigl[ i \partial_{t_j} - H_j , i \partial_{t_k} - H_k \bigr] = 0 
\quad \forall j,k \in \{1,\ldots,N\}\,,
\label{eq:consist}
\ee
known as the \emph{integrability condition} or \emph{consistency 
condition} for the system \eqref{eq:multisys} 
\cite{dirac:1932,dfp:1932,bloch:1934,pt:2013a}. If it is violated, then 
the equations \eqref{eq:multisys} cannot be expected to be 
simultaneously solvable, except perhaps for very special initial data. An \emph{initial datum} means here the restriction of $\Phi$ to space-time configurations on the hyperplane $\{x^0=0\}$ in space-time $\sM$ or, in other words, the values of $\Phi$ when all time variables are set to 0. To \emph{prove consistency} of \eqref{eq:multisys} amounts to proving the existence and uniqueness of the solution for a sufficiently large set of initial data (such as a dense subset of Hilbert space). The commutator condition \eqref{eq:consist} is closely related to the integrability condition of the Tomonaga--Schwinger equation, which in turn is closely related to the causality axiom of the Wightman axioms \cite[p.~65]{reedsimon2}.

Since interaction potentials violate \eqref{eq:consist} \cite{pt:2013a, 
ND:2016}, interaction needs to be implemented in a different way, such 
as through particle creation \cite{pt:2013c,pt:2013d,LN:2018}, 
zero-range interaction \cite{lienert:2015a,lienert:2015c,LN:2015,KLTZ19}, 
interaction along light cones 
\cite{lienert:2018,LT:2018,LT:2019,LN:2019}, or other ways 
\cite{dfp:1932,DV82b,DV85,ND:2019}. In this paper, we focus on particle 
creation.

If the particle number is variable, then the wave function $\Phi$ on 
$\sS$ consists (like a vector in Fock space) of sectors with 
different particle numbers, and thus with different numbers of time 
variables. Correspondigly, the terms $H_k\Phi$ on the right-hand side of 
\eqref{eq:multisys} can involve other sectors, and the consistency 
question becomes more involved. As a consequence, it is not obvious that the ``consistency condition'' 
\eqref{eq:consist} is actually sufficient for consistency; and even less obvious if the solution $\Phi$ is not required to exist on \emph{all} space-time configurations but only on the set $\sS$ of \emph{spacelike} configurations. Thus, a consistency proof requires much more than just checking a commutator condition such as \eqref{eq:consist}.

\subsection{Model and Goals}

Our model, adopted from \cite{pt:2013c} and inspired by models of Lee 
\cite{Lee54}, Schweber \cite[Sec.~12a]{schweber:1961}, and Nelson \cite{Nel64}, is 
a simple QFT in 1+3 space-time dimensions in which one species of 
particles, called $x$-particles in the following, can emit and absorb 
particles of another kind, called $y$-particles; in short,
\be\label{xxy}
x\leftrightarrows x+y\,.
\ee
We take the $x$-particles to be fermions and the $y$s to be bosons, and 
we take both to be massive and to have spin $\tfrac{1}{2}$ (although in 
nature bosons have integer spin). The free Hamiltonians are Dirac 
operators. A multi-time formulation of this model was given in 
\cite{pt:2013c}, however in a UV divergent form. In order to enable a
rigorous treatment, we introduce a UV cut-off; correspondingly, we need to
slightly modify the definition of the domain $\sS$ of $\Phi$ (i.e.,
to replace it by the set $\sS_\delta$ of $\delta$-spacelike configurations, see below). 
We formulate the multi-time Schr\"odinger equations of our model with cut-off (Section~\ref{sec:physmod})
and prove the existence and uniqueness of a 
solution for every 
initial condition of sufficient regularity (Theorem~\ref{thm:1} in Section~\ref{sec:results}). In particular, this result 
proves consistency of the multi-time equations; some steps towards a consistency proof of this kind were already taken in \cite{p:2010}.

\x{It would be desirable to treat more realistic models in the future. Here, the model is chosen for allowing a rigorous proof of the multi-time consistency for a variable number of time variables; in fact, our proof is the first such proof for a model in 1+3 dimensions.}

The solutions we will construct are strong solutions, i.e., smooth 
functions that have derivatives in the classical sense. As a by-product 
of our analysis, we also prove smoothness of the single-time wave 
function $\Psi$ as a function of $t$, the $\bx_k$, and the $\by_\ell$ coordinates (Lemma~\ref{lem:diffable} in Section~\ref{sec:lemmas}); as far as we know, 
this fact was not in the literature before for any similar QFT model. 

Furthermore, we prove a statement (Lemma~\ref{lem:supp} in Section~\ref{sec:lemmas}) on how fast the support of $\Psi$ in $\RRR^3$ can grow with time: it can only grow at the speed of light, except for an additional instantaneous growth by the cut-off length $\delta$. A corresponding statement holds true for the multi-time wave function $\Phi$: its evolution is propagation local up to $\delta$; it is also interaction local up to $2\delta$.

\subsection{Motivation}
\label{sec:motivation}

\x{Although in QFT one often focuses on the operators, there are of course also quantum state vectors. And quantum state vectors are something real, as various considerations such as the Pusey-Barrett-Rudolph theorem \cite{PBR} indicate. By that it is meant that quantum states are things in nature and in that sense physical; they are not like probability distributions (representing the observer's knowledge). In view of that, and since the usual Fock space vectors refer to a chosen Lorentz frame, we are led to the question, how can the quantum state be represented in a covariant way? This is a conceptual, physically relevant question that naturally leads us to considerung multi-time wave functions, as they provide a covariant particle-position representation of the quantum state.

Most approaches to QFT focus on the operators and put the link with space-time points into the field operators, while the quantum state is a functional on the algebra or an element of an abstract Hilbert space that is not directly related to space-time. In contrast, the approach we follow here puts a link with space-time into the quantum state. Indeed, just as in non-relativistic quantum mechanics, where the quantum state can be represented by a wave function depending on many space-time points, we describe the state of a quantum field theory as a function $\Phi$ of space-time points. 

Multi-time wave functions thus offer a reformulation of QFTs and an alternative perspective on them. Richard Feynman \cite{Fey65} wrote in 1965: ``I, therefore, think that a good theoretical physicist today might find it useful to have a wide range of physical viewpoints and mathematical expressions of the same theory (for example, of quantum electrodynamics) available to him.'' We think this is still true today.

Our perspective on QFT provides a kind of ``pedestrian'' approach that expresses a QFT as a system of covariant PDEs. And indeed, the perspective of the particle-position representation of the quantum state has proven useful, for example a few years ago through the realization that certain UV-divergent particle creation terms can be tamed by means of certain boundary conditions on the wave function, so-called interior-boundary conditions \cite{ibc}. Another advantage of multi-time wave functions is that they are rather intuitive objects and governed by remarkably simple and natural equations, see \eqref{multi12}. Moreover, this approach can be used as a \emph{derivation} of the mathematical structure of QFT: instead of starting from (say) the Wightman axioms for the field operators, we may start from non-relativistic quantum mechanics, ask for including particle creation and annihilation while observing Lorentz invariance, and are then led \cite{LPT:2020} to multi-time wave functions and evolution equations such as \eqref{multi12}, which are presumably ultimately equivalent to operator-centered approaches to QFT.

Many physicists are worried that a particle-position representation may not be possible. Well, it is certainly of interest to explore to which extent it is possible, and the situation is often better than it may seem. For example, it is sometimes said that photons cannot be localized, but that boils down to difficulties with specifying the \emph{probability density} in position space, not with specifying the \emph{wave function} in position space \cite{BB96}; a photon wave function is mathematically equivalent to a (complexified) Maxwell field and thus unproblematical. 

Another advantage of multi-time wave functions is that they, like Tomonaga-Schwinger wave functions, provide a quantum state on every \emph{arbitrary} (curved) Cauchy surface (at least in the absence of UV cut-offs), something not available through a representation of the Poincar\'e group as provided by the Wightman axioms. An advantage over Tomonaga-Schwinger wave functions is that multi-time wave functions can be considered with UV cut-offs. Finally, multi-time evolution equations may also be expected to generalize rather easily to curved space-time, although this question has not yet been fully explored in the literature.

Now the possibility of multi-time wave functions is crucially linked to a mathematical problem: that of consistency of the multi-time evolution equations. And that is what this paper contributes to. 

In this work, Lorentz invariance has a delicate status: On the one hand, Lorentz invariance is a main motivation for using multi-time wave functions; on the other, the wish for a rigorous discussion leads us to introducing an UV cut-off, which breaks the Lorentz invariance. Then again, its violation is limited to certain aspects while other aspects of Lorentz invariance are retained. For example, although the propagation of wave functions is not confined to the light cone, it is confined to a neighborhood of the light cone with spatial radius given by the cut-off length $\delta>0$ (see Figure~\ref{fig:supp2} and Eq.~\eqref{thm1Gr} for more precise elucidation). It would be desirable to work without UV cut-off, but presently that is not an option in 1+3 dimensions.
}

\subsection{Comparison to Prior Works}

Our work is similar to recent work in \cite{LN:2018,ND:2019}. In 
\cite{LN:2018}, a rigorous consistency proof 
is presented for a QFT model in 1+1 
space-time dimensions involving one species (``$x$'') of 
spin-$\tfrac{1}{2}$ particles with creation and annihilation according 
to $x \leftrightarrows x+x$ instead of \eqref{xxy}; instead of a UV 
cut-off, that model uses interior-boundary conditions. In 
\cite{ND:2019}, a rigorous consistency proof is presented for the QFT 
model of Dirac, Fock, and Podolsky \cite{dfp:1932} in 1+3 space-time 
dimensions with UV cut-off, in which $x$-particles interact through a 
quantized $y$-field; the main difference to our model is that here we 
give a separate time variable to every $y$-particle. Further consistency 
proofs for multi-time equations are contained in \cite{pt:2013a}, where 
it was shown for any fixed number $N$ of particles that the consistency 
condition \eqref{eq:consist} is necessary and sufficient for consistency 
of \eqref{eq:multisys}, provided that $\Phi$ is defined on all of 
$\sM^N$ (instead of $\sS^{(N)}$) and that the partial 
Hamiltonians are bounded operators on 
$L^2(\RRR^{3N},\CCC^k)$ that depend smoothly on 
$t_1,...,t_N$, or else are (possibly unbounded) self-adjoint operators 
and do not depend on $t_1,...,t_N$. Rigorous consistency proofs for multi-time equations governing $N$ particles with zero-range interaction in 1+1 space-time dimensions were given in \cite{lienert:2015a,lienert:2015c,LN:2015,KLTZ19}.

An overview of work on multi-time wave functions can be found in 
\x{\cite{LPT:2017a,LPT:2020}}. The idea of a multi-time wave function was conceived as early as 1929 \cite{Edd1929,Gaunt1929,Mott1929}.
Early examples of consistent multi-time evolutions 
with interaction were given in \cite{DV82b,CVA:1983,DV85,CVA:1997}. The 
approach is contrasted with the idea of multiple timelike dimensions in 
\cite{LPT:2017b} and with multi-time equations in classical mechanics in \cite{pt:2013e}. 
The appropriate version of the Born rule for $\Phi$ is 
formulated and proved in \cite{LT:2017} after pioneering work in 
\cite{Mott1929,bloch:1934}. 

A big difference between the mathematics of multi-time and single-time 
evolution concerns the role of Hilbert spaces. In the case of a fixed 
number $N$ of particles, we cannot simply fix the values of all time 
variables $t_1,...,t_N$ and consider the wave function as a function of 
the spatial variables $\bx_1,...,\bx_N$ alone because for some values of 
the $\bx_j$, $((t_1,\bx_1),...,(t_N,\bx_N))$ will not be a spacelike 
configuration and $\Phi$ will not be defined there. Nevertheless, 
Hilbert spaces play a role as a tool in our proof, as we will consider 
subsets of particles at equal time values. We will also make use of the 
fact that wave functions do not propagate faster than at the speed of 
light, so that $\Psi_t$, when considered only in a given region of 
3-space, is determined by initial data $\Psi_0$ in a suitably larger region.

This paper is organized as follows. Section~\ref{sec:physmod} is 
dedicated to the definition of the QFT model and the formulation of its 
multi-time Schr\"odinger equations. In Section~\ref{sec:results}, we describe our main results.
In Section~\ref{sec:singletime}, we 
establish lemmas about the single-time evolution that we will need as 
tools for the multi-time analysis. In Section~\ref{sec:multitime}, we 
then construct the solution to the multi-time equations and prove 
uniqueness of the solution. In Section~\ref{sec:conclusions}, we conclude.

\section{The Physical Model}	
\label{sec:physmod}

Our model is a toy QFT in which a fixed number $M\in\NNN$ of $x$-particles can emit and absorb $y$-particles. It is a UV-regularized version of the multi-time model of \cite{pt:2013c}.

\subsection{Original, UV Divergent Equations}

The multi-time model of \cite{pt:2013c} is defined by the formal evolution equations
\begin{subequations}\label{multi12}
\begin{align}\label{multi1}
i \frac{\partial\Phi^{(N)}}{\partial x_k^0}(x^{4M}, y^{4N}) &= H^{\free}_{x_k} \Phi^{(N)}(x^{4M}, y^{4N}) 
+ \sqrt{N+1}  \sum_{s_{N+1}=1}^4 g_{s_{N+1}}^* \, \Phi_{s_{N+1}}^{(N+1)}\bigl(x^{4M}, (y^{4N}, x_k)\bigr)\nonumber \\
& \quad + \frac{1}{\sqrt{N}} \sum_{\ell=1}^N G_{0,s_\ell}(y_\ell - x_k) \, \Phi_{\widehat{s_\ell}}^{(N-1)}\bigl(x^{4M}, y^{4N} \backslash y_\ell\bigr)  \\[3mm] 
\label{multi2}
i \frac{\partial\Phi^{(N)}}{\partial y_\ell^0}(x^{4M}, y^{4N}) &= H^{\free}_{y_\ell}\Phi^{(N)}(x^{4M}, y^{4N}).
\end{align}
\end{subequations}
Here, $N\in\NNN\cup\{0\}=:\NNN_0$, $q^4:=(x^{4M},y^{4N})=(x_1,\ldots,x_M,y_1,\ldots,y_N)\in \sS^{(M+N)}$ is a spacelike space-time configuration, $H^{\free}_{x_k}$ is the free Dirac operator acting on particle $x_k$, $g\in \CCC^4$ is a fixed spinor playing the role of a coupling constant, $s_\ell\in\{1,2,3,4\}$ is the spin index of particle $y_\ell$ while most spin indices are not explicitly written, $\widehat{\:}$ means omission, $\backslash$ means to remove an entry, and $G_0$ is a Green function: it is the $\CCC^4$-valued distribution on $\RRR^4$ that is the solution of
\be\label{Gdef1}
i\frac{\partial G_0}{\partial t} = H^{\free}_y \, G_0
\ee
with initial condition
\be\label{Gdef2}
G_{0,s}(0,\by) = g_s \delta^3(\by)\,. 
\ee
Complete definitions tailored to our model are given in Sections~\ref{sec:singletimemodeldef} and \ref{sec:multitimemodeldef} below. (The superscripts $(N)$, $(N+1)$ etc.\ are actually unnecessary because the sector is determined by the argument of $\Phi$; they are mentioned merely for easier readability.)

In contrast to \eqref{eq:multisys}, we here encounter a $q$-dependent (and unbounded) number of PDEs, since the number $N$ of $ y $-particles is unbounded, and a separate PDE is assigned to each particle.
In \cite{pt:2013c}, non-rigorous arguments were given for the consistency of \eqref{multi12}. 

It follows from \eqref{multi12} that the single-time wave function $\Psi$ evolves with the (UV-divergent) Hamiltonian 
\begin{align}
(H\Psi)(\bx^{3M},\by^{3N}) 
&= \sum_{k=1}^M H^{\free}_{x_k} \Psi(\bx^{3M},\by^{3N}) + \sum_{\ell=1}^N H^{\free}_{y_\ell} \Psi(\bx^{3M},\by^{3N})\nonumber\\
&\quad + \sqrt{N+1} \sum_{k=1}^M \sum_{s_{N+1}=1}^4 g_{s_{N+1}}^* \, \Psi_{s_{N+1}}\bigl(\bx^{3M}, (\by^{3N}, \bx_k)\bigr)\nonumber \\
& \quad + \frac{1}{\sqrt{N}} \sum_{k=1}^M \sum_{\ell=1}^N g_{s_\ell}\, \delta^3(\by_\ell - \bx_k) \, \Psi_{\widehat{s_\ell}}(\bx^{3M}, \by^{3N} \backslash \by_\ell) 
\label{Hdef}
\end{align}
at the spatial configuration $q^3:=(\bx^{3M},\by^{3N})=(\bx_1,\ldots,\bx_M,\by_1,\ldots,\by_N)$. Put in a different notation,
\begin{multline}
H =  d\Gamma_x(H^{\free}_{x}) + d\Gamma_y(H^{\free}_{y})\\
+\sum_{r,s=1}^4 \int_{\RRR^3}d^3\bx \, a_{x,r}^\dagger(\bx) \, a_{x,r}(\bx) \, \bigl(g_s^* a_{y,s}(\bx) + g_s a_{y,s}^\dagger(\bx)\bigr)\,,
\end{multline}
where we have assumed a fermionic Fock space for the $x$-particles (although $H$ will map its $M$-particle sector to itself). We have written $d\Gamma(S)$ for the second quantization of the 1-particle operator $S$, $H_y^\free$ for the 1-particle Dirac operator, and $a_x$ and $a_y$ for the annihilation operators for $x$- and $y$-particles in the position representation. 

In this paper, we formulate our proofs for the slightly more general possibility that the coupling coefficients $g$ also act on the spin index $r$ of the emitting or absorbing $x$-particle according to
\begin{multline}
H =  d\Gamma_x(H^{\free}_{x}) + d\Gamma_y(H^{\free}_{y})\\
+\sum_{r,r',s=1}^4 \int_{\RRR^3}d^3\bx \, a_{x,r}^\dagger(\bx) \, a_{x,r'}(\bx) \, \bigl(g_{r'rs}^* a_{y,s}(\bx) + g_{rr's} a_{y,s}^\dagger(\bx)\bigr)\,,
\end{multline}
with the consequence that the ``coupling constant'' $g$ is an element of $\CCC^4 \otimes \CCC^4 \otimes \CCC^4$. That is analogous to quantum electrodynamics, where the coupling coefficients are proportional to the 4-vector of Dirac gamma matrices $\gamma^\mu$ with their $4\times 4\times 4$ entries $\gamma^\mu_{rr'}$.

The UV divergence of $H$ manifests itself in \eqref{Hdef} in the fact that $\delta^3$ is not an $L^2$ function and thus $H\Psi\notin \sH$; the UV cut-off will consist in replacing $\delta^3$ by an $L^2$ function $\cutoff$. We now turn to defining mathematically the cut-off version, our model in this paper, first in the single-time formulation, then multi-time.

\subsection{Single-Time Formulation With Cut-Off}
\label{sec:singletimemodeldef}

Although our goal is a multi-time formulation, we will use the single-time formulation as a tool, along with the Hilbert space and the Hamiltonian operator.

To begin with, the wave function of a single Dirac particle can be described in two different ways: Either we assign a $ \CCC^4 $-vector to each point of $\RRR^3$, or we consider a $ \CCC $-valued function on $\RRR^3\times\{1,2,3,4\}$ (i.e., 4 disjoint copies of $ \RRR^3 $, see Figure~\ref{fig:configspace3D}):
\be
	\Psi \in 
	L^2(\RRR^3, \CCC^4) \cong L^2(\RRR^3\times\{1,2,3,4\}, \CCC) 
\ee

\begin{figure}[hbt]
    \centering
    \def\cube at (#1,#2){
\filldraw[draw=black, fill=blue!25!white] (#1,#2) rectangle ++(1,1);
\filldraw[draw=black, fill=blue!35!white] (#1+1,#2) -- ++(0.375,0.375) -- ++(0,1) -- ++(-0.375,-0.375) -- cycle;
\filldraw[draw=black, fill=blue!15!white] (#1,#2+1) -- ++(1,0) -- ++(0.375,0.375) -- ++(-1,0) -- cycle;
}
\begin{tikzpicture}
	\cube at (0.6,1.3);
	\fill (1.3,2) circle (0.06);
	\draw (1.5,2) .. controls (1.9,2) and (2.1,2.4) .. (2.3,2.8);
	\node at (2.3,3) {$ q $};
	\node at (3.6,1.7) {$ \Psi(q) = \begin{pmatrix} \Psi_1\\ \Psi_2\\ \Psi_3\\ \Psi_4 \end{pmatrix}$};
	\node at (1.3,0.5) {$ \mathcal{Q}_1 = \mathbb{R}^3 $};
	
	\cube at (7,1)
	\cube at (8.75,1)
	\cube at (7,2.6)
	\cube at (8.75,2.6)
	\fill (9.5,1.8) circle (0.06);
	\draw (9.7,1.8) .. controls (10,1.8) and (10.3,1.8) .. (10.6,2);
	\node at (10.8,2.2) {$ (q,s) $};
	\node at (11.8,1.2) {$ \Psi(q,s) \in \mathbb{C} $};
	\node at (8.6,0.5) {$ \mathcal{Q}_1^s = \mathbb{R}^3 \times \{1,2,3,4\}$};
	
\end{tikzpicture}
    \caption{Wave functions $ \RRR^3 \rightarrow \CCC^4 $ and $ \RRR^3\times\{1,2,3,4\} \rightarrow \CCC $ are equivalent.}
    \label{fig:configspace3D}
\end{figure}
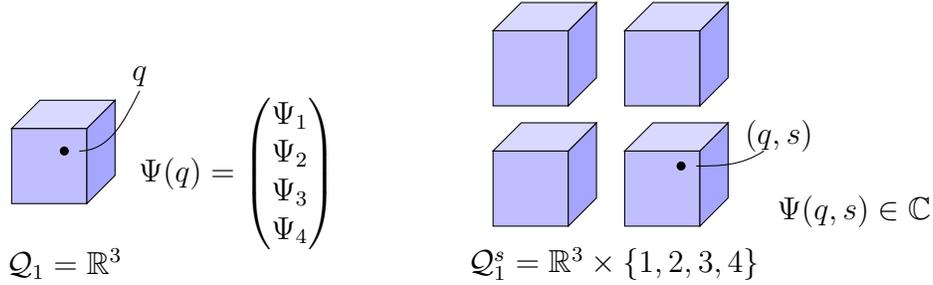

The corresponding \emph{spin-configuration spaces} \x{of $M$ $x$-particles and $N$ $y$-particles} are
\be
	\cQ^{s3}_x = (\RRR^3\times\{1,2,3,4\})^M, \quad\cQ_y^{s3,(N)} = (\RRR^3\times\{1,2,3,4\})^N.
\ee
\x{In our model, $M$ is fixed and $N$ is variable, so}
\be
	 \cQ^{s3} := \bigcup_{N = 0}^{\infty} \left( \cQ^{s3}_x \times \cQ_y^{s3,(N)} \right) =: \bigcup_{N = 0}^{\infty} \cQ^{s3,(N)}.
\label{eq:Q}
\ee
An element of $\cQ^{s3}$, i.e., a spin-configuration in 3D, will be denoted by $q^{s3}$. In the analogous way we define the ``spin-free'' configuration spaces $ \cQ_x^3, \cQ_y^3 $ and $ \cQ^3 $ (where the superscript means an index, not a Cartesian power). 

An equal-time quantum state is an element of the Hilbert space
\be\label{eq:H}
\sH=\sH_x\otimes \sH_y := L^2(\RRR^3, \CCC^4)^{\otimes_A M} \otimes \bigoplus_{N=0}^{\infty} L^2(\RRR^3, \CCC^4)^{\otimes_S N}\,,
\ee
where $M\in\NNN$ is the (fixed) number of $x$-particles, $\otimes_A$ denotes the anti-symmetric and $\otimes_S$ the symmetric tensor power; in particular, $\sH_y$ is the bosonic Fock space. Elements of $\sH$ can be represented as wave functions in $L^2(\cQ^{s3}):=L^2(\cQ^{s3}, \CCC)$, denoted by
\be
	\Psi = \Psi_{r_1,...,r_M,s_1,...,s_N}^{(N)}(\bx_1,...,\bx_M, \by_1,...,\by_N) =: \Psi_{\br,\bs}^{(N)}(\bx^{3M},\by^{3N}) = \Psi(q^{s3}),
\ee
which are suitably symmetric against permutations. Here, the spin indices are gathered in vectors $ \br \in \{1,2,3,4\}^M $ and $ \bs \in \{1,2,3,4\}^N $. 

The free time evolution is given by the Dirac operators
\be
	\bH^{\free}_{x_k} := - i \sum_{a = 1}^3 \alpha^a \; \frac{\partial}{\partial x^a_k} + m_x \beta , \quad \bH^{\free}_{y_\ell} := - i \sum_{a = 1}^3 \alpha^a \; \frac{\partial}{\partial y^a_\ell} + m_y \beta ,
\label{eq:diracops}
\ee
where $ \alpha^1, \alpha^2, \alpha^3 $ and $ \beta $ are the Dirac alpha and beta matrices. Fermions are assumed to have a rest mass $ m_x > 0 $ and bosons have a rest mass $ m_y > 0 $. (Domains of operators will be specified in Section~\ref{sec:results}. Henceforth, we use bold face font for operators on Hilbert space, as well as still for vectors in dimensions $3,3M,3N,M,N$.)

Let $e_1,\ldots,e_4$ denote the standard basis in $\CCC^4$. The annihilation operator of a $y$-particle in the position representation at location $\bx$ in the spin state $e_s$ will be denoted by $\ba(e_s\delta^3(\cdot-\bx))$. To implement the UV cut-off, we will replace $\delta^3$ in the Hamiltonian by a smooth (i.e., infinitely often differentiable) cut-off function 
\be
\cutoff \in C_c^{\infty}(\RRR^3, \RRR)
\ee
with compact support inside a ball $ B_\delta(\bzero)$ of radius $\delta > 0 $ around the origin $\bzero\in\RRR^3$. Physically, the radius $\delta$ should be small, but our mathematical results apply to any positive $\delta$. The corresponding smeared-out annihilation operator $\ba(e_s\cutoff(\cdot-\bx))$ will be abbreviated as $\ba_s(\bx)$. For the location $\bx\in\RRR^3$, we will need to insert the position operator of (say) the $k$-th $x$-particle, and the resulting annihilation operator will be denoted $\ba_s(\bx_k^{op})$. That is,
\be
\begin{aligned}
	(\ba_s(\bx_k^{op}) \Psi)^{(N)}(\bx^{3M},\by^{3N}) &= \sqrt{N + 1} \int d^3 \boldsymbol{\tilde{y}} \, \cutoff(\boldsymbol{\tilde{y}} - \bx_k) \Psi^{(N+1)}_{s_{N+1} = s}(\bx^{3M},(\by^{3N},\boldsymbol{\tilde{y}}))\\
	(\ba_s^{\dagger}(\bx_k^{op}) \Psi)^{(N)}(\bx^{3M},\by^{3N}) &= \frac{1}{\sqrt{N}} \sum_{\ell=1}^{N} \delta_{s s_\ell} \cutoff(\by_\ell - \bx_k) \Psi^{(N-1)}_{\widehat{s_\ell}}(\bx^{3M},\by^{3N} \setminus \by_\ell).
\end{aligned}
\label{eq:a}
\ee

The interaction Hamiltonian is of the form
\be
\bH^{\inter}= \sum_{k=1}^M \bH_{x_k}^{\inter}\,,
\ee
where the $k$-th term represents emission and absorption by particle $x_k$ and is defined by
\be
	(\bH_{x_k}^{\inter} \Psi)_{r_k = r} := \sum_{s,r' = 1}^4 \Bigl( g^*_{r'rs} \ba_s(\bx_k^{op}) \Psi_{r_k = r'} + g_{rr's}\ba_s^{\dagger}(\bx_k^{op})  \Psi_{r_k = r'}\Bigr)\,.
\label{eq:Hint} 
\ee
The full Hamiltonian $\bH : \sH \supset \mathrm{dom}(\bH) \rightarrow \sH$ (where $\mathrm{dom}$ means domain) of the model reads
\be
	\bH := \sum_{k = 1}^M \left( \bH_{x_k}^{\free} + \bH_{x_k}^{\inter} \right) + d\Gamma_y(\bH_{y}^{\free}).
\label{eq:Hamiltonian}
\ee

\subsection{Multi-Time Formulation With Cut-Off}
\label{sec:multitimemodeldef}

We now write down the system of multi-time equations corresponding to the Hamiltonian \eqref{eq:Hamiltonian}; the equations are a version of \eqref{multi12} with UV cut-off $\cutoff$. As a preparation, we first need to

\begin{itemize}
\item define the set $\sS_\delta$ of multi-time configurations $q^4=(x^{4M},y^{4N})$ for which $\Phi(q^4)$ will be defined (Section~\ref{sec:Sdeltadef});

\item describe which functions on $\sS_\delta$ are regarded as smooth, and which directional derivatives can be taken of them (Section~\ref{sec:smoothdef}).
\end{itemize}
Afterwards, we will formulate the multi-time equations in Section~\ref{sec:multieq}.

\subsubsection{Admissible Configurations}
\label{sec:Sdeltadef}

We now define the set $\sS_\delta$ of $\delta$-spacelike configurations. A similar modification of $\sS$ was used in Sections 1.5.3 and 4 of \cite{pt:2013a} for the purpose of using interaction potentials of range less than $\delta$. 

We write $\cQ^4$ for the set of \emph{all} space-time configurations (4D configurations) $q^4=(x^{4M},y^{4N})$,
\be
\cQ^4 := \bigcup_{N=0}^\infty \sM^{M+N}\,.
\ee
Sometimes we want to talk about spin-configurations in 4D, 
\be
	q^{s4} = (x_1,r_1,...,x_M,r_M,y_1,s_1,...,y_N,s_N) 
	\in \bigcup_{N=0}^\infty \bigl(\sM\times \{1,2,3,4\}\bigr)^{M+N} =: \cQ^{s4}\,.
\ee

As a consequence of the cut-off $\cutoff$ over the distance $\delta>0$, an $x$-particle can instantaneously interact with (i.e., emit or absorb) a $y$-particle at a distance up to $\delta$, and two $x$-particles can instantaneously interact with each other (if one emits a $y$ and the other absorbs it) at a distance up to $2\delta$. Since this instantaneous interaction will generically conflict with the consistency of the multi-time evolution, we make the domain $\sS$ of $\Phi$ slightly smaller and allow configurations involving two particles with such a small distance \emph{only} if the time coordinates of the two particles are equal. So on each sector $N$, we define the set $\sS_\delta^{(N)}$ of $\delta$-spacelike configurations as follows:
\be
\begin{aligned}
	\sS_{\delta}^{(N)} := \Bigl\{(x^{4M},y^{4N})\in \sM^{M+N}~\Big|~
	&\Vert \bx_k - \bx_{k'} \Vert > \vert x^0_k - x^0_{k'} \vert + 2\delta ~~\text{or}~~ x^0_k = x^0_{k'}\\
	&\Vert \bx_k - \by_\ell \Vert > \vert x^0_k - y^0_\ell \vert + \delta ~~~~~\,\text{or}~~ x^0_k = y^0_\ell\\[2mm]
	&\Vert \by_\ell - \by_{\ell'} \Vert > \vert y^0_\ell - y^0_{\ell'} \vert ~~~~~~~~~~\text{or}\;~~ y_\ell = y_{\ell'}\\[1mm]
	&\forall k,k' \in \{1,...,M\}, \; \ell,\ell' \in \{1,...,N\}\Bigr\}.
\end{aligned}
\label{eq:SdeltaN}
\ee
We set
\be\label{Sdeltadef}
\sS_\delta := \bigcup_{N=0}^\infty \sS_\delta^{(N)} \subset \cQ^{4}\,.
\ee
The set $ \sS_\delta$ is illustrated in Figure \ref{fig:Sdelta}. Whenever two particles satisfy the appropriate inequality in \eqref{eq:SdeltaN}, we say that they \emph{keep their safety distance}. The sets $ \sS^{s,(N)}_\delta \subset (\sM\times \{1,2,3,4\})^{M+N}$ and $ \sS^{s}_\delta \subset \cQ^{s4}$ are defined analogously (i.e., by the same conditions as in \eqref{eq:SdeltaN}) for spin-configurations in 4D.

\begin{figure}[hbt]
    \centering
    \begin{tikzpicture}
	\fill[red,opacity = 0.1] (-1,0) -- ++(-2.5,-2.5) -- ++(7,0) -- ++(-2.5,2.5) -- ++(2.5,2.5) -- ++ (-7,0) -- cycle;
	\fill[green!50!black,opacity = 0.4] (-1,0) -- ++(-2.5,-2.5) -- ++ (0,5) -- cycle;
	\fill[green!50!black,opacity = 0.4] (1,0) -- ++(2.5,-2.5) -- ++ (0,5) -- cycle;
	\draw[red,thick,dotted] (-3.5,-2.5) -- ++(2.5,2.5) -- ++ (-2.5,2.5);
	\draw[red,thick,dotted] (3.5,-2.5) -- ++(-2.5,2.5) -- ++ (2.5,2.5);
	\draw[->,thick] (-3.5,0) -- ++(7,0) node[anchor = south] {$ \boldsymbol{x}_j - \boldsymbol{x}_k $};
	\draw[->,thick] (0,-2.5) -- ++(0,5) node[anchor = south] {$ x^0_j - x^0_k $};
	\draw[green!50!black,line width = 3] (-1,0) -- ++(2,0);
	\node at (-0.3,-0.3) {$ 0 $};
	\node[green!50!black] at (2.4,-0.6) {\huge{$ \mathscr{S}_{\delta}^s $}};
\end{tikzpicture}
    \caption{A cross-section of the set $ \sS_{\delta}$ is depicted in green (or dark grey), showing for which values of $x_j-x_k$ the space-time configuration can be in $\sS_\delta$.}
    \label{fig:Sdelta}
\end{figure}
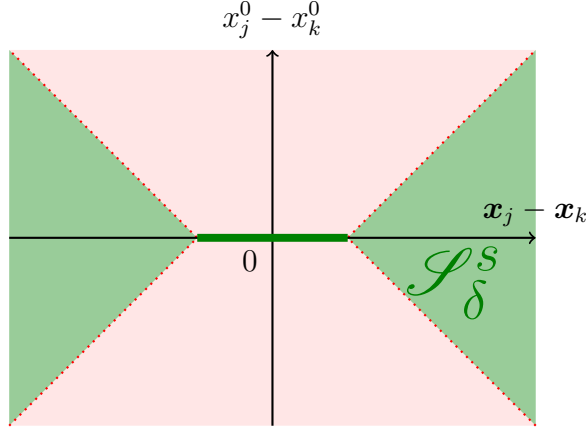

A \emph{multi-time wave function} is a mapping
\be
\begin{aligned}
	\Phi: \sS^s_\delta &\rightarrow \CCC\\
	q^{s4} &\mapsto \Phi(q^{s4}) = \Phi^{(N)}_{\br,\bs}(x_1,...,x_M,y_1,...,y_N).
\end{aligned}
\ee
Note that the definition of $\sS_\delta$ and $\sS_\delta^s$ depends on the frame of reference, as Lorentz invariance is broken by the cut-off $ \cutoff$.

\subsubsection{Admissible Wave Functions}
\label{sec:smoothdef}

Our considerations focus on \emph{smooth} wave functions $\Phi$; for them, derivatives can be understood in the classical sense. However, since $\sS_\delta^{(N)}$ is not an open set in $(\RRR^4)^{M+N}$, we need to explain what we mean by a smooth function on $\sS_\delta^{(N)}$. Put briefly, we regard a function as \emph{smooth} at $q^4\in \sS_\delta$ if it is smooth in the local number of dimensions of $\sS_\delta$ at $q^4$. 

We now approach the detailed definition of smoothness, following \cite{pt:2013a}. To begin with, a function $\Psi: \cQ^{s3}\to\CCC$ will be called \emph{smooth}, $\Psi\in C^\infty(\cQ^{s3})$, if its restriction to each sector is smooth.

For any $ q^4 \in \sS_\delta $, if two particles do not keep their safety distance (say, $ \Vert \bx_j - \bx_k \Vert < 2\delta $), 
then, by the definition \eqref{eq:SdeltaN}, their times must be equal ($x_j^0=x_k^0$) and must remain so in the vicinity of $ q^4 $. Hence, we cannot vary the two time coordinates independently of each other, we can only increase both by the same amount. And hence, we cannot form the partial $x_k^0$ derivative of $\Phi$, we can only form the directional derivative in the direction $e_{x_j}^0+e_{x_k}^0$ in $(\RRR^4)^{M+N}$ (unless further particles do not keep their safety distance from either $x_j$ or $x_k$); here, $e^\mu$ means the standard basis of $\RRR^4$; see Figure~\ref{fig:Sdelta2}. That is why we proceed by grouping particles into families of equal time coordinate.

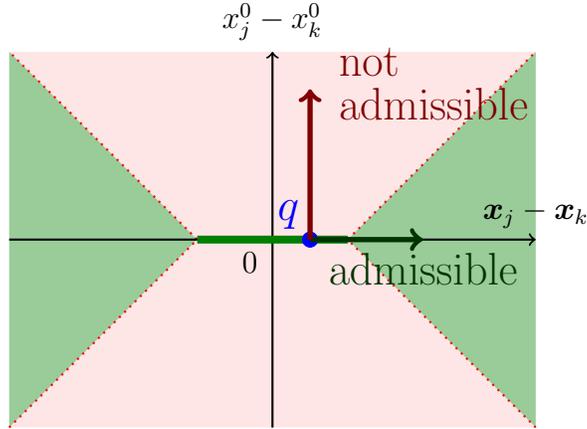
\begin{figure}[hbt]
    \centering
    \begin{tikzpicture}
	\fill[red,opacity = 0.1] (-1,0) -- ++(-2.5,-2.5) -- ++(7,0) -- ++(-2.5,2.5) -- ++(2.5,2.5) -- ++ (-7,0) -- cycle;
	\fill[green!50!black,opacity = 0.4] (-1,0) -- ++(-2.5,-2.5) -- ++ (0,5) -- cycle;
	\fill[green!50!black,opacity = 0.4] (1,0) -- ++(2.5,-2.5) -- ++ (0,5) -- cycle;
	\draw[red,thick,dotted] (-3.5,-2.5) -- ++(2.5,2.5) -- ++ (-2.5,2.5);
	\draw[red,thick,dotted] (3.5,-2.5) -- ++(-2.5,2.5) -- ++ (2.5,2.5);
	\draw[->,thick] (-3.5,0) -- ++(7,0) node[anchor = south] {$ \boldsymbol{x}_j - \boldsymbol{x}_k $};
	\draw[->,thick] (0,-2.5) -- ++(0,5) node[anchor = south] {$ x^0_j - x^0_k $};
	\draw[green!50!black,line width = 3] (-1,0) -- ++(2,0);
	\node at (-0.3,-0.3) {$ 0 $};
	\filldraw[blue] (0.5,0) circle (0.1) node[anchor = south east] {\Large{$q$}};
	\draw[->,green!20!black,line width=2] (0.5,0) -- ++(1.5,0) node[anchor = north] {\Large{admissible}};
	\draw[->,red!50!black,line width=2] (0.5,0) -- ++(0,2) node[anchor = west] {\begin{tabular}{l} \Large{not} \\ \Large{admissible} \end{tabular}};
\end{tikzpicture}
    \caption{At configurations with $x_j^0=x_k^0$ and 
    	$\|\bx_j-\bx_k\|<\delta$, only derivatives in directions 
		changing $ x_j^0 $ and $ x_k^0 $ simultaneously are admissible.}
    \label{fig:Sdelta2}
\end{figure}

To this end, let $\mathscr{L}^{(N)}$ be the $M+N$-element set of particle labels; that is,
\be
\mathscr{L}^{(N)}=\{x_1,\ldots,x_M,y_1,\ldots,y_N\}\,,
\ee
where $x_k$ and $y_\ell$ are now not space-time points but names for the particles. A \emph{partition} $P$ of $\mathscr{L}^{(N)}$ is a set $P = \{ P_1,\dots,P_J \}$ of non-empty subsets $P_j$ of $\mathscr{L}^{(N)}$ (called \emph{families} in the following) with $\cup_{j=1}^J P_j = \mathscr{L}^{(N)}$ and $P_i \cap P_j = \emptyset$ for $i \neq j$. Let $\cP^{(N)}$ denote the set of partitions of $\mathscr{L}^{(N)}$, and let, as a unified notation, $z^\mu_{x_k} = x_k^\mu$ and $z^\mu_{y_\ell}=y_\ell^\mu$, so $z$ could represent any particle ($x$ or $y$). For every $P\in \cP^{(N)}$, we define
\be
\begin{aligned}
	\sS_{\delta}^P := \Bigl\{q^4\in \sM^{M+N}~\Big|~
	&\forall j \in \{1,\dots,J\}~\forall m,n \in P_j : z_m^0 = z_n^0\\
	&\forall j \neq j'~\forall m\in P_j ~\forall n\in P_{j'}: \Vert \bz_m-\bz_n \Vert > \vert z^0_m - z^0_n \vert + d(m,n)\Bigr\},
\end{aligned}
\label{eq:SdeltaP}
\ee
where we abbreviated the ``safety distance'' by
\be\label{safetyd}
d(x_k,x_{k'}):=2\delta, ~~~d(x_k,y_\ell):=\delta, ~~~d(y_\ell, y_{\ell'}):=0.
\ee
In words, $\sS_{\delta}^P$ contains those configurations for which particles in the same family have equal time coordinate and particles in different families keep their safety distance. 
Note that $\sS_\delta^P \subset \sS_\delta^{(N)}$; that for different choices of $P$, the $\sS_\delta^P$ are not necessarily disjoint (because it is allowed that particles in different families have equal time coordinate); and that together, they cover
$\sS_\delta^{(N)}$, i.e., $\cup_{P\in \cP^{(N)}} \sS_\delta^P = \sS_\delta^{(N)}$ (because for given $q^4\in\sS_\delta^{(N)}$, we can group particles into the same family whenever they do not keep their safety distance; that is, $P$ consists of the equivalence classes of the reflexive and transitive hull of the relation ``do not keep their safety distance''; this is the finest partition with $q^4\in\sS_\delta^P$).  

If $q^4\in \sS_\delta^P$, then we also write $t_j$ for the joint time variable of all particles in $P_j$, and $q_j$ for the list of space coordinates of all particles belonging to $P_j$; using that notation, we also write
\be\label{qjdef}
q^4 = (t_1,q_1;\dots;t_J,q_J)\,,
\ee
so that $\sS_\delta^P$ can also be regarded as an open subset of $\RRR^{3M+3N+J}$ (while we continue to use the notation $q^4=(x_1,\ldots,x_M,y_1,\ldots,y_N)$).

\begin{definition}\label{def:smooth}
	A function $\Phi$ is \emph{smooth on $\sS_\delta^P$} if it is smooth as a function of the variables $t_1,q_1,\ldots,t_J,q_J$; $\Phi$ is \emph{smooth on $\sS_\delta^{(N)}$} if it is smooth on $\sS_\delta^P$ for every $P\in\cP^{(N)}$; $\Phi$ is \emph{smooth on $\sS_\delta$} if it is smooth on $\sS_{\delta}^{(N)}$ for every $N\in\NNN_0$. Likewise for functions on spin-configurations: A function $ \Phi$ is \emph{smooth on $\sS_\delta^{s,(N)}$} if for each $(\br,\bs)\in \{1,2,3,4\}^{M+N}$, its spin component $\Phi_{\br,\bs}$ is smooth on $\sS_\delta^{(N)}$; $\Phi$ is \emph{smooth on $\sS_\delta^{s}$} if it is smooth on $\sS_{\delta}^{s,(N)}$ for every $N\in\NNN_0$. The set of all smooth complex-valued functions on $\sS_\delta^s$ is denoted by $C^\infty(\sS_\delta^s)$, the set of smooth functions $\sS_\delta^P\to \CCC^k$ by $C^\infty(\sS_\delta^P, \CCC^k)$.
\end{definition}

Put differently, $\sS_\delta^P$ is a submanifold of $\sM^{M+N}$ of dimension $3M+3N+J$, and $\Phi$ counts as smooth on $\sS_\delta^P$ if it is smooth as a function on this submanifold. It then also follows that at $q^4\in \sS_\delta^P$, derivatives of $\Phi$ in any direction tangent to the submanifold can be taken, as well as higher-order derivatives. That is, $\Phi$ can be differentiated relative to any position coordinate and relative to any family time coordinate $t_j$.

\subsubsection{Multi-Time Equations}
\label{sec:multieq}
\label{sec:IVP}

We now formulate the analog of the multi-time equations \eqref{multi12} with UV cut-off. 

\x{Here, we have to face an issue that we touched upon already in Section~\ref{sec:smoothdef} and Figure~\ref{fig:Sdelta2}: At configurations in which $x_j$ and $x_k$ have equal time coordinate but spatial distance $<\delta$, we cannot vary $x_j^0$ and $x_k^0$ independently without leaving the set $\sS_\delta$ on which $\Phi$ is defined. As a consequence, the partial time derivatives such as $\partial_{x_k^0}$ that were specified in \eqref{multi12} (the original equations using Dirac delta functions) do not make immediate sense with UV cut-off. However, suitable linear combinations of these equations will specify derivatives such as $\partial_{x_j^0}+\partial_{x_k^0}$, which can be understood as the derivative in the direction $e_{x_j}^0+e_{x_k}^0$, which does make sense. Correspondlingly, we take these combinations as the precise evolution equations, as the way how \eqref{multi12} should be interpreted. Again, we group together particles that do not keep their safety distance (and thus have equal time coordinate) into families with common time $t_j$ and formulate equations for the $\partial_{t_j}$-derivative; in particular, at $q^4\in \sS_\delta^P$ with $P=\{P_1,\ldots,P_J\}$, we specify $J$ equations rather than $M+N$.} 
Concretely, the evolution equations can be expressed as follows:
For any $N\in\NNN_0$, $P\in\cP^{(N)}$, and $q^4\in \sS_\delta^P$,
\begin{subequations}
\begin{align}\label{multi34}
&i\frac{\partial\Phi}{\partial t_j}(q^4) 
= \sum_{x_k\in P_j} H^{\free}_{x_k} \Phi(q^4) + \sum_{y_\ell\in P_j} H^{\free}_{y_\ell} \Phi(q^4)\nonumber\\
& + \sqrt{N+1} \sum_{x_k\in P_j} \sum_{r_k',s_{N+1}} g^*_{r'_k r_k s_{N+1}} \int_{B_\delta(\bx_k)} \hspace{-7mm} d^3\tilde\by\:\: \cutoff(\tilde\by-\bx_k) \:\: \Phi^{(N+1)}_{r_k',s_{N+1}}\Bigl(x^{4M},\bigl(y^{4N}, (x_k^0, \tilde\by)\bigr)\Bigr) \nonumber\\
& + \frac{1}{\sqrt{N}} \sum_{x_k\in P_j} \sum_{y_\ell\in P_j} \sum_{r'_k} g_{r_k r'_k s_\ell} \: \cutoff(\by_\ell-\bx_k) \: \Phi^{(N-1)}_{r_k' \widehat{s_\ell}} \bigl(x^{4M}, y^{4N}\setminus y_\ell \bigr)
\\[3mm] \label{HjPdef}
&=:H_j^P \, \Phi(q^4)
\end{align}
\end{subequations}
for all $j\in\{1,\ldots,J\}$. Here, $\sum_{x_k\in P_j}$ means the sum over all $x$-particles in $P_j$, etc.; the index $r_k$ is made explicit at some terms and not others, but occurs at each term. Note that the arguments of $\Phi^{(N+1)}$ and $\Phi^{(N-1)}$ lie in $\sS_\delta$ again. Thus, the right-hand side of \eqref{multi34} defines, for given $N,P,j$, an operator
\be
H_j^P: C^\infty(\sS_\delta^{s}) \to C^\infty \Bigl(\sS_\delta^{P},(\CCC^4)^{\otimes M+N}\Bigr)\,.
\ee
Note that the number of time variables $t_j$ involved in the system of equations \eqref{multi34} is $J=\# P$. Since $q^4$ can lie in $\sS_\delta^P$ for several $P$ (say $P$ and $P'$), one needs to check that the equations \eqref{multi34} from $P$ and $P'$ are compatible with each other. Indeed, if $P'$ is a refinement of $P$, then then the equations from $P$ are linear combinations of the equations from $P'$. Now for any given $q^4$ there is a coarsest partition 
(particles with equal time coordinate belong to the same family) and a finest partition (only particles that do not keep the safety distance belong to the same family), 
so the equations from either $P$ or $P'$ are just linear combinations of the equations obtained from the finest partition. Thus, if the equations \eqref{multi34} hold for the finest partition then they hold for every partition $P$ such that $q^4\in\sS_\delta^P$.

Eq.~\eqref{multi34} is the system of equations that we require to hold for every $N,P,j,q^4\in\sS_\delta^P$. It is the system defining the multi-time evolution of $\Phi$. The \emph{initial value problem} amounts to solving \eqref{multi34} for a given initial datum $\Psi_0\in\sH$,
\be\label{IVP}
\Phi(0,\bx_1, \ldots, 0,\by_N) = \Psi_0(\bx_1,\ldots,\by_N)\,.
\ee
The single-time wave function $\Psi$ is the restriction of $\Phi$ to (the union over $N$ of) $\sS_\delta^P$ with $P=\{\mathscr{L}^{(N)}\}$ (the coarsest of all partitions). It is immediate from \eqref{multi34} that $\Psi$ obeys the Schr\"odinger equation with Hamiltonian \eqref{eq:Hamiltonian} (with strong derivatives).

\section{Results}
\label{sec:results}

The main statement of this paper, Theorem~\ref{thm:1} in Section~\ref{sec:thm} below, asserts that the system of multi-time equations \eqref{multi34} has a unique solution for every sufficiently regular initial datum, and that the solution has the expected properties.

\subsection{Supports}

We also need to talk about the \emph{support} of a wave function; in particular, it will be relevant to exploit the advantages of a \emph{compact support}. 
However, we cannot expect $\Psi$ to have compact support in $\cQ^{s3}$ as that would imply a concentration of $\Psi$ on finitely many sectors. We will thus additionally define the 3-support of $\Psi$, i.e., the region $G \subseteq \RRR^3$ where particles can be encountered at all. \x{We use again the notation $q^3:=(\bx^{3M},\by^{3N})$.}

\begin{definition}
For $\Psi\in\sH$ or $\Psi:\cQ^{s3}\to\CCC$ let $\supp\, \Psi$ denote the essential support of $\Psi$ in $\cQ^3$, i.e., the smallest closed set $G\subseteq \cQ^3$ such that $\Psi=0$ almost everywhere outside $G$. We set
\be
\begin{aligned}
	\supp_{3x}\Psi &:= \overline{\left\{ \bx\in \RRR^3 \;| \;\exists q^3\in \supp\, \Psi\;\exists k: \bx_k = \bx  \right\} }\\[2mm]
	\supp_{3y}\Psi &:= \overline{\left\{ \by\in\RRR^3 \; | \; \exists q^3\in \supp\, \Psi\; \exists \ell:\by_\ell = \by  \right\} } \\[2mm]
	\supp_3\Psi &:= \supp_{3x}\Psi \cup \supp_{3y}\Psi,
\end{aligned}
\label{eq:support}
\ee
where the overbar means the closure in $\RRR^3$. \x{Equivalently, $ \supp_{3x} \Psi $ is the closed union of the projections of $\supp \; \Psi$ to each $\bx_k$-variable in $\RRR^3$ (and similarly for $\supp_{3y} \Psi$).} 
\end{definition}

The following definition will be convenient for expressing propagation locality.

\begin{definition}
We define the \emph{grown set} of $G\subseteq \RRR^3$ as
\be
	\Gr(G,t) := \Bigl\{ \bx' \in \RRR^3 ~\Big|~  \exists \bx \in G ~:~ \Vert \bx - \bx' \Vert \leq t \Bigr\}= \bigcup_{\bx\in G}\overline{B_t}(\bx)
\ee
for any $t\geq 0$, where $\overline{B_r}(\bx)$ means the closed ball of radius $r$ around $\bx$.
\end{definition}

\subsection{Admissible Initial Data} 
\label{sec:Hcinfty}

For every $N,n\in\NNN_0$, let $M(N,n)$ be the set of multi-indices 
\be
\alpha=(\alpha_{x_1^1},\ldots,\alpha_{x_M^3},\alpha_{y_1^1},\ldots,\alpha_{y_N^3}) \in \NNN_0^{\mathscr{L}^{(N)}\times\{1,2,3\}}
\ee
of degree $|\alpha|:=\alpha_{x_1^1}+\ldots+\alpha_{y_N^3}=n$, and let $\partial^\alpha=\partial_{x_1^1}^{\alpha_{x_1^1}} \cdots \partial_{y_N^3}^{\alpha_{y_N^3}}$ be the corresponding derivative.

\begin{definition}\label{def:Hcinfty}
Let $\sH_c^\infty$ be the set of $\Psi\in\sH\subset L^2(\cQ^{s3})$ such that
\begin{enumerate}
\item $\Psi$ possesses a smooth representative (again denoted by $\Psi$), $\Psi \in C^\infty(\cQ^{s3})$
\item $\supp_3 \, \Psi\subseteq \RRR^3$ is compact
\item For every $m,n\in\NNN_0$,
\be\label{Hcinftynorm}
\sum_{N=0}^\infty N^m\hspace{-4mm} \sum_{\alpha\in M(N,n)} \Bigl\|\partial^\alpha \Psi^{(N)}\Bigr\|_N^2<\infty\,,
\ee
where $\|\cdot\|_N$ means the norm of $L^2(\cQ^{s3,(N)})$.
\end{enumerate}
\end{definition}

\bigskip

Note that, since the $r$-th Sobolev norm of $\Psi^{(N)}$ is given by
\be
\|\Psi^{(N)}\|^2_{\HHH^r(\cQ^{s3,(N)})}= \sum_{n=0}^r \sum_{\alpha\in M(N,n)} \Bigl\|\partial^\alpha \Psi^{(N)}\Bigr\|_N^2\,,
\ee
condition 3 is equivalent to the condition that for every $m,r\in\NNN_0$, the sum over $N$ of $N^m$ times the square of the $r$-th Sobolev norm of $\Psi^{(N)}$ is finite. 

Note also that $\sH_c^\infty$ is a dense subspace of $\sH$; for example, it contains the dense subspace of all smooth functions $\cQ^{s3}\to\CCC$ with compact support in $\cQ^{s3}$ (in particular, which vanish outside finitely many sectors) that satisfy the fermionic and bosonic permutation symmetry. An alternative characterization of $\sH_c^\infty$ is given in Lemma~\ref{lem:Hcinfty} in Section~\ref{sec:lemHcinfty}.

\subsection{Theorem About Multi-Time Evolution}
\label{sec:thm}

The following summability property of the solutions $\Phi:\sS^s_\delta\to \CCC$ will be relevant: {\it For every $J\in \NNN$ and every choice of $t_1,\ldots, t_J\in\RRR$,}
\be\label{summability}
\sum_{N=0}^\infty N^m \hspace{-5mm} \int\limits_{\sS^{(N)}(t_1...t_J)} \hspace{-5mm} dq^4 \; \bigl|\Phi(q^4)\bigr|^2 < \infty \quad \text{\it for }m\in\{0,1\}.
\ee
Here,
\be\label{St1tJ}
\sS^{(N)}(t_1...t_J) := \sS_\delta^{(N)}\cap \bigl(\{t_1...t_J\}\times \RRR^3\bigr)^{M+N}
\ee
is the set of all $\delta$-spacelike configurations with $N$ bosons in which only the times $t_1, \ldots,t_J$ occur; it has dimension $3M+3N$. We are now ready to formulate our main result.

\begin{theorem}
For any initial datum $ \Psi_0 \in \sH_c^{\infty} $, the multi-time initial value problem \eqref{multi34}, \eqref{IVP} has a solution $ \Phi \in C^{\infty}(\sS^s_\delta) $ satisfying the summability property \eqref{summability}. The solution is unique among the functions in $C^{\infty}(\sS^s_\delta)$ satisfying \eqref{summability} and has the following further properties:
\begin{enumerate}

\item $ \Phi $ is anti-symmetric under fermion and symmetric under boson permutations.

\item The single-time wave function $ \Psi_t $ recovered from $ \Phi $ as in \eqref{psiphit} evolves unitarily according to $ \Psi_t = e^{-i\bH t} \Psi_0 $.

\item Propagation locality up to $\delta$, i.e., the following bounds on the growth of 3-supports (depicted in Figures \ref{fig:supp} and \ref{fig:supp2}): $ \Phi(q^4) = 0 $ whenever
\be\label{thm1Gr}
\begin{aligned}
\boldsymbol{x}_k &\notin \Gr(\supp_{3x} \Psi_0,\vert x_k^0 \vert)\\
\text{or} \quad
\boldsymbol{y}_\ell &\notin \Gr(\supp_{3y} \Psi_0,\vert y_\ell^0 \vert) \cup \Gr(\supp_{3x} \Psi_0,\vert y_\ell^0 \vert + \delta)
\end{aligned}
\ee
for any $ x_k $ or $ y_\ell $ in $ q^4 $.

\end{enumerate}
\label{thm:1}
\end{theorem}

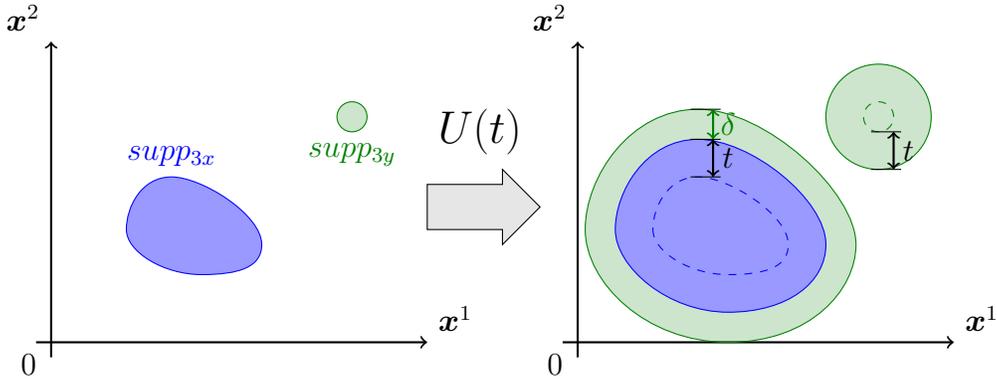
\begin{figure}[hbt]
    \centering
    \def\coordinates at (#1,#2){
	\draw[->,thick] (#1-0.2,#2) -- ++(5.2,0) node[anchor = south west] {$\boldsymbol{x}^1$};
	\draw[->,thick] (#1,#2-0.2) -- ++(0,4.2) node[anchor = south east] {$\boldsymbol{x}^2$};
	\node at (#1-0.3,#2-0.3) {$0$};
}
\begin{tikzpicture}
	\coordinates at (0,0);
	\coordinates at (7,0);
	
	\filldraw[draw = blue, fill = blue!40!white] (1,1.5) .. controls (1,1.2) and (1.5,0.9) .. (2,0.9) .. controls (2.5,0.9) and (2.8,1) .. (2.8,1.3) .. controls (2.8,1.7) and (2,2.2) .. (1.6,2.2) node[blue,anchor = south] {$ supp_{3x} $} .. controls (1.2,2.2) and (1,1.8) .. cycle;
	\filldraw[green!50!black,fill opacity = 0.2] (4,3) circle (0.2);
	\node[green!50!black] at (4,2.5) {$supp_{3y}$};
	\filldraw[fill = gray!20!white] (5,1.5) -- ++(1,0) -- ++(0,-0.2) -- ++(0.5,0.5) -- ++(-0.5,0.5) -- ++(0,-0.2) -- ++(-1,0) -- cycle;
	\node at (5.7,2.8) {\Large{$U(t)$}};
	
	\filldraw[green!50!black,fill opacity = 0.2] (7.1,1.5) .. controls (7.1,0.8) and (7.9,0) .. (9,0) .. controls (10.1,0) and (10.7,0.7) .. (10.7,1.3) .. controls (10.7,2) and (9.5,3.1) .. (8.6,3.1) .. controls (7.7,3.1) and (7.1,2.2) .. cycle;
	\filldraw[draw = blue, fill = blue!40!white] (7.5,1.5) .. controls (7.5,1) and (8.2,0.4) .. (9,0.4) .. controls (9.8,0.4) and (10.3,0.8) .. (10.3,1.3) .. controls (10.3,1.9) and (9.3,2.7) .. (8.6,2.7) .. controls (7.9,2.7) and (7.5,2) .. cycle;
	\draw[blue, dashed] (8,1.5) .. controls (8,1.2) and (8.5,0.9) .. (9,0.9) .. controls (9.5,0.9) and (9.8,1) .. (9.8,1.3) .. controls (9.8,1.7) and (9,2.2) .. (8.6,2.2) .. controls (8.2,2.2) and (8,1.8) .. cycle;
	\filldraw[green!50!black,fill opacity = 0.2] (11,3) circle (0.7);
	\draw[green!50!black,dashed] (11,3) circle (0.2);
	\draw (8.5,2.2) -- ++(0.4,0);
	\draw (8.5,2.7) -- ++(0.4,0);
	\draw (8.5,3.1) -- ++(0.4,0);
	\draw[<->, thick] (8.8,2.2) -- ++(0,0.5);
	\draw[<->, thick,green!50!black] (8.8,2.7) -- ++(0,0.4); 
	\node at (9,2.45) {$t$};
	\node[green!50!black] at (9,2.9) {$\delta$};
	\draw (10.9,2.3) -- ++(0.4,0);
	\draw (10.9,2.8) -- ++(0.4,0);
	\draw[<->, thick] (11.2,2.3) -- ++(0,0.5);
	\node at (11.4,2.55) {$t$};
\end{tikzpicture}
    \caption{Growth of $ \supp_{3x} $ and $ \supp_{3y} $ in physical 3-space (only two dimensions drawn).}
    \label{fig:supp}
\end{figure}

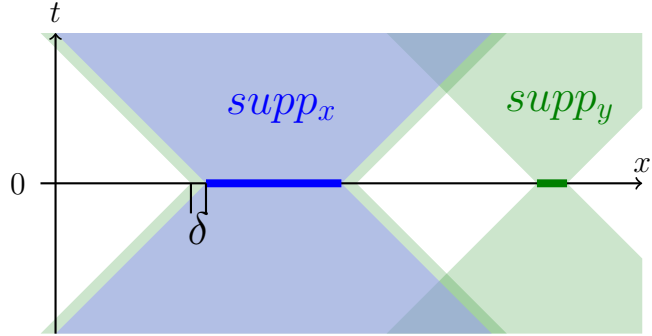
\begin{figure}[hbt]
    \centering
    \begin{tikzpicture}
	\fill[green!50!black, opacity = 0.2] (1.8,0) -- ++(-2,-2) -- ++(6.2,0) -- ++(-2,2) -- ++(2,2) -- ++(-6.2,0) -- cycle;
	\fill[green!50!black, opacity = 0.2] (6.4,0) -- ++(-2,-2) -- ++(3.4,0) -- ++(0,1) -- ++(-1,1) -- ++(1,1) -- ++(0,1) -- ++ (-3.4,0)-- cycle;
	\fill[blue!40!white, opacity = 0.5] (2,0) -- ++(-2,-2) -- ++(5.8,0) -- ++(-2,2) -- ++(2,2) -- ++(-5.8,0) -- cycle;
	\draw[->,thick] (-0.2,0) -- ++(8,0) node[anchor = south] {$ x $};
	\draw[->,thick] (0,-2) -- ++(0,4) node[anchor = south] {$ t $};
	\node at (-0.5,0) {$0$};
	\draw[blue,line width = 3] (2,0) -- ++(1.8,0);
	\draw[green!50!black,line width = 3] (6.4,0) -- ++(0.4,0);
	\node[blue] at (3,1) {\Large{$supp_x$}};
	\node[green!50!black] at (6.7,1) {\Large{$supp_y$}};
	\draw[thick] (1.8,0) -- ++(0,-0.4);
	\draw[thick] (2,0) -- ++(0,-0.4);
	\node at (1.9,-0.6) {\Large{$\delta$}};
\end{tikzpicture}
    \caption{Growth of $ \supp_{3x} $ and $ \supp_{3y} $ in a Minkowski diagram.}
    \label{fig:supp2}
\end{figure}

\noindent{\bf Remarks.}
\begin{enumerate}
\setcounter{enumi}{\theremarks}
\item {\it 4-supports.} The condition \eqref{thm1Gr} for configurations outside the support of $\Phi$ can equivalently be re-expressed as
\be
\begin{aligned}
x_k &\notin \mathrm{Infl}\bigl( \{0\} \times \supp_{3x} \Psi_0\bigr)\\
\text{or} \quad y_\ell &\notin \mathrm{Infl} \Bigl(\{0\} \times \Bigl[\supp_{3y} \Psi_0 \cup \Gr\bigl( \supp_{3x} \Psi_0,\delta \bigr) \Bigr]\Bigr)\,,
\end{aligned}
\ee
where $\mathrm{Infl}(A)$ means the domain of influence of the space-time set $A\subseteq \sM$,
\be
\mathrm{Infl}(A) := \mathrm{future}(A) \cup \mathrm{past}(A)\,,
\ee
where future$(A)$ is the union over $x\in A$ of the closed future light cone of $x$, and correspondingly past$(A)$.
We can also express the situation in terms of the 4-support of $\Phi$:
\be
\begin{aligned}
	\supp_{4x}\Phi &:= \overline{\left\{ x\in \sM \;| \;\exists q^4\in \supp\, \Phi\;\exists k: x_k = x  \right\} }\\[2mm]
	\supp_{4y}\Phi &:= \overline{\left\{ y\in\sM \; | \; \exists q^4\in \supp\, \Phi\; \exists \ell:y_\ell = y  \right\} } \\[2mm]
	\supp_4\Phi &:= \supp_{4x}\Phi \cup \supp_{4y}\Phi,
\end{aligned}
\label{supp4def}
\ee
Then
\be
\begin{aligned}
\supp_{4x}\Phi &\subseteq \mathrm{Infl}\bigl( \{0\} \times \supp_{3x} \Psi_0\bigr)\\
\supp_{4y}\Phi &\subseteq \mathrm{Infl}\Bigl(\{0\} \times \Bigl[\supp_{3y} \Psi_0 \cup \Gr\bigl( \supp_{3x} \Psi_0,\delta \bigr) \Bigr]\Bigr)\,.
\end{aligned}
\label{4supp}
\ee
These sets are shown in Figure~\ref{fig:supp2}.

\item \x{{\it Asymmetry between $x$ and $y$.} Patently, \eqref{thm1Gr} is not symmetric in $x$ and $y$. The reason is that the Hamiltonian $ \bH $ makes all $x$-particles create and annihilate $y$-particles, so even if at $ t=0 $ there are no $y$-particles near an $x$-particle, this situation will change at $t>0$ and $\supp_{3y} \Psi$ will immediately be extended by $ \Gr(\supp_{3x} \Psi_0, \delta) $. The same holds for $t<0$. However, if there is a separated $y$-particle, $\bH$ will not cause a creation of $x$-particles near it. So $ \supp_{3x} \Psi$ does not grow by an additional amount depending on $\supp_{3y} \Psi$.}

\item {\it Uniqueness.} The summability property \eqref{summability} is relevant to the \emph{uniqueness} of the solution; in $C^\infty(\sS_\delta^s)$, the solution would not be unique, essentially because a further \x{contribution to the wave function could come ``from infinity on the $N$ axis.''} This effect has nothing to do with the multiple time variables, it arises from considering smooth functions that solve the equations in the classical sense and is, in the usual one-time formulation, taken care of by choosing as the Hamiltonian a \emph{self-adjoint extension} of the differential expression for $\bH$. This point is elucidated further in Section~\ref{sec:smooth1time}.

\item {\it $L^2$ solutions.} Usually, when considering solutions of the 1-time Schr\"odinger equation, one does not require classical differentiability. One allows weak (distributional) derivatives and thus wave functions from Sobolev spaces, and even more, one allows arbitrary $L^2$ functions as initial data, as the unitary time evolution operator $\exp(-i\bH t)$ is defined on the entire Hilbert space. In the same way, we can define a notion of $L^2$ solution of our multi-time equation \eqref{multi34}, and then allow for arbitrary $\Psi_0\in\sH$. In fact, our construction of the solution $\Phi$ works in this way, and we have to invest further work for proving that for nice initial data $\Psi_0\in\sH^\infty_c$, the function $\Phi$ we construct is a classical solution.

\item\label{rem:Shat} {\it Domain of consistency.} As it happens, the simple model we are considering is consistent on an even larger set than $\sS_\delta$. This fact is presumably a curious artifact of our simple model that will not extend to more realistic models. This larger set
$\widehat{\sS_\delta} \supset \sS_\delta$ is defined by the same condition as for $\sS_\delta$ between $x$-particles but no restriction on the $y$-particles, not even being spacelike separated, neither from $x$-particles nor among themselves. That is,
\begin{align}
	\widehat{\sS_{\delta}}^{(N)} := \Bigl\{(x^{4M},y^{4N})\in \sM^{M+N}~\Big|~
	&\Vert \bx_k - \bx_{k'} \Vert > \vert x^0_k - x^0_{k'} \vert + 2\delta ~~\text{or}~~ x^0_k = x^0_{k'}\nonumber\\
	&\forall k,k' \in \{1,...,M\}\Bigr\}\label{eq:hatSdeltaN}
\end{align}
and
\be\label{Shatdeltadef}
\widehat{\sS_\delta} := \bigcup_{N=0}^\infty \widehat{\sS_\delta}^{(N)}\,.
\ee
\x{The physical reason is that the bosons do not interact with each other, nor with the fermions, except for being emitted and absorbed by the fermions. It is well known \cite{LPT:2020} that for non-interacting particles, multi-time equations are consistent on \emph{all} space-time configurations ($\in \cQ^4$), not only the spacelike ones.}
Correspondingly for $\widehat{\sS_\delta^s}^{(N)}$ and $\widehat{\sS_\delta^s}$; in the same way as for $\sS_\delta^s$, one defines $C^\infty(\widehat{\sS_\delta^s})$. 
In Section~\ref{sec:Shat}, we show that
\emph{for any initial datum $\Psi_0\in\sH_c^\infty$, the system of multi-time equations \eqref{multi56} below with initial values $\Psi_0$ as in \eqref{IVP} has a unique solution $\Phi\in C^\infty(\widehat{\sS_\delta^s})$ satisfying the summability condition} \eqref{summability}.

The multi-time equations are essentially given by \eqref{multi12} with the Green function $G_0$ replaced by a family (indexed by $r,r'$) of cut-off Green functions $\G_{rr'}$, each in $C^\infty(\RRR^4,\CCC^4)$, viz., the solution of
\be\label{Gdeltadef1}
i\frac{\partial \G}{\partial t} = H^{\free}_y \, \G
\ee
with initial condition
\be\label{Gdeltadef2}
\G_{rr's}(0,\by) = g_{rr's}\, \cutoff(\by)\,. 
\ee
More precisely, since, as discussed in Section~\ref{sec:multieq} for $\sS_\delta$, the derivative $\partial_{x^0_k}$ is not admissible within $\widehat{\sS_\delta}$ when two $x$-particles are closer than $2\delta$, we need to group together $x$-particles that do not keep their safety distance into a family with a common time. That is, for a given $q^4\in\widehat{\sS_\delta}$ let $\{P_1,\ldots,P_J\}$ be the finest partition of $\{x_1,\ldots,x_M\}$ such that particles that do not keep their safety distance belong to the same partition set. By the definition \eqref{eq:hatSdeltaN} of $\widehat{\sS_\delta}$, all particles in $P_j$ have the same time coordinate, now called $t_j$. The multi-time equations \x{\eqref{multi34} can then be written as}
\begin{subequations}\label{multi56}
\begin{align}\label{multi5}
i \frac{\partial\Phi}{\partial t_j}(x^{4M}, y^{4N}) &= \sum_{x_k\in P_j}H^{\free}_{x_k} \Phi(x^{4M}, y^{4N}) \nonumber\\
& \hspace{-24mm} + \sqrt{N+1} \sum_{x_k\in P_j} \sum_{r_k',s_{N+1}} g^*_{r'_k r_k s_{N+1}} \int_{B_\delta(\bx_k)} \hspace{-7mm} d^3\tilde\by\:\: \cutoff(\tilde\by-\bx_k) \:\: \Phi^{(N+1)}_{r_k',s_{N+1}}\Bigl(x^{4M},\bigl(y^{4N}, (x_k^0, \tilde\by)\bigr)\Bigr) \nonumber\\
& + \frac{1}{\sqrt{N}} \sum_{x_k\in P_j}\sum_{\ell=1}^N\sum_{r'_k} \G_{r_k r'_k s_\ell}(y_\ell - x_k) \, \Phi_{r'_k\widehat{s_\ell}}^{(N-1)}\bigl(x^{4M}, y^{4N} \backslash y_\ell\bigr)  \\[3mm] 
\label{multi6}
i \frac{\partial\Phi}{\partial y_\ell^0}(x^{4M}, y^{4N}) &= H^{\free}_{y_\ell}\Phi(x^{4M}, y^{4N}).
\end{align}
\end{subequations}
We note that every solution $\Phi\in C^\infty(\widehat{\sS^s_\delta})$ of \eqref{multi56}, when restricted to $\sS^s_\delta$, also solves \eqref{multi34}. For verifying this, the only real issue is that \eqref{multi5} involves a sum over \emph{all} $y$-particles ($\ell$ runs from 1 to $N$), whereas \eqref{multi34} involves a sum over only the $y$-particles in the family $P_j$ of $x_k$. However, this difference has no consequences in $\sS_\delta$, as there every $y$-particle outside the family of $x_k$ is $\delta$-spacelike from $x_k$, while $G$ vanishes on all points $\delta$-spacelike from $x_k$.
\end{enumerate}
\setcounter{remarks}{\theenumi}

\subsection{Lemmas About Single-Time Evolution} \label{sec:lemmas}

A few statements about the single-time evolution that we need as tools for proving Theorem~\ref{thm:1} may be worth mentioning in their own right.

\begin{lemma}[Self-adjointness]\label{lem:selfadjoint}
The single-time Hamiltonian $ \bH $ defined in \eqref{eq:Hamiltonian} is essentially self-adjoint on $ \sH_c^{\infty} \subset \sH $.
\end{lemma}

By Lemma \ref{lem:selfadjoint}, $ \bH $ has a unique self-adjoint extension, which we will also denote by $ \bH $ and which provides the unitary single-time evolution operator
\be
\bU(t) := e^{-i\bH t}: \sH\to \sH 
\ee
for all $t\in\RRR$. 
Lemma~\ref{lem:supp} below yields information about the growth of 3-supports, and Lemma~\ref{lem:diffable} the regularity of $\Psi_t$.

\begin{lemma}[Support growth]\label{lem:supp}
Let $ \Psi_0 \in \sH$ and $ \Psi_t = \bU(t) \Psi_0 $ for all $ t \in \RRR $. Supports will grow under time evolution at most as follows:
\be
\begin{aligned}
	&\supp_{3x}\Psi_{t} \subseteq \Gr(\supp_{3x}\Psi_{0},\vert t \vert) \quad &&\forall t \in \RRR\\
	&\supp_{3y}\Psi_{t} \subseteq \Gr(\supp_{3y}\Psi_{0},\vert t \vert) \cup \Gr(\supp_{3x}\Psi_{0},\vert t \vert + \delta) \quad &&\forall t \in \RRR,
\end{aligned}
\label{eq:growxy}
\ee
in particular
\be
	\supp_3\Psi_{t} \subseteq \Gr(\supp_{3y}\Psi_{0},\vert t \vert) \cup \Gr(\supp_{3x}\Psi_{0},\vert t \vert + \delta) \quad \forall t \in \RRR.
\label{eq:grow}
\ee
\end{lemma}

\noindent{\bf Remark.}
\begin{enumerate}
\setcounter{enumi}{\theremarks}
\item\label{rem:W} {\it Local evolution operator.} As a consequence, the wave function $\Psi_t$ on configurations with all particles in a region $G\subset \RRR^3$ depends only on the initial data in $\Gr(G,|t|+\delta)$. Thus, the single-time evolution defines an operator
\be\label{Wdef}
W_t: \sH\bigl(\Gr(G,|t|+\delta)\bigr) \to \sH(G)\,, \quad W_t(\Psi_0)=\Psi_t\Big|_G\,,
\ee
where $\sH(G)$ means the Hilbert space associated with the region $G$, i.e., the subspace of $L^2(\cQ^{s3}(G))$ with the appropriate permutation symmetry, where
\be
\cQ^{s3}(G) = (G\times \{1,2,3,4\})^M \times \bigcup_{N=0}^\infty (G\times \{1,2,3,4\})^N
\ee
is the set of spin-configurations concentrated in $G$. 
\end{enumerate}
\setcounter{remarks}{\theenumi}

\begin{lemma}[Invariance of $\sH_c^\infty$]\label{lem:diffable}
For all $t\in\RRR$, ~~$\bU(t) \, \sH_c^\infty \subseteq \sH_c^\infty$.
\end{lemma}

The following lemma ensures that the solution of the single-time evolution is smooth also in time; in fact, $\Psi(t,q^3)$ will be smooth as a function of $t$ and $q^3$. 

\begin{lemma}[Smoothness in $q$ and $t$]\label{lem:smootht}
For every $\Psi_0\in \sH_c^\infty$, one can choose a representative of each $\Psi_t$ in such a way that $\Psi_t(q^{s3})$ is a smooth function of $t$ and $q^{s3}$.
\end{lemma}

If, moreover, $\Psi_0$ depends smoothly on a parameter $\lambda$, then the solution $\Psi(t,q^3)$ will also depend smoothly on $\lambda$.
To formulate this statement precisely, let $d\in\NNN$, and let 
$\sH_d$ denote the subspace of $L^2(\RRR^d\times \cQ^{s3})$ of functions $\Psi(\lambda,\bx^{3M},\br,\by^{3N},\bs)$ that are anti-symmetric against permutations of $(\bx^{3M},\br)$ and symmetric against permutations of $(\by^{3N},\bs)$. Let $\sH_{cd}^\infty$ denote the subspace of $\sH_d$ of functions such that
\begin{enumerate}
\item $\Psi\in C^\infty(\RRR^d\times \cQ^{s3})$
\item $\supp_{3x}\Psi$, $\supp_{3y}\Psi$, and $\supp_\lambda \Psi \subseteq \RRR^d$
are compact
\item For every $m,n\in\NNN_0$,
\be\label{HcLambdainftynorm}
\sum_{N=0}^\infty N^m \!\! \sum_{\alpha\in M(N,n,d)}\Bigl\| \partial^\alpha \Psi^{(N)}\Bigr\|_N^2 < \infty\,,
\ee
where $M(N,n,d)$ is the set of multi-indices of degree $n$ for the variables $\lambda_1,\ldots,\lambda_d$, $x_1^1,\ldots,x_M^3$, $y_1^1,\ldots,y_N^3$.
\end{enumerate}

\begin{lemma}[Smoothness with parameter]\label{lem:smoothsingletime}
Suppose $\Psi_0\in \sH_{cd}^\infty$. Then $\Psi_0(\lambda,\cdot)\in \sH_c^\infty$ for almost every $\lambda\in \RRR^d$. For those $\lambda\in\RRR^d$ and every $t\in\RRR$, define
\be\label{Psitlambdadef}
\Psi(t,\lambda,\cdot) := \bU(t) \Psi_0(\lambda,\cdot)\,.
\ee
Then $\Psi$ is smooth in $t,\lambda,q^3$, $\Psi\in C^\infty(\RRR\times \RRR^d \times \cQ^{s3})$. 
\end{lemma}

\subsection{Strong Single-Time Solutions}
\label{sec:smooth1time}

Since we have chosen to consider \emph{strong} solutions (smooth functions with classical derivatives) of our multi-time PDEs, and since we will apply single-time considerations to each time variable, we also need to consider \emph{strong} solutions of our single-time evolution equation, which reads explicitly, with $q^{s3} = (\br,\bx^{3M},\bs,\by^{3N})$,
\begin{align}
&i\frac{\partial\Psi}{\partial t}(t,q^{s3}) 
= \sum_{k=1}^M H^{\free}_{x_k} \Psi(t,q^{s3}) + \sum_{\ell=1}^N H^{\free}_{y_\ell} \Psi(t,q^{s3})\nonumber\\
& + \sqrt{N+1} \sum_{k=1}^M \sum_{r_k',s_{N+1}} g^*_{r'_k r_k s_{N+1}} \int_{B_\delta(\bx_k)} \hspace{-7mm} d^3\tilde\by\:\: \cutoff(\tilde\by-\bx_k) \:\: \Psi^{(N+1)}_{r_k',s_{N+1}}\Bigl(t,\bx^{3M},\bigl(\by^{3N}, \tilde\by \bigr)\Bigr) \nonumber\\
& + \frac{1}{\sqrt{N}} \sum_{k=1}^M \sum_{\ell=1}^N \sum_{r'_k, s_\ell} g_{r_k r'_k s_\ell} \: \cutoff(\by_\ell-\bx_k) \: \Psi^{(N-1)}_{r_k' \widehat{s_\ell}} \bigl(t,\bx^{3M}, \by^{3N}\setminus \by_\ell \bigr) \,.
\label{1time}
\end{align}
While Lemmas~\ref{lem:selfadjoint}, \ref{lem:diffable}, and \ref{lem:smootht} together guarantee that the function $\Psi$ obtained by applying the exponential $e^{-i\bH t}$ of the self-adjoint operator $\bH$ to $\Psi_0$ is a strong solution of \eqref{1time},\footnote{That is because $\Psi_t\in \sH_c^\infty$ by Lemma~\ref{lem:diffable} and $\Psi_t\in\mathrm{dom}(\bH)$ by Lemma~\ref{lem:selfadjoint}, so $t\mapsto \Psi_t$ is a differentiable curve in $\sH$, whose derivative agrees on the one hand with the strong $t$-derivative of the smooth function $\Psi$ provided by Lemma~\ref{lem:smootht} and on the other hand with the right-hand side of \eqref{1time} by the definition of $\bH$ on $\sH_c^\infty$.} we also need to verify that \eqref{1time} has no further classical solutions with initial value $\Psi_0$; this is done in Lemma~\ref{lem:1timeunique} below. We will also use that our bounds on support growth follow directly from \eqref{1time}, a fact established in Lemma~\ref{lem:smoothsupp}.

\begin{lemma}[Support growth of strong single-time solutions]\label{lem:smoothsupp}
If $\Psi\in C^\infty(\RRR_t\times \cQ^{s3})$ solves the 1-time equation \eqref{1time}, and if $\Psi_t(\cdot):=\Psi(t,\cdot)$ satisfies $\|\Psi_t\|<\infty$ and $\|N^{1/2}\Psi_t\|<\infty$ with $N$ the $y$-number operator for every $t\in\RRR$, then the support of $\Psi_t$ grows at most according to \eqref{eq:growxy} (and thus to \eqref{eq:grow}).
\end{lemma}

The differences between Lemma~\ref{lem:supp} and Lemma~\ref{lem:smoothsupp} arise from their different methods of proof: Lemma~\ref{lem:supp} uses the Trotter product formula, Lemma~\ref{lem:smoothsupp} integration of the current; Lemma~\ref{lem:supp} concerns $L^2$ initial data and a time evolution defined by a self-adjoint operator $\bH$, Lemma~\ref{lem:smoothsupp} concerns a $C^\infty$ solution of \eqref{1time}.

\begin{lemma}[Uniqueness of strong single-time solutions]\label{lem:1timeunique}
For every $\Psi_0\in C^\infty(\cQ^{s3})$ there is at most one solution $\Psi\in C^\infty(\RRR \times \cQ^{s3})$ to \eqref{1time} satisfying $\|\Psi_t\|<\infty$ and $\|N^{1/2}\Psi_t\|<\infty$ for every $t\in\RRR$ with initial data $\Psi_0$.
\end{lemma}

Without demanding the summability properties $\|\Psi_t\|<\infty$ and $\|N^{1/2}\Psi_t\|<\infty$, the solution would not be unique in $C^\infty(\RRR_t\times \cQ^{s3})$, with further contributions coming from $N\to\infty$. This point can be illustrated by the following minimalistic variant of \eqref{1time}, in which spin is dropped and space $\RRR^3$ is replaced by a single point, so that the configuration is fully described by the number of $y$-particles and no space variables appear any more, $\Psi=(\Psi^{(0)},\Psi^{(1)},\ldots)$, while $\Psi^{(N)}(t)\in\CCC$ is governed by the Schr\"odinger equation
\be\label{1point}
i\frac{d \Psi^{(N)}}{dt}= g^*\, \Psi^{(N+1)}(t) + g \, \Psi^{(N-1)}(t)\,.
\ee
Then an example of a smooth non-zero solution with initial condition $\Psi_0=(0,0,\ldots)$ is provided by
\be
\Psi^{(N)}(t) = \begin{cases}
P_N(1/t) \, e^{-1/t}&t>0\\
0&t=0\\
\tilde P_N(1/t) \, e^{1/t}&t<0
\end{cases}
\ee
with polynomials $P_N,\tilde P_N$ defined recursively by
\begin{align}
P_{N+1}(\nu) &= \hspace{3mm} \tfrac{i}{g^*} \nu^2 P_N(\nu) -\tfrac{i}{g^*} \nu^2 P'_N(\nu)-\tfrac{g}{g^*}P_{N-1}(\nu)\\
\tilde P_{N+1}(\nu) &= -\tfrac{i}{g^*} \nu^2 \tilde P_N(\nu) -\tfrac{i}{g^*} \nu^2 \tilde P'_N(\nu)-\tfrac{g}{g^*}\tilde P_{N-1}(\nu)
\end{align}
(the prime means the $\nu$-derivative) with starting values $P_{-1}=0, P_0=1, \tilde P_{-1}=0, \tilde P_0=1$.

\section{Proofs: Single-Time Evolution}	
\label{sec:singletime}

In the following, we will repeatedly use that
\be
 \left\| \sum_{i=1}^n a_i\right\|^2 \leq \left( \sum_{i = 1}^n \|a_i\| \right)^2 \le n \sum_{i = 1}^n \|a_i\|^2
\label{eq:ineq1b}
\ee
for elements $a_i$ of a normed space.

\subsection{Essential Self-Adjointness}	
\label{subsec:selfadjoint}

\begin{proof}[Proof of Lemma~\ref{lem:selfadjoint}]
We first verify that $\bH$ is well defined as an operator $\sH_c^\infty\to\sH$, using standard arguments. Suppose $\Psi\in\sH_c^\infty$. It is clear that $\bH_{x_k}^\free\Psi^{(N)}\in L^2(\cQ^{s3,(N)})$ because $\Psi^{(N)}$ is smooth and has compact support; $\sum_k \bH_{x_k}^\free\Psi^{(N)}$ is anti-symmetric against $x$-permutations; and $\sum_k \bH_{x_k}^\free\Psi\in\sH$ because
\begin{align}
\sum_{N=0}^\infty \biggl\|\sum_{k=1}^M \bH_{x_k}^\free\Psi^{(N)} \biggr\|_N^2
&\leq\sum_{N=0}^\infty \biggl(\Bigl\|\sum_{k=1}^M (-i)\balpha_k\cdot\bnabla_k\Psi^{(N)} \Bigr\|_N+ \Bigl\|m_x\sum_{k=1}^M\beta_k\Psi^{(N)}\Bigr\|_N\biggr)^2\\
&\overset{\eqref{eq:ineq1b}}{\leq}\sum_{N=0}^\infty 2\biggl(\Bigl\|\sum_{k=1}^M (-i)\balpha_k\cdot\bnabla_k\Psi^{(N)} \Bigr\|_N^2+ \Bigl\|m_x\sum_{k=1}^M\beta_k\Psi^{(N)}\Bigr\|_N^2\biggr)\\
&\overset{\eqref{eq:ineq1b}}{\leq}\sum_{N=0}^\infty 2M\sum_{k=1}^M \Bigl\|\balpha_k\cdot\bnabla_k\Psi^{(N)} \Bigr\|_N^2+ 2Mm_x^2\sum_{N=0}^\infty \sum_{k=1}^M\bigl\|\beta_k\Psi^{(N)}\bigr\|_N^2\\
&\overset{\eqref{eq:ineq1b}}{\leq}6M\sum_{N=0}^\infty \sum_{k=1}^M \sum_{a=1}^3\bigl\|\partial_{x_k^a}\Psi^{(N)} \bigr\|_N^2+ 2Mm_x^2\sum_{N=0}^\infty \sum_{k=1}^M\bigl\|\Psi^{(N)}\bigr\|_N^2
\end{align}
because the $\alpha$ and $\beta$ matrices are unitary. The first term is finite by \eqref{Hcinftynorm} with $m=0$ and $n=1$, the second because $\Psi\in\sH$.

In the same way, one can show that $d\Gamma_y(\bH_y^\free)\Psi\in\sH$, except that one has to apply \eqref{Hcinftynorm} once with $m=1,n=1$ and once with $m=2,n=0$.

Now we show that $\bH^\inter \Psi\in \sH$. Since $\cutoff$ is smooth with compact support, $\bH^\inter \Psi$ is smooth and has compact 3-support, so each sector of it is square-integrable; it is clear that $\bH^\inter\Psi$ has the right permutation symmetry; it remains to show that $\sum_N \| \tilde\Psi^{(N)} \|_N^2<\infty$ for $\tilde\Psi = \bH^\inter\Psi$. We show this for $\tilde\Psi=\ba(\bx_k^{op})\Psi$ and $\tilde\Psi=\ba^\dagger(\bx_k^{op})\Psi$.

In cases where the operators
\be
	\ba(f) := \int d^3 \bx\, f(\bx)\, \ba(\bx)  \quad \text{or} \quad \ba^{\dagger}(f) := \int d^3 \bx\, f(\bx) \, \ba^{\dagger}(\bx)
\ee
are annihilating or creating a fixed function $ f \in L^2(\RRR^3) $, we may bound them against the $y$-number operator
\be
(\bN\Psi)^{(N)} = N \Psi^{(N)}
\ee
as follows (e.g., \cite{Nel64}):
\be
	\Vert \ba(f) \Psi \Vert \le \Vert f \Vert \, \Vert \bN^{1/2} \Psi \Vert \quad \text{and} \quad \Vert \ba^{\dagger}(f) \Psi \Vert \le \Vert f \Vert \, \Vert (\bN + \bone)^{1/2} \Psi \Vert\,.
\label{eq:nelsonbound}
\ee
In our case, the function $ f(\by) = \cutoff(\by - \bx_k) $ depends on the position $ \bx_k $ of the fermion, so $ f \in L^2 $, but it is not fixed. However, we may write the Hilbert space $ \sH $ as a direct integral over all fermion configurations
\be
	\sH = \int_{\RRR^{3M}}^{\oplus} d\bx^{3M} \sH_y
\ee
and decompose $ \ba(\bx^{op}_k) $ and $ \ba^{\dagger}(\bx^{op}_k) $ into fibers:
\be
	\ba(\bx^{op}_k) = \int_{\RRR^{3M}}^{\oplus} d\bx^{3M} \ba(\cutoff(\cdot - \bx_k)), \qquad \ba^{\dagger}(\bx^{op}_k) = \int_{\RRR^{3M}}^{\oplus} d\bx^{3M} \ba^{\dagger}(\cutoff(\cdot - \bx_k)).
\label{eq:fiber}
\ee
Applying the bound \eqref{eq:nelsonbound} on each fiber, we get the estimates
\be\label{eq:nelsonboundcutoff}
	\Vert \ba(\bx^{op}) \Psi \Vert \le \Vert \cutoff \Vert \, \Vert \bN^{1/2} \Psi \Vert \quad \text{and} \quad \Vert \ba^{\dagger}(\bx^{op}) \Psi \Vert \le \Vert \cutoff \Vert \, \Vert (\bN + \bone)^{1/2} \Psi \Vert\,.
\ee
Since $(N+1)^{1/2}\leq N+1$ for every $N\in\NNN_0$, we have that $(\bN+\bone)^{1/2} \leq \bN+\bone$, so
\be
\sum_{N=0}^\infty \Bigl\| (\bH_{x_k}^\inter \Psi)^{(N)} \Bigr\|_N^2 \leq c_1 \sum_{N=0}^\infty N \|\Psi^{(N)}\|_N^2 + c_2 \sum_{N=0}^\infty \|\Psi^{(N)}\|_N^2\,.
\ee
The first term is finite by \eqref{Hcinftynorm} with $m=1,n=0$, the second by $\Psi\in\sH$. This completes the proof that $\bH\Psi$ is well defined and $\in\sH$ for $\Psi\in\sH_c^\infty$.

\bigskip

Now we turn to essential self-adjointness. Following \cite{ND:2019} and \cite{arai}, we use the following ``commutator theorem'' \cite[Theorem X.37]{reedsimon2}, \cite[Theorem 1]{farislavine}:
{\it Let $\comp$ be a self-adjoint operator in $\sH$ with $\comp\geq \bone$, and let $\bH$ be a symmetric operator with domain $D\subseteq \sH$ which is a core for $\comp$. Suppose that for some $c,d\in\RRR$ and all $\Psi\in D$,
\be
	\Vert \bH \Psi \Vert < c \Vert \comp \Psi \Vert\,,
\label{eq:commcond1}
\ee
\be
	\bigl\vert \langle \bH \Psi, \comp \Psi \rangle - \langle \comp \Psi, \bH \Psi \rangle \bigr\vert \le d \Vert \comp^{1/2} \Psi \Vert^2\,.
\label{eq:commcond2}
\ee
Then $\bH$ is essentially self-adjoint on $D$.}

In our case, $D=\sH_c^\infty$; it is straightforward to verify that $\bH$ is symmetric on $\sH_c^\infty$. We choose as the comparison operator
\be\label{compdef}
	\comp := \sum_{k = 1}^M (\bone- \bDelta_{x_k}) + d\Gamma_y(\bone-\bDelta_y) \,.
\ee

We first prove that $ \comp $ is essentially self-adjoint  on $ \sH_c^{\infty} $, using standard arguments: To begin with, it is clearly defined on $\sH_c^\infty$. The negative Laplacian $-\bDelta$ is known to be positive and essentially self-adjoint on $C_c^\infty(\RRR^3)$ (compactly supported smooth functions). Adding a bounded self-adjoint operator such as $\bone$ has no effect on essential self-adjointness. It is known further \cite[Thm.~VIII.33 and Ex.~2 p.~302]{reedsimon1} that $\comp$ (and thus also $\comp$) is essentially self-adjoint on $C_c^\infty(\cQ^{s3})$ (smooth functions that vanish on all but finitely many sectors and have compact support in each sector). Since the latter space is contained in $\sH_c^\infty$ and dense in $\sH$, $\comp$ is also essentially self-adjoint on $\sH_c^\infty$. Positivity follows from positivity on each sector (using that $\comp$ maps each sector to itself); in fact, $\comp\geq \bone$ because $M\geq 1$.

It remains to prove the two inequalities \eqref{eq:commcond1} and \eqref{eq:commcond2}. For \eqref{eq:commcond1}, we need to bound $ \bH_{x_k}^{\free} $, $ \bH_{y_\ell}^{\free} $, and $ \bH_{x_k}^{\inter} $ against $ \comp $. The subsequent paragraphs deal with these terms separately.

\underline{$ \bH_{x_k}^{\free} $:} In the following, we focus on a single sector $N$ and simply write $\Psi^{(N)}=\Psi$ and $\|\cdot\|_N=\|\cdot\|$. 
The free Dirac operator is rather easy to bound against the Laplacian: First, 
\be
	\Vert \bH_{x_k}^{\free} \Psi \Vert
	= \left\Vert  \left( - i \balpha_k \cdot \bnabla_k + m_x \beta_k \right) \Psi \right\Vert
	\leq \left\Vert - i \balpha_k \cdot \bnabla_k\Psi \right\Vert + \left\Vert m_x \beta_k  \Psi \right\Vert.
\label{eq:hfreexk1}
\ee
Since $\beta$ is unitary, the right term is simply bounded by $ m_x \|\Psi\| \le m_x \|\comp\Psi\|$. The left term can be bounded by $ - \bDelta_k $, using that the anti-commutator of the $ \alpha $-matrices is $ \{ \alpha^a, \alpha^b \} = 2 \delta^{ab} $:
\be
\begin{aligned}
	& \Vert - i \balpha_k\cdot \bnabla_k \Psi \Vert^2 =  \left\Vert - i \sum_a \alpha^a_k \partial_{x_k^a} \Psi \right\Vert^2 =  \sum_{ab} \left\langle \Psi , - \alpha^a_k \alpha^b_k \partial_{x_k^a} \partial_{x_k^b} \Psi \right\rangle_N\\
	= & \sum_a \left\langle \Psi , -\partial_{x_k^a} \partial_{x_k^a} \Psi \right\rangle_N =   \left\langle \Psi , -\bDelta_k \Psi \right\rangle_N \,.
\end{aligned}
\ee
By considering the Fourier transform of $\Psi$ and using that $|\bk|\leq \sqrt{1+|\bk|^2}\leq 1+|\bk|^2$, we obtain that
\be
	\left\langle \Psi, - \bDelta_k \Psi \right\rangle_N = \bigl\|\sqrt{-\bDelta_k} \Psi \bigr\|^2
	\leq  \bigl\|(\bone-\bDelta_k) \Psi\bigr\|^2 \le \Vert \comp \Psi \Vert^2.
\label{eq:derivbound}
\ee
Hence,
\be
\Vert \bH_{x_k}^{\free} \Psi \Vert \leq (1+m_x)\Vert \comp\Psi\Vert
\ee
in every sector and thus also in $\sH$.

\underline{$ \bH_{y_\ell}^\free $:} The same reasoning shows that
\begin{align}
\Bigl\| \sum_{\ell=1}^N \bH_{y_\ell}^\free \Psi^{(N)}\Bigr\|_N 
&\leq  \Bigl\| \sum_{\ell=1}^N (-i)\balpha_\ell \cdot \bnabla_\ell \Psi^{(N)} \Bigr\|_N +Nm_y \|\Psi^{(N)}\|_N \\
&\leq  \| \comp\Psi^{(N)} \|_N + m_y \|\comp \Psi^{(N)}\|_N\,.
\end{align}

\underline{$ \bH_{x_k}^{\inter} $:} is a finite sum of annihilation and creation operators $ \ba_s(\bx_k^{op}) $ and $ \ba^{\dagger}_s(\bx_k^{op}) $. They can be bounded using \eqref{eq:nelsonboundcutoff}: Since $(\bN+\bone)^{1/2}\leq \bN+\bone = d\Gamma(\bone)+\bone\leq \comp$, both $\|\ba(\bx^{op})\Psi\|$ and $\|\ba^\dagger(\bx^{op})\Psi\|$ are bounded by $c_1 \|\comp\Psi\|$ for some $c_1>0$. 
Summing up the bounds for $ \bH_{x_k}^{\free} $, $ \bH_{y_\ell}^{\free} $, and $ \bH_{x_k}^{\inter} $, we obtain the desired comparison inequality \eqref{eq:commcond1}.

\bigskip

For a proof of the commutator inequality \eqref{eq:commcond2}, we need to bound the commutator $ \vert \langle \Psi, [\bH, \comp] \Psi \rangle \vert $ by $ \Vert \comp^{1/2} \Psi \Vert^2 = \langle \Psi, \comp \Psi \rangle $. Since $ \bH $ is a sum of three terms, $\sum \bH_{x_k}^{\free} +\sum \bH_{y_\ell}^{\free} +\sum \bH_{x_k}^{\inter} $, and $ \comp $ is a sum of two terms, $\sum(\bone-\bDelta_k) +\sum(\bone-\bDelta_\ell)$, the commutator $ [\bH, \comp] $ contains 6 kinds of terms. The $ \bone $ trivially commutes with all operators and so can be dropped; both $ \bH_{x_k}^{\free} $ and $ \bH_{y_\ell}^{\free} $ commute with both $\bDelta_{k'} $ and $\bDelta_{\ell'}$, so 4 out of the 6 summands vanish; we are left with terms of the kind $ [\bH^{\inter}_{x_k}, - \bDelta_{k'}] $ and $ [\bH^{\inter}_{x_k}, -\bDelta_{\ell'}] $.

\underline{$ [\bH^{\inter}_{x_k}, - \bDelta_{k'}] $:} vanishes unless $ k = k' $ because $\bH^{\inter}_{x_k}$ does not involve $\bx_{k'}$. 
We drop the fermion index $ k $, write $ - \bDelta = \sum_{a=1}^3 i \partial_a i \partial_a $ and shift one derivative to the other side of the scalar product:
\be
\begin{aligned}
	&\Bigl\vert \bigl\langle \Psi, [\bH^{\inter}, - \bDelta] \Psi \bigr\rangle \Bigr\vert
	= \left\vert \sum_a \Bigl(\langle i \partial_a \bH^{\inter} \Psi, i \partial_a \Psi \rangle - \langle i \partial_a \Psi, i \partial_a \bH^{\inter} \Psi \rangle\Bigr) \right\vert\\
	\le &2 \sum_a \Bigl|\Im \langle i \partial_a \bH^{\inter} \Psi, i \partial_a \Psi \rangle\Bigr|
\end{aligned}
\ee
What happens if a derivative $ i \partial_a $ hits $ \bH^{\inter} \Psi $? The coordinate $ \bx_k $ appears in each of $ \ba(\bx^{op}_k) \Psi^{(N)} $ and $ \ba^{\dagger}(\bx^{op}_k) \Psi^{(N)} $ twice, cf.~\eqref{eq:a}: once in $ \cutoff(\by - \bx_k) $ and once in the factor $ \Psi^{(N+1)} $ or $ \Psi^{(N-1)} $. Therefore, by the product rule, $ i \partial_a \bH^{\inter} \Psi $ will contain 2 terms: In the first term, $ \cutoff $ is replaced by $ i \partial_a \cutoff $; 
the annihilation and creation operators modified in this way will be denoted by $\ba_{i\partial_a\cutoff}(\bx^{op})$ and $\ba^{\dagger}_{i\partial_a\cutoff}(\bx^{op}) $. The second term is just $ i \partial_a $ directly hitting $ \Psi $ before creation or annihilation takes place. That is,
\be
\begin{aligned}
	i \partial_a \left( \ba(\bx^{op}) \Psi \right)
	= &\ba_{i\partial_a \cutoff}(\bx^{op})  \Psi 
	+ \ba(\bx^{op}) \left( i \partial_a \Psi \right)\\
	i \partial_a \left( \ba^{\dagger}(\bx^{op}) \Psi \right)
	= &\ba^{\dagger}_{i\partial_a\cutoff}(\bx^{op})  \Psi 
	+ \ba^{\dagger}(\bx^{op}) \left( i \partial_a \Psi \right).
\end{aligned}
\label{eq:acomm}
\ee
This allows for further treatment of the commutator:
\be
\begin{aligned}
	&\Bigl\vert \langle \Psi, [\bH^{\inter}, - \bDelta] \Psi \rangle \Bigr\vert
	\le 2 \sum_a \Bigl| \Im \langle i \partial_a \bH^{\inter} \Psi, i \partial_a \Psi \rangle \Bigr|\\
	\le &c_3 \sum_a \biggl| \Im \Bigl\langle i \partial_a (\ba(\bx^{op}) + \ba^{\dagger}(\bx^{op})) \Psi, i \partial_a \Psi \Bigr\rangle \biggr|\\
	\le &c_3 \sum_a \biggl| \Im \biggl( \Bigl\langle \bigl( \ba_{i\partial_a\cutoff}(\bx^{op}) +  \ba^{\dagger}_{ i \partial_a\cutoff}(\bx^{op}) \bigr)\Psi, i \partial_a \Psi \Bigr\rangle + \Bigl\langle \bigl(\ba(\bx^{op}) + \ba^{\dagger}(\bx^{op})\bigr)  i \partial_a \Psi, i \partial_a \Psi \Bigr\rangle \biggr)\biggr|
\end{aligned}\label{83}
\ee
with $ c_3 \in \RRR_+ $. By symmetry of $ (\ba + \ba^{\dagger}) $, the second scalar product is real, so its imaginary part is 0. Using $|\Im\,z|\leq |z|$, the Cauchy-Schwarz $(CS)$ inequality, and \eqref{eq:ineq1b}, we obtain that
\be
\begin{aligned}
	\Bigl\vert \langle \Psi, [\bH^{\inter}, - \bDelta] \Psi \rangle \Bigr\vert
	&\le c_3 \sum_a \biggl\vert \Bigl\langle \bigl(  \ba_{i \partial_a\cutoff}(\bx^{op})) + \ba^{\dagger}_{ i \partial_a\cutoff}(\bx^{op}) \bigr) \Psi, i \partial_a \Psi \Bigr\rangle \biggr\vert \\
	&\overset{(CS)}{\le} c_3 \sum_a \Bigl\Vert \bigl( \ba_{i \partial_a \cutoff}(\bx^{op})+ \ba^{\dagger}_{i \partial_a\cutoff}(\bx^{op}) \bigr) \Psi\Bigr\Vert \; \bigl\Vert i \partial_a \Psi \bigr\Vert\\
	&\overset{\eqref{eq:nelsonboundcutoff}}{\leq} c_3 \sum_a \|i\partial_a\cutoff\| \Bigl( \|\bN^{1/2}\Psi\| + \|(\bN+\bone)^{1/2}\Psi\| \Bigr) \bigl\Vert i \partial_a \Psi \bigr\Vert \\
	&\leq c_3 \sum_a \|i\partial_a\cutoff\| \: 2 \bigl\|(\bN+\bone)^{1/2}\Psi\bigr\| \: \bigl\Vert i \partial_a \Psi \bigr\Vert \\
	&\leq d_1 \sum_a  \bigl\|\comp^{1/2}\Psi\bigr\| \: \bigl\Vert i \partial_a \Psi \bigr\Vert \\	
	&= d_1  \bigl\|\comp^{1/2}\Psi\bigr\| \sum_a  \langle \Psi, - \partial_a \partial_a \Psi \rangle^{1/2} \\	
	&\overset{\eqref{eq:ineq1b}}{\leq} d_1  \bigl\|\comp^{1/2}\Psi\bigr\| \: \biggl(3 \sum_a\langle \Psi, - \partial_a\partial_a \Psi \rangle\biggr)^{1/2} \,.	
\end{aligned}
\label{eq:bound4}
\ee
In the same way as in \eqref{eq:derivbound} but using $|\bk|\leq \sqrt{1+|\bk|^2}$, we find that
\be\label{DeltaboundR}
\langle \Psi^{(N)}, -\bDelta \Psi^{(N)}\rangle_N \leq \|\comp^{1/2} \Psi^{(N)}\|_N^2\,.
\ee
Hence,
\be
	\Bigl\vert \langle \Psi, [\bH^{\inter}_{x_k}, - \bDelta_k] \Psi \rangle \Bigr\vert
	\leq d_2  \bigl\|\comp^{1/2}\Psi\bigr\| \: \bigl\| \comp^{1/2}\Psi\bigr\| 
	= d_2 \|\comp^{1/2}\Psi\|^2\,,
\ee
as desired.

\bigskip

\underline{$ [\bH^{\inter}_{x_k}, -\bDelta_{\ell'}] $:} can be evaluated using the commutation relations
\be
\begin{aligned}
	\bigl[ \ba(\bx_k^{op}), d\Gamma_y(-\bDelta_y)\bigr] \Psi
	&=  \ba_{-\bDelta\cutoff}(\bx_k^{op})  \Psi \\
	\bigl[ \ba^{\dagger}(\bx_k^{op}), d\Gamma_y(-\bDelta_y) \bigr] \Psi
	&= - \ba^{\dagger}_{-\bDelta\cutoff}(\bx_k^{op}) \Psi\, ,
\end{aligned}
\label{eq:acomm2}
\ee
\x{which are perhaps most easily verified through a few lines of calculation by applying $d\Gamma_y(-\bDelta_y)$ to the defining equations \eqref{eq:a} of $\ba$ and $\ba^\dagger$; for the first relation, integrate by parts twice in $\tilde\by$; for the second, note that $\Psi^{(N-1)}$ does not depend on $\by_\ell$. From \eqref{eq:acomm2}, we obtain that} 
\be
\begin{aligned}
	&\left\vert \left\langle \Psi , \bigl[ \bH^{\inter}_{x_k}, d\Gamma_y(-\bDelta_y) \bigr] \Psi \right\rangle \right\vert\\
	&\le c_4 \biggl\vert \left\langle \Psi , \left[ \ba^{\dagger}(\bx_k^{op}), d\Gamma_y(-\bDelta_y) \right] \Psi \right\rangle + \left\langle \Psi , \left[ \ba(\bx_k^{op}), d\Gamma_y(-\bDelta_y) \right] \Psi \right\rangle  \biggr\vert\\
	&\overset{\eqref{eq:acomm2}}{=} c_4 \biggl\vert -\Bigl\langle \Psi ,  \ba^{\dagger}_{-\bDelta\cutoff}(\bx_k^{op})  \Psi \Bigr\rangle + \Bigl\langle \Psi , \ba_{-\bDelta\cutoff}(\bx_k^{op})  \Psi \Bigr\rangle  \biggr\vert\\
	&\overset{(CS)}\le c_4 \|\Psi\| \biggl( \Bigl\Vert  \ba_{-\bDelta\cutoff}^{\dagger}(\bx_k^{op})\Psi \Bigr\Vert + \Bigl\Vert \ba_{-\bDelta\cutoff}(\bx_k^{op})  \Psi \Bigr\Vert \biggr)\\
	&\overset{\eqref{eq:nelsonboundcutoff}}{\le} 2c_4 \|\bDelta\cutoff\| \; \|\Psi\| \; \Bigl\Vert  (\bN+\bone)^{1/2} \Psi \Bigr\Vert \\
	&\leq d_2 \, \|\comp^{1/2}\Psi\| \; \|\comp^{1/2}\Psi\|
	= d_2 \Vert \comp^{1/2} \Psi \Vert^2
\end{aligned}
\label{eq:bound5}
\ee
with $ c_4, d_2 \in \RRR_+ $. This concludes the proof of \eqref{eq:commcond2} and thus of Lemma~\ref{lem:selfadjoint}.
\end{proof}

\label{sec:lemHcinfty}

\begin{lemma}\label{lem:Hcinfty}
$\sH_c^\infty$ contains exactly those $\Psi\in\sH$ for which
\begin{enumerate}
\item For every $p\in\NNN_0$, $\Psi \in \mathrm{dom}(\comp^p)$ with $\comp$ as in \eqref{compdef}.
\item $\supp_3\,\Psi \subseteq \RRR^3$ is compact.
\end{enumerate}
In other words, the first and the third condition in the definition of $\sH_c^\infty$ can equivalently be replaced by $\Psi\in \mathrm{dom}(\comp^p)$ $\forall p$.
\end{lemma}

\begin{proof}
In terms of the (sector-wise) Fourier transform $\widehat\Psi$ of $\Psi$, the third condition \eqref{Hcinftynorm} can be expressed as
\be
\sum_{N=0}^\infty \int\limits_{\cQ^{s3,(N)}} \!\!\! dk \: P_{Nmn}(k) \: \bigl|\widehat\Psi^{(N)}(k)\bigr|^2   <\infty
\ee
for all $m,n\in\NNN_0$, where $P_{Nmn}$ is the polynomial
\be\label{PNmn}
P_{Nmn}(k_1,\ldots,k_{d}) = N^m \! \sum_{\substack{\alpha\in\NNN_0^{d}\\|\alpha|=n}} \! k_1^{2\alpha_1}\cdots k_{d}^{2\alpha_{d}}
\ee
with $d=3M+3N$. Likewise, the condition $\Psi\in \mathrm{dom}(\comp^p)$ can be expressed as
\be
\sum_{N=0}^\infty \int\limits_{\cQ^{s3,(N)}} \!\!\! dk \: Q_{Np}(k) \: \bigl|\widehat\Psi^{(N)}(k)\bigr|^2   <\infty\,,
\ee
where $Q_{Np}$ is the polynomial
\be\label{QNp}
Q_{Np}(k_1,\ldots,k_{d}) = \biggl(M+N+\sum_{j=1}^{d}k_j^2 \biggr)^{2p}\,. 
\ee
Every $Q_{Np}$ is bounded from above by a linear combination with $N$-independent coefficients of a finite number of $P_{Nmn}$'s with the same $N$, $m\leq 2p$, and $n\leq 2p$; thus, $\sH_c^\infty \subseteq \mathrm{dom}(\comp^p)$.

Conversely, every $P_{Nmn}$ is bounded from above by a linear combination with $N$-independent coefficients of a finite number of $Q_{Np}$'s with the same $N$ and $p\leq (m + n +1)/2$. Suppose that $\Psi\in \mathrm{dom}(\comp^p)$ for all $p$; then the weak derivatives of $\Psi$ satisfy \eqref{Hcinftynorm}. Therefore, each $\Psi^{(N)}$ lies in the Sobolev space $\HHH^n(\RRR^{3M+3N},(\CCC^4)^{\otimes (M+N)})\subset L^2(\RRR^{3M+3N},(\CCC^4)^{\otimes (M+N)})$ for every $n$; by the Sobolev embedding theorem \cite{adams}, $\Psi^{(N)}$ possesses a smooth representative and satisfies \eqref{Hcinftynorm} in the strong sense.
\end{proof}

\subsection{Support Growth}	\label{subsec:supp}

\begin{proof}[Proof of Lemma~\ref{lem:supp}]
For any number $M+N$ of free Dirac particles, it is well known \cite{dimock,zenk}, \cite[Thm.~2.20]{DM14}, \cite[Lemma 14]{pt:2013a}\footnote{The proof in \cite{pt:2013a} concerns smooth wave functions, but that implies the same support growth for $L^2$ functions, as an $L^2$ function can be approximated in the $L^2$ norm by a smooth function with the same support.} that the wave function propagates no faster than light, in particular that 
\begin{align}
\supp_{3x} \Psi_t &\subseteq \Gr(\supp_{3x} \Psi_0,|t|)\label{3xgrowth}\\
\supp_{3y} \Psi_t &\subseteq \Gr(\supp_{3y} \Psi_0,|t|)\label{3ygrowth}\,.
\end{align}
In fact, if we took the Hamiltonian of the $M+N$ particles to be \emph{only}
\be
	\bH_{Mx}^{\free} := \sum_{k = 1}^M \bH_{x_k}^{\free},
\ee
without the $\bH^\free_{y_\ell}$, then the $3y$-support would be invariant, $\supp_{3y}\Psi_t=\supp_{3y}\Psi_0$, while $\supp_{3x}\Psi$ would grow according to \eqref{3xgrowth}. Conversely, if we took \emph{only} $\sum_\ell \bH_{y_\ell}^\free$ as the Hamiltonian, then the $3x$-support would be invariant, while the $3y$-support grows according to \eqref{3ygrowth}.

Let us turn to the case with particle creation. Following a strategy of \cite[Theorem 3.4]{zenk}, we now prove that $ \supp_{3x}\Psi_t $ is not altered if we change the free evolution by including $ \bH^{\inter}$. 

First, we consider the evolution with \emph{only} $\bH^{\inter}_{x_k}$ as the Hamiltonian, without the free Dirac operators, and claim that then the 3$x$-support is invariant,
\be\label{supp3xHinterxk}
\supp_{3x}\Psi_t = \supp_{3x}\Psi_0\,.
\ee
Indeed, since $\bH^{\inter}_{x_k}$ can be decomposed into fibers as in \eqref{eq:fiber} (the continuous analog of being block diagonal), so can $\exp(-i \bH^{\inter}_{x_k} t)$ \cite[Thm.~XIII.85(c)]{reedsimon4}. Considering $\Psi_0$ in this fiber decomposition (i.e., as a function of $\bx^{3M}$ with values in $\sH_y$), it follows that $\Psi_t = \exp(-i \bH^{\inter}_{x_k} t)\Psi_0$ vanishes in those fibers where $\Psi_0$ does. So the $3M$-$x$-support is invariant, and therefore also the $3x$-support.

Second, under $ \bH^{\inter}_{x_k} $,
\be\label{supp3yHinterxk}
\supp_{3y} \Psi_t\subseteq \supp_{3y} \Psi_0 \cup \Gr(\supp_{3x} \Psi_0,\delta)\,.
\ee
Indeed, due to the fiber decomposition over $\bx^{3M}$ just mentioned, it suffices to consider just one fiber $\bx^{3M}$ and show that on it, for any set $G\subseteq \RRR^3$ with $B_\delta(\bx_k)\subseteq G$ and any vector $\Psi_0\in \sH_y$ with $3y$-support in $G$, also $\exp(-i \bH^{\inter}_{x_k}t)\Psi_0$ has $3y$-support in $G$. To see this, note that for disjoint sets $A,B$, $\sH_{A\cup B} = \sH_A \otimes \sH_B$, where $\sH_X$ means the bosonic Fock space over the set $X$, and consider $A=G\times \{1,2,3,4\}$ and $B=(\RRR^3 \setminus G)\times \{1,2,3,4\}$, so $\sH_{A\cup B}=\sH_y$. With respect to this tensor product decomposition, any $\Psi\in\sH_y$ factorizes as $\Psi=\Psi_A \otimes |\emptyset\rangle_B$ (where $|\emptyset\rangle$ denotes the Fock vacuum) iff it has $3y$-support in $G$. Since $\bH^{\inter}_{x_k}$ acts only on $\sH_A$, it is of the form $\bH_A \otimes \bone_B$; as a consequence, $\exp(-i \bH^{\inter}_{x_k} t)=\exp(-i\bH_A t)\otimes \bone_B$ maps $\Psi_A \otimes |\emptyset\rangle_B$ to $(\exp(-i\bH_A t)\Psi_A) \otimes |\emptyset\rangle_B$. Now the claim follows, and with it \eqref{supp3yHinterxk}.\footnote{\x{An alternative argument might go as follows but would require further work to be made rigorous.} On any fiber $\bx^{3M}$ and for any set $G\subseteq \RRR^3$ with $B_\delta(\bx_k)\subseteq G$, $\bH^{\inter}_{x_k}$ commutes with $\bone_G$. Therefore, also $\exp(-i\bH^{\inter}_{x_k} t)$ commutes with $\bone_G$. As a consequence, a vector $\Psi_0\in\sH_y$ with 3-support in $G$ evolves to one with 3-support in $G$.}

Third, back at the full time evolution, we can now decompose $ \bH $ into a part $\bH_{Mx}^\free$ which makes $ \supp_{3x} \Psi$ grow at most at the speed of light
and a sum of terms leaving $ \supp_{3x} \Psi$ invariant,
\be
	\sum_{k = 1}^M \bH_{x_k}^{\inter} +d\Gamma_y(\bH_y^\free) \,.
\ee
All contributions to $\bH$ are now put together using Trotter's product formula \cite[Thm.~VIII.31]{reedsimon1}:
\be
	e^{-i \bH t} \Psi_0
	= \lim_{n \rightarrow \infty} \left( \left( \prod_{k=1}^M e^{-i  \bH_{x_k}^{\inter} \frac{t}{n}} \right) e^{-i  d\Gamma_y(\bH_y^{\free}) \frac{t}{n}} e^{- i  \bH_{Mx}^{\free} \frac{t}{n}} \right)^n \Psi_0 =: \lim_{n \rightarrow \infty} \bU_{\! t,n}^n \Psi_0 \,.
\label{eq:trotter}
\ee
We claim that for all $k=0,\ldots,n$,
\begin{align}
\label{3xsuppgrTrotter}
	\supp_{3x}  \bU_{\! t,n}^k \Psi_0  &\subseteq \Gr\bigl(\supp_{3x} \Psi_0,\tfrac{k}{n} |t|\bigr)\\
\label{3ysuppgrTrotter}
\supp_{3y} \bU_{\! t,n}^k\Psi_0 &\subseteq \Gr\bigl(\supp_{3y} \Psi_0,\tfrac{k}{n}|t|\bigr) \cup \Gr\bigl(\supp_{3x} \Psi_0,\tfrac{k}{n}|t|+\delta\bigr)\,.
\end{align}
We proceed by induction along $k$. For $k=0$, the claim is trivially true. For the induction step $k\to k+1$, we use that $\Gr(\Gr(G,s),t) = \Gr(G,s+t)$ and $A\subseteq B \Rightarrow \Gr(A,t) \subseteq \Gr(B,t)$ and conclude that each factor $ e^{-i \frac{t}{n} \bH_{Mx}^{\free}} $ makes $ \supp_{3x} \Psi$ grow by $ |t|/n $, while all other factors in $\bU_{\! t,n}$ leave it invariant. Thus,
\begin{align}
\supp_{3x}\bU_{\! t,n}^{k+1}\Psi_0 
&\subseteq \Gr\bigl(\supp_{3x}\bU_{\! t,n}^k \Psi_0 ,\tfrac{1}{n}|t|\bigr) \\
&\subseteq \Gr\bigl( \Gr\bigl( \supp_{3x}\Psi_0,\tfrac{k}{n}|t|\bigr) , \tfrac{1}{n}|t|\bigr) \\
&= \Gr\bigl(\supp_{3x}\Psi_0,\tfrac{k+1}{n}|t| \bigr)\,.
\end{align}
Likewise, using $\Gr(A\cup B,t) = \Gr(A,t) \cup \Gr(B,t)$,
\begin{align}
	&\supp_{3y} \Bigl(e^{- i \frac{t}{n} d\Gamma_y(\bH_y^{\free})} e^{- i \frac{t}{n} \bH_{Mx}^{\free}} \bU_{\! t,n}^k \Psi_0\Bigr) \nonumber\\ 
	&\subseteq \Gr\bigl( \supp_{3y}\bU_{\! t,n}^k\Psi_0, \tfrac{1}{n}|t| \bigr)\\
	&\subseteq \Gr\Bigl( \Gr\bigl(\supp_{3y} \Psi_0,\tfrac{k}{n}|t|\bigr) \cup \Gr\bigl(\supp_{3x} \Psi_0,\tfrac{k}{n}|t|+\delta\bigr), \tfrac{1}{n}|t| \Bigr)\\
	&\subseteq \Gr\Bigl( \Gr\bigl(\supp_{3y} \Psi_0,\tfrac{k}{n}|t|\bigr),\tfrac{1}{n}|t| \Bigr) \cup \Gr\Bigl( \Gr\bigl(\supp_{3x} \Psi_0,\tfrac{k}{n}|t|+\delta\bigr), \tfrac{1}{n}|t| \Bigr)\\
	&= \Gr\bigl(\supp_{3y} \Psi_0,\tfrac{k+1}{n}|t|\bigr) \cup \Gr\bigl(\supp_{3x} \Psi_0,\tfrac{k+1}{n}|t|+\delta\bigr)\,.
\label{eq:supp3yinduction}
\end{align}
\x{What we need is the $3y$-support after $\bU_{\! t,n}^{k+1}$, and this operator comprises in addition the interaction factors $\exp(-i  \frac{t}{n}\bH_{x_k}^{\inter})$, see \eqref{eq:trotter}. The effect of these factors on $\supp_{3y}$ is that it will at most be joined with the $\delta$-grown set of $\supp_{3x}$ of the wave function obtained at that point, i.e., of $e^{-i  \frac{t}{n}d\Gamma_y(\bH_y^{\free})} e^{- i \frac{t}{n} \bH_{Mx}^{\free}}\bU_{\! t,n}^k \Psi_0$. Since $e^{-i \frac{t}{n}  d\Gamma_y(\bH_y^{\free})}$ does not change the $x$-support, we have that}
\begin{align}
&\supp_{3y}\bU_{\! t,n}^{k+1}\Psi_0 \nonumber\\
&\subseteq \eqref{eq:supp3yinduction}\cup \Gr\Bigl(\supp_{3x}e^{- i \frac{t}{n} \bH_{Mx}^{\free}}\bU_{\! t,n}^k \Psi_0,\delta\Bigr)\\
&\subseteq \eqref{eq:supp3yinduction}\cup \Gr\bigl(\supp_{3x} \Psi_0,\tfrac{k+1}{n}|t|+\delta\bigr)\\
	&= \Gr\bigl(\supp_{3y} \Psi_0,\tfrac{k+1}{n}|t|\bigr) \cup \Gr\bigl(\supp_{3x} \Psi_0,\tfrac{k+1}{n}|t|+\delta\bigr)\,,
\end{align}
completing the induction and thus the proof of \eqref{3xsuppgrTrotter} and \eqref{3ysuppgrTrotter}. Now we set $k=n$ and let $n\to\infty$. 
If a vector $\Psi_t\in \sH$ is a limit of a sequence $\Psi_{t,n}\to \Psi_t$ (such as $\Psi_{t,n}:=\bU_{\! t,n}^n\Psi_0$) and $\supp_{3x}\Psi_{t,n}\subseteq G$ for all $n$, then $\supp_{3x}\Psi_t \subseteq G$; likewise for $\supp_{3y}$. Thus, we have proved \eqref{eq:growxy} and therefore also \eqref{eq:grow}.
\end{proof}

\subsection{Smoothness Conditions}	\label{subsec:smooth}

We now prove Lemma~\ref{lem:diffable}. The proof is very similar to \cite[Thm. 7, Lemma 8]{ND:2019} and uses a commutator theorem due to Min-Jei Huang \cite[Theorem 2.3]{huang} that we will use in the following specialized form:
{\it Let $ \comp $ be a strictly positive, self-adjoint operator and $\bH$ a self-adjoint operator. Suppose that, for every $n\in\NNN_0$,
\begin{equation}
	\bZ_{n} := \comp^{n-1} [\bH,\comp] \comp^{-n}
\end{equation}
is densely defined and bounded. Then for every $p\in\NNN$,}
\be
e^{-i\bH t}[\mathrm{dom}(\comp^p)] = \mathrm{dom}(\comp^p) \quad  \forall t \in \RRR.
\label{eq:huang}
\ee

\bigskip

\begin{proof}[Proof of Lemma \ref{lem:diffable}]
Due to Lemmas~\ref{lem:supp} and \ref{lem:Hcinfty}, it suffices to show that $\bU(t)\sH_c^\infty\subseteq \mathrm{dom}(\comp^p)$ for the strictly positive comparison operator $ \comp $ given by \eqref{compdef} and every $p\in\NNN$. By Huang's theorem, it suffices to show that every $\bZ_n$ is densely defined and bounded.

To begin with, $\bZ_n$ is defined on the dense subspace $\mathrm{dom}(\comp^2)$. Indeed, for $\Psi \in \mathrm{dom}(\comp^2)$, $\comp^{-n}\Psi \in \mathrm{dom}(\comp^{n+2})$. Whenever $\Psi'\in \mathrm{dom}(\comp^p)$ for $p\geq 1$, then $\comp \Psi'$ exists and lies in $\mathrm{dom}(\comp^{p-1})$, and $\bH \Psi'$ exists by virtue of \eqref{eq:commcond1} and lies in $\mathrm{dom}(\comp^{p-1})$ as well. Thus, $[\bH,\comp]\comp^{-n}\Psi$ exists and lies in $\mathrm{dom}(\comp^n)$, so $\bZ_n\Psi$ exists. Now we show that $\bZ_n$ is bounded.

First, we consider $n=1$ and show that $ \bZ_1 = [\bH,\comp] \comp^{-1} $ is bounded or, equivalently, 
that there is a constant $d_1>0$ and a dense domain $D$ such that
\be\label{Z1bd1}
\bigl\| [\bH,\comp] \Psi \bigr\| \leq d_1\,\bigl\| \comp \Psi\bigr\|
\ee
for all $\Psi\in D$; we can choose $D=\sH_c^\infty$. 
We have already encountered this commutator in \eqref{eq:commcond2} in the proof of Lemma \ref{lem:selfadjoint}, where we noted that 
$ [\bH_{x_k}^{\inter},-\bDelta_{x_k}] $ and $ [\bH_{x_k}^{\inter},-\bDelta_{y_\ell}] $ are the only non-zero contributions. There, we could show that the form $ \langle \Psi, [\bH,\comp] \Psi \rangle $ is bounded by $ d \Vert \comp^{1/2} \Psi \Vert^2 = d \langle \Psi, \comp \Psi \rangle $. Now, we will prove \eqref{Z1bd1} instead.

\bigskip

\underline{$ \left[ \bH^{\inter}_{x_k}, -\bDelta_k \right] $} is evaluated using the product rule as in \eqref{eq:acomm}. 
We obtain that
\be\label{commaDelta}
	\left[ \ba(\bx^{op}_k), -\bDelta_k \right] 
	= - \ba_{-\bDelta\cutoff}(\bx^{op}_k) \Psi - 2 \sum_{a = 1}^3 \ba_{i\partial_a\cutoff}(\bx^{op}_k) (i \partial_{x_k^a} \Psi).
\ee
A similar equality holds true for $ \ba^{\dagger} $. So, with suitable $c_3,c_4,c_5\in\RRR_+$,
\be
\begin{aligned}
	& \Bigl\Vert \left[ \bH^{\inter}_{x_k} , -\bDelta_k \right] \Psi \Bigr\Vert
	\le c_3 \Bigl\Vert \left[ \ba^{\dagger}(\bx_k^{op}) , -\bDelta_k \right] \Psi \Bigr\Vert 
	+ c_3 \Bigl\Vert \left[ \ba(\bx_k^{op}) , -\bDelta_k \right] \Psi \Bigr\Vert \\
	&\overset{\eqref{commaDelta},\eqref{eq:nelsonboundcutoff}}{\le}
	c_4  \Bigl\Vert  (\bN+\bone)^{1/2} \Psi \Bigr\Vert  + c_5 \sum_{a=1}^3 \Bigl\Vert (\bN+\bone)^{1/2} i \partial_{x_k^a} \Psi \Bigr\Vert \\
	&\le  c_4  \bigl\Vert  \comp^{1/2} \Psi \bigr\Vert  + c_5 \biggl(3\sum_{a=1}^3 \sum_{N=0}^\infty (N+1) \|i \partial_{x_k^a} \Psi^{(N)}\|^2 \biggr)^{1/2}\\
	&\overset{\eqref{DeltaboundR}}{\le}  c_4  \left\Vert  \comp \Psi \right\Vert  + c_5 \biggl(3 \sum_{N=0}^\infty (N+1) \|\comp^{1/2} \Psi^{(N)}\|^2 \biggr)^{1/2}\\
	&\le c_4  \left\Vert  \comp \Psi \right\Vert  + c_5 \biggl(3 \sum_{N=0}^\infty \|\comp \Psi^{(N)}\|^2 \biggr)^{1/2}\\
	&= (c_4+c_5\sqrt{3}) \Vert \comp \Psi \Vert\,,
\end{aligned}
\label{eq:bound6}
\ee
which is what we wanted to show.

\bigskip

\underline{$ \left[ \bH^{\inter}_{x_k}, d\Gamma_y(-\bDelta_y) \right] $} is bounded as in \eqref{eq:bound5}:
\be
\begin{aligned}
	& \left\Vert \left[ \bH^{\inter}_{x_k} , d\Gamma_y(-\bDelta_y) \right] \Psi \right\Vert\\
	&\overset{\eqref{eq:acomm2}}\le  c_6 \left( \Bigl\Vert \ba_{-\bDelta\cutoff}^{\dagger}(\bx_k^{op})  \Psi \Bigr\Vert + \Bigl\Vert \ba_{-\bDelta\cutoff}(\bx_k^{op})  \Psi \Bigr\Vert \right) \\
	&\overset{\eqref{eq:nelsonboundcutoff}}\le  c_6  \Bigl\Vert (\bN+\bone)^{1/2} \Psi \Bigr\Vert 
	\le c_6 \left\Vert \comp^{1/2} \Psi \right\Vert 
	\le c_6 \left\Vert \comp \Psi \right\Vert\,.
\end{aligned}
\ee
So we have obtained the desired inequality \eqref{Z1bd1}.

\bigskip

We now turn to arbitrary $ n \in \NNN $ and show how the bounds generalize. Since $\comp\geq \bone$, clearly $ \Vert \comp^{-n} \Vert \le 1 $, so it remains to show
\be
	\Vert \comp^{n-1} [\bH, \comp] \Psi \Vert \le d_n \Vert \comp^n \Psi \Vert
\label{eq:Nn1ineq}
\ee
for a suitable $ d_n \in \RRR_+ $. To this end, we write
\be\label{commcompncomm}
\comp^{n-1} [\bH, \comp] = \bigl[ \comp^{n-1}, [\bH, \comp]\bigr] + [\bH, \comp]\comp^{n-1} \,.
\ee
Bounding these terms works similarly to the $n=1$ case. The $N$-sector of the first term applied to $\Psi$ is a linear combination of terms of the form
\be
\ba^\#_{\partial^\beta\cutoff}(\bx_k^{op}) \partial^\gamma \Psi^{(N)}\,,
\ee
where $\ba^\#$ means either $\ba$ or $\ba^\dagger$ and $\beta,\gamma$ are multi-indices with $|\beta+\gamma|\leq 2n$ (with $\beta$ acting on 3 variables and $\gamma$ on $3M+3N$). By \eqref{eq:nelsonboundcutoff}, each such term is bounded by $c (N+1)^{1/2} \|\comp^{|\gamma|/2}\Psi^{(N)}\|$, and they are jointly bounded by $c\|\comp^n\Psi\|$. The second term in \eqref{commcompncomm} is bounded by $d_1 \|\comp^n\Psi\|$, as follows from \eqref{Z1bd1} by inserting $\Psi \to \comp^{n-1}\Psi$, which still lies in $\sH_c^\infty$.
\end{proof}

\begin{proof}[Proof of Lemma~\ref{lem:smootht}]
Let $T>0$ be arbitrary but fixed, and let $\sH^T$ be the subspace of $L^2((-T,T)\times \cQ^{s3})$ of functions with the appropriate fermionic and bosonic permutation symmetry. We can identify $\sH^T$ with $L^2((-T,T),\sH)$. The function $t\mapsto \Psi_t$ ($\Psi_\square$ for short) belongs to $\sH^T$ because $\bU(t)$ is unitary, in fact $\|\Psi_\square\|_{\sH^T}=\sqrt{2T}\|\Psi_0\|_{\sH}$. 

We now show that $\Psi_\square \in \HHH^p((-T,T),\sH)$ for all $p\in\NNN_0$, i.e., that $\Psi_\square$ possesses weak time derivatives of any order that are square integrable over $t$. 
Indeed, by Lemma~\ref{lem:diffable}, $\Psi_t \in \sH_c^\infty \subseteq \mathrm{dom}(\bH^p)$, so $\Psi_\square$ is a $p$ times differentiable function $(-T,T)\to\sH$ with \x{$p$-th derivative $t\mapsto (-i\bH)^p \Psi_t=: (-i\bH)^p\Psi_\square$. It follows further that also weak time derivatives of $\Psi_\square$ exist and are given by $(-i\bH)^p \Psi_\square$, and that $t\mapsto \|\partial_t^p\Psi_t\|^2$ is differentiable for every $p$, therefore continuous for every $p$, and thus integrable over any $(-T,T)$. The upshot is that} $(-i\bH)^p \Psi_\square\in \sH^T$, which shows that $\Psi_\square \in \HHH^p((-T,T),\sH)$.

Now we show that $\Psi_\square\in \HHH^p((-T,T)\times \cQ^{s3})$. The weak $t$ derivative in $L^2((-T,T),L^2(\cQ^{s3}))$ gets translated to the weak $t$ derivative in $L^2((-T,T)\times \cQ^{s3})$. Since $\comp^q\Psi_0\in \sH_c^\infty$, also $\partial_t^p\comp^q \Psi_\square \in L^2((-T,T)\times \cQ^{s3})$ for every $q\in\NNN_0$. It follows that $\Psi_\square\in \HHH^p((-T,T)\times \cQ^{s3})$. By the Sobolev embedding theorem, $\Psi_\square\in L^2((-T,T)\times \cQ^{s3})$ has a smooth representative, and by continuity of $t\mapsto \Psi_t$, this smooth function is also a representative of $\Psi_t$ for every $t\in (-T,T)$. Since $T$ was arbitrary, $\Psi$ is smooth on $\RRR\times\cQ^{s3}$.
\end{proof}

\begin{proof}[Proof of Lemma~\ref{lem:smoothsingletime}]
Let $\Psi\in \sH_{cd}^\infty$. The first statement, that $\Psi(\lambda,\cdot)\in \sH_c^\infty$ for almost all $\lambda$, follows because, by the Fubini--Tonelli theorem, the left-hand side of \eqref{HcLambdainftynorm} is equal to
\be
\int_{\RRR^d} d\lambda \, \sum_{N=0}^\infty N^m \sum_\alpha \Bigl\| \partial^\alpha \Psi^{(N)}(\lambda,\cdot)\Bigr\|^2_{\sH^{(N)}}\,.
\ee
If this is finite, then the integrand has to be finite (so $\Psi(\lambda,\cdot)\in\sH_c^\infty$) for almost every $\lambda$, and we can define $\Psi_t(\lambda,\cdot)$ by \eqref{Psitlambdadef}. 

On the other hand, we can define $\bH_d$ on $\sH_{cd}^\infty$ by the same formulas as $\bH$ (so it does not act on $\lambda$), and define 
\be
\comp_d = \sum_{i=1}^d (\bone-\bDelta_{\lambda_i}) + \sum_{k=1}^M (\bone-\bDelta_{x_k}) + d\Gamma_y(\bone-\bDelta_{y})\,.
\ee
We may think of $\lambda$ as the coordinates of further particles with Hamiltonian $\bH_\lambda=0$. The same proof as for Lemma~\ref{lem:selfadjoint} shows that $\bH_d$ is essentially self-adjoint on $\sH_{cd}^\infty$. Since $\bH_d$ does not act on $\lambda$, $e^{-i\bH_d t}\Psi$ agrees with $\Psi_t$ as defined in \eqref{Psitlambdadef}. The same proof as for Lemma~\ref{lem:Hcinfty} shows that $\sH_{cd}^\infty$ consists of the functions in $\mathrm{dom}(\comp_d^p)$ with compact $x$-, $y$-, and $\lambda$-support. The same proof as for Lemma~\ref{lem:diffable} shows that $\sH_{cd}^\infty$ is invariant under $e^{-i\bH_d t}$, and the same proof as for Lemma~\ref{lem:smootht} shows that $\Psi$ is a smooth function of $t,\lambda$, and $q^{s3}$. 
\end{proof}

\subsection{Support Growth from Current Balance}

We now give a different argument for our bounds on support growth, based on the probability current and similar to \cite[Lemma 13]{pt:2013a} and \cite[Lemma 14]{pt:2013a}. We will need both arguments later. The probability current argument works under different hypotheses and yields a proof of Lemmas~\ref{lem:smoothsupp} and \ref{lem:1timeunique}. A version of this proof was first given in \cite[Lemma 4.12]{p:2010}.

\begin{proof}[Proof of Lemma~\ref{lem:smoothsupp}]
We write $T$ for $t$ and formulate the proof for $T>0$; it works in the same way for $T<0$. Fix any configuration $Q^{3}=(\bX^{3M},\bY^{3N})\in\cQ^3$ such that every $\bX_k\notin \Gr(\supp_{3x}\Psi_0,T)$ and every $\bY_{\!\ell}\notin \Gr(\supp_{3y}\Psi_0,T) \cup \Gr(\supp_{3x}\Psi_0,T+\delta)$. We need to show that $\Psi(T,Q^{3})=0$.

Since the $\Gr$ sets are closed, each $\bX_k$ and $\bY_{\!\ell}$ has positive distance from them, and so there is $\varepsilon>0$ such that each $B_\varepsilon(\bX_k)$ is disjoint from $\Gr(\supp_{3x}\Psi_0,T)$ and each $B_\varepsilon(\bY_{\!\ell})$ is disjoint from $\Gr(\supp_{3y}\Psi_0,T) \cup \Gr(\supp_{3x}\Psi_0,T+\delta)$. We define the set $C\subset \RRR_t\times \cQ^3$, called the truncated cone and depicted in Figure~\ref{fig:unique1}, by
\begin{align}
C&:=\biggl\{ (t;\bx^{3M},\by^{3n})\in\RRR\times\cQ^3 \bigg|t\in[0,T], \:\bx^{3M}\subset \Gr\bigl(\bX^{3M},T-t+\varepsilon\bigr), \nonumber\\
&\hspace{35mm}\by^{3n}\subset \Gr\bigl(\bX^{3M},T-t+\delta+\varepsilon\bigr)\cup \Gr\bigl(\bY^{3N},T-t+\varepsilon\bigr) \biggr\}\,.\label{Cdef}
\end{align}
We will show that $\Psi$ vanishes everywhere in $C$, which includes $(T,Q^3)$.

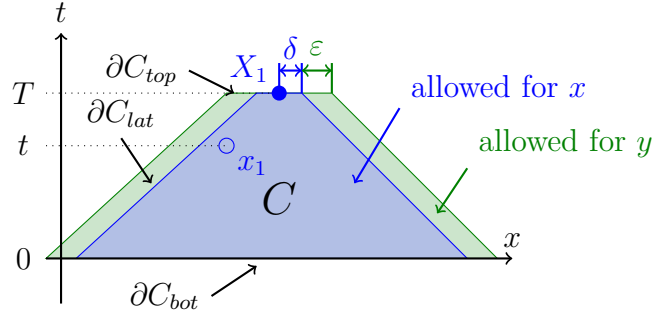
\begin{figure}[hbt]
    \centering
    \begin{tikzpicture}
	\filldraw[fill = green!50!black, fill opacity = 0.2,draw = green!50!black] (-0.2,0) -- ++(6,0) -- ++(-2.2,2.2) -- ++(-1.4,0) -- cycle;
	\filldraw[fill = blue!40!white, fill opacity = 0.5,draw = blue] (0.2,0) -- ++(5.2,0) -- ++(-2.2,2.2) -- ++(-0.6,0) -- cycle;
	\draw[->,thick] (-0.2,0) -- ++(6.2,0) node[anchor = south] {$ x $};
	\draw[->,thick] (0,-0.6) -- ++(0,3.6) node[anchor = south] {$ t $};
	\node at (-0.5,0) {$0$};
	\draw[dotted] (-0.2,1.5) -- ++(2.4,0);
	\node at (-0.5,1.5) {$t$};
	\draw[blue] (2.2,1.5) circle (0.1) node[anchor = north west] {$x_1$};
	\draw[dotted] (-0.2,2.2) -- ++(3.1,0);
	\node at (-0.5,2.2) {$T$};
	\fill[blue] (2.9,2.2) circle (0.1) node[anchor = south east] {$X_1$};
	\node at (2.9,0.8) {\Large{$C$}};
	\draw[->,thick] (2,-0.5) node[anchor = east] {$ \partial C_{bot} $} -- ++(0.6,0.4) ;
	\draw[->,thick] (0.8,1.6) node[anchor = south] {$ \partial C_{lat} $} -- ++(0.4,-0.6) ;
	\draw[->,thick] (1.7,2.5) node[anchor = east] {$ \partial C_{top} $} -- ++(0.6,-0.3) ;
	\draw[->,thick,blue] (4.5,2.0) node[anchor = south west] {allowed for $ x $} -- ++(-0.6,-1.0) ;
	\draw[->,thick,green!50!black] (5.4,1.2) node[anchor = south west] {allowed for $ y $} -- ++(-0.4,-0.6) ;
	\draw[thick,blue] (2.9,2.2) -- ++(0,0.4);
	\draw[thick,blue] (3.2,2.2) -- ++(0,0.4);
	\draw[thick,green!50!black] (3.6,2.2) -- ++(0,0.4);
	\draw[<->,thick,blue] (2.9,2.5) -- ++(0.3,0);
	\draw[<->,thick,green!50!black] (3.2,2.5) -- ++(0.4,0);
	\node[blue] at (3.05,2.8) {$\delta$};
	\node[green!50!black] at (3.4,2.8) {$\varepsilon$};
\end{tikzpicture}
    \caption{Depiction of one specific sector of the set $C$}
    \label{fig:unique1}
\end{figure}

The boundary of $C$ is piecewise smooth and consists of three parts:
\emph{bottom}, \emph{top} and \emph{lateral surface}, i.e.,
\be\label{eq:cone}
\begin{aligned}
	\partial C_\mathrm{bot} := &\{(0,q^3) \in C \} \\
	\partial C_\mathrm{top} := &\{(T,q^3) \in C \} \\
	\partial C_\mathrm{lat} := &\partial C \setminus ( \partial C_\mathrm{bot} \cup \partial C_\mathrm{top} ).
\end{aligned}
\ee
The lateral surface consists itself of several faces: one where $\|\bx_{k'}-\bX_k\|=T-t+\varepsilon$ (called $ \partial C_{\mathrm{lat},k'k} $), one where $\|\by_{\ell'}-\bY_\ell\|=T-t+\varepsilon$ (called $\partial C_{\mathrm{lat},\ell'\ell}$), and one where $\|\by_{\ell'}-\bX_k\|=T-t+\delta+\varepsilon$ (called $\partial C_{\mathrm{lat},\ell'k}$).

The \emph{unit surface normal vector} $\bn\in \RRR^{3M+3n+1}$ used in surface integrals is then:
\be
\begin{aligned}
	\bn &= (-1,0,\ldots,0) &\text{ at }(0,q^3)\\[3mm]
	\bn &= (1,0,\ldots,0) &\text{ at }(T,q^3) \\[2mm]
	\bn &= \frac{1}{\sqrt{2}}\Bigl(1,0,\ldots,\bn_{\bx_{k'}} = \frac{\bx_{k'} - \bX_k}{\Vert \bx_{k'} - \bX_k \Vert},\ldots,0\Bigr) &\text{ at }(t,q^3)\in\partial C_{\mathrm{lat},k'k}\\
	\bn &= \frac{1}{\sqrt{2}}\Bigl(1,0,\ldots,\bn_{\by_{\ell'}} = \frac{\by_{\ell'} - \bY_{\!\ell}}{\Vert \by_{\ell'} - \bY_{\!\ell} \Vert},\ldots,0\Bigr) &\text{ at }(t,q^3)\in\partial C_{\mathrm{lat},\ell'\ell}\\
	\bn &= \frac{1}{\sqrt{2}}\Bigl(1,0,\ldots,\bn_{\by_{\ell'}} = \frac{\by_{\ell'} - \bX_k}{\Vert \by_{\ell'} - \bX_k \Vert},\ldots,0\Bigr)&\text{ at }(t,q^3)\in\partial C_{\mathrm{lat},\ell'k}.
\end{aligned}
\label{eq:n}
\ee
Sketches of the cones and $\bn$ are provided by Figures \ref{fig:unique1a} through \ref{fig:unique3}.

\begin{minipage}{0.45\textwidth}
    \begin{tikzpicture}
	\filldraw[fill = green!50!black, fill opacity = 0.2,draw = green!50!black] (-0.2,0) -- ++(6,0) -- ++(-2.2,2.2) -- ++(-1.4,0) -- cycle;
	\filldraw[fill = blue!40!white, fill opacity = 0.5,draw = blue] (0.2,0) -- ++(5.2,0) -- ++(-2.2,2.2) -- ++(-0.6,0) -- cycle;
	\draw[->,thick] (-0.2,0) -- ++(6.2,0) node[anchor = south] {$ x $};
	\draw[->,thick] (0,-0.6) -- ++(0,3.6) node[anchor = south] {$ t $};
	\node at (-0.5,0) {$0$};
	\draw[dotted] (-0.2,2.2) -- ++(3.1,0);
	\node at (-0.5,2.2) {$T$};
	\fill[blue] (2.9,2.2) circle (0.1) node[anchor = south east] {$X_1$};
	\node at (2.9,0.8) {\Large{$C$}};
	\draw[->,blue, line width = 1] (3.1,2.2) -- ++(0,0.5) node[anchor = south] {$ \boldsymbol{n} $};
	\draw[->,blue, line width = 1] (3.1,0) -- ++(0,-0.5) node[anchor = west] {$ \boldsymbol{n} $};
	\draw[->,blue, line width = 1] (4.2,1.2) -- ++(0.3,0.3) node[anchor = south west] {$ \boldsymbol{n} $};
	\draw[->,green!50!black, line width = 1] (5.2,0.6) -- ++(0.3,0.3) node[anchor = south west] {$ \boldsymbol{n} $};
\end{tikzpicture}
	
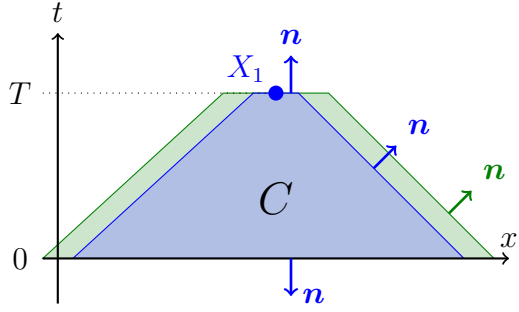
\captionof{figure}{The normal vectors $ \bn $ for the example in Figure \ref{fig:unique1}}
    \label{fig:unique1a}
\end{minipage}
\hfill
\begin{minipage}{0.45\textwidth}
	%Render engine by Sascha Lill
%(X,Z) = \R(x,y,z) casts a point on the X-Z-screen
%The viewer is located at y = -\viewer_distance (in cm)
%P = (\Px,\Py) is the point of infinite distance (in cm)
\def\viewer_distance{20}
\def\Px{4}
\def\Py{8}
\def\frac(#1,#2){#1/#2}
\def\s(#1){\viewer_distance/(\viewer_distance + #1)}
\def\R(#1,#2,#3){({\s(#2)*#1+(1-\s(#2))*\Px}, {\s(#2)*#3+(1-\s(#2))*\Py})}
\begin{tikzpicture}
	\draw[dotted] \R(0,0,0) -- \R(5,0,0) -- \R(5,5.8,0) -- \R(0,5.8,0)-- cycle;
	\filldraw[black!80!white,fill opacity = 0.2] \R(0,0,0) -- \R(2.2,2.2,2.2) -- \R(2.2,3.6,2.2) -- \R(0,5.8,0)-- cycle;
	\filldraw[black!40!white,fill opacity = 0.2] \R(0,0,0) -- \R(5,0,0) -- \R(2.8,2.2,2.2) -- \R(2.2,2.2,2.2)-- cycle;
	\filldraw[black!20!white,fill opacity = 0.2] \R(5,0,0) -- \R(5,5.8,0) -- \R(2.8,3.6,2.2) -- \R(2.8,2.2,2.2)-- cycle;
	\filldraw[black!60!white,fill opacity = 0.2] \R(2.2,2.2,2.2) -- \R(2.8,2.2,2.2) -- \R(2.8,3.6,2.2) -- \R(2.2,3.6,2.2)-- cycle;
	\draw[->,thick] (-0.2,0) -- ++(6,0) node[anchor = south west] {$x$};
	\draw[->,thick] \R(0,-0.2,0) -- ++\R(0,8,0) node[anchor = south west] {$y$};
	\draw[->,thick] (0,-0.2) -- ++(0,3.7) node[anchor = south west] {$z$};
	\node at (-0.3,-0.3) {$0$};
	
	\fill \R(2.5,2.9,2.2) circle (0.1) node[anchor = south east] {$X$};
	\node at (-0.3,2.2) {$T$};
	\draw[dotted] (-0.1,2.2) -- (2.5,2.2) -- \R(2.5,2.9,2.2);
	\draw \R(2.5,2.9,2.2) -- \R(2.5,5,2.2);
	\draw \R(2.8,2.9,2.2) -- \R(2.8,5,2.2);
	\node at \R(2.6,6,2.2) {$\varepsilon$};
	\draw \R(2.5,2.9,2.2) -- \R(3.5,2.9,2.2);
	\draw \R(2.5,2.2,2.2) -- \R(3.5,2.2,2.2);
	\node at \R(4.2,2.6,2.2) {$\delta + \varepsilon$};
	
	\draw[->,line width = 1] \R(4,2.5,1) -- ++(0.4,0.4) node[anchor = south west] {$\boldsymbol{n}$};
	\draw[->,line width = 1] \R(1.5,0.5,0.5) -- \R(1.5,0,1) node[anchor = south] {$\boldsymbol{n}$};
	\draw[->,line width = 1, gray] \R(3,3,0) -- \R(3,3,-0.6) node[anchor = north] {$\boldsymbol{n}$};
\end{tikzpicture}
	
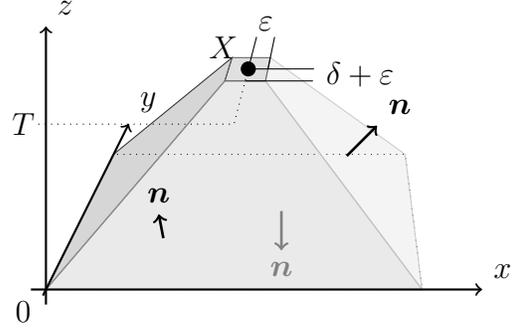
\captionof{figure}{The set $C$ is frustum-shaped in configuration space.}
    \label{fig:unique2}
\end{minipage}

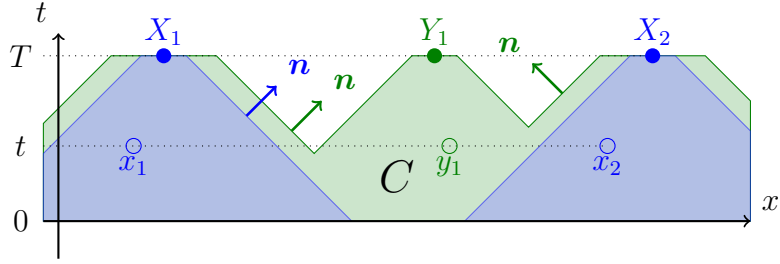
\begin{figure}[hbt]
    \centering
    \begin{tikzpicture}
	\filldraw[green!50!black, fill opacity = 0.2] (-0.2,1.3) -- ++(0,-1.3) -- ++(9.4,0) -- ++(0,1.6) -- ++(-0.6,0.6) -- ++(-1.4,0) -- ++(-0.95,-0.95) -- ++(-0.95,0.95) -- ++(-0.6,0) -- ++(-1.3,-1.3) -- ++(-1.3,1.3) -- ++(-1.4,0) -- cycle;
	\filldraw[draw = blue, fill = blue!40!white, opacity = .5] (-0.2,0.9) -- ++(0,-0.9) -- ++(4.1,0) -- ++(-2.2,2.2) -- ++ (-0.6,0) -- cycle;
	\filldraw[draw = blue, fill = blue!40!white, opacity = .5] (5.4,0) -- ++(3.8,0) -- ++(0,1.2) -- ++(-1,1) -- ++ (-0.6,0) -- cycle;
	\draw[->,thick] (-0.2,0) -- ++ (9.4,0) node[anchor = south west] {$x$};
	\draw[->,thick] (0,-0.5) -- ++ (0,3) node[anchor = south east] {$t$};
	\node at (-0.5,0) {$0$};
	\draw[dotted] (-0.2,1) -- ++(7.5,0);
	\draw[dotted] (-0.2,2.2) -- ++(8.1,0);
	\node at (-0.5,1) {$t$};
	\node at (-0.5,2.2) {$T$};
	\fill[blue] (1.4,2.2) circle (0.1) node[anchor = south] {$X_1$};
	\fill[green!50!black] (5,2.2) circle (0.1) node[anchor = south] {$Y_1$};
	\fill[blue] (7.9,2.2) circle (0.1) node[anchor = south] {$X_2$};
	\draw[blue] (1,1) circle (0.1) node[anchor = north] {$x_1$};
	\draw[green!50!black] (5.2,1) circle (0.1) node[anchor = north] {$y_1$};
	\draw[blue] (7.3,1) circle (0.1) node[anchor = north] {$x_2$};
	\node at (4.5,0.6) {\Large{$C$}};
	\draw[->,blue,line width = 1] (2.5,1.4) -- ++ (0.4,0.4) node[anchor = south west] {$ \boldsymbol{n} $};
	\draw[->,green!50!black,line width = 1] (3.1,1.2) -- ++ (0.4,0.4) node[anchor = south west] {$ \boldsymbol{n} $};
	\draw[->,green!50!black,line width = 1] (6.7,1.7) -- ++ (-0.4,0.4) node[anchor = south east] {$ \boldsymbol{n} $};
\end{tikzpicture}
    \caption{Projection to space-time of the set $C$ for an example of a configuration $Q= (X_1,X_2,Y_1)$}
    \label{fig:unique3}
\end{figure}

We define the probability current $\bj$, a vector field on $\RRR\times \cQ^3$, in analogy to the Dirac current; on the $n$-$y$-particle sector, it is defined by
\be
\begin{aligned}
	\bj: \RRR \times {\cQ}^{3,(n)} &\rightarrow \RRR \times \RRR^{3M + 3n}\\
	\bj^{0}(t,q^3) &:= \sum_{\br,\bs}\Psi(t,q^3,\br,\bs)^* \Psi(t,q^3,\br,\bs)\\
	\bj^{x_k,a}(t,q^3) &:= \sum_{\br,\bs,r'_k}\Psi(t,q^3,r_k)^* (\gamma^0\gamma^a)_{r_k,r'_k} \Psi(t,q^3,r'_k)\\
	\bj^{y_\ell,a}(t,q^3) &:= \sum_{\br,\bs,s'_\ell} \Psi(t,q^3,s_\ell)^* (\gamma^0\gamma^a)_{s_\ell,s'_\ell} \Psi(t,q^3,s'_\ell),
\end{aligned}
\label{eq:j}
\ee
where $a=1,2,3$, and not all spin indices are always made explicit. The current vector field has the property that for each $k$ and each $\ell$, $(\bj^0,\bj^{x_k,1},\bj^{x_k,2},\bj^{x_k,3})$ and $(\bj^0,\bj^{y_\ell,1},\bj^{y_\ell,2},\bj^{y_\ell,3})$ are future-causal (i.e., future-timelike or future-lightlike). It follows that $\bj\cdot\bn\geq 0$ (with $\cdot$ the Euclidean inner product in $3M+3n+1$ dimensions) on $\partial C_\mathrm{lat}$ and $\partial C_\mathrm{top}$. 

Let $C^{(n)}:= C \cap (\RRR\times \cQ^{3,(n)})$ be the $n$-$y$-particle sector of $C$. By the Ostrogradski--Gauss integral theorem (divergence theorem),
\begin{align}
	 \int_{C^{(n)}} \hspace{-4mm} d(t,q^3) \! \sum_{\mu=0}^{3M+3n}\partial_{\mu} \bj^{\mu}
	 &= \int_{\partial C^{(n)}} \hspace{-6mm}  d(t,q^3) \; \bj \cdot \bn\\
	 &= \int_{\partial C_\mathrm{bot}^{(n)}} \hspace{-6mm} d(t,q^3) \; \bj \cdot \bn 
	 + \int_{\partial C_\mathrm{top}^{(n)}} \hspace{-6mm} d(t,q^3) \; \bj \cdot \bn  
	 + \int_{\partial C_\mathrm{lat}^{(n)}} \hspace{-6mm} d(t,q^3) \; \bj \cdot \bn \,.\label{bot+top+lat}
\end{align}
The integral over $ \partial C_\mathrm{bot}^{(n)} $ is 0 by hypothesis, those over $ \partial C_\mathrm{lat}^{(n)} $ and $ \partial C_\mathrm{top}^{(n)} $ are non-negative, so the left-hand side must be non-negative. We will show that
\be\label{intdiv0}
	\sum_{n=0}^\infty \int_{C^{(n)}} \hspace{-4mm} d(t,q^3) \! \sum_{\mu=0}^{3M+3n}\partial_{\mu} \bj^{\mu} = 0\,.
\ee
It then follows that each summand must vanish, so the right-hand side of \eqref{bot+top+lat} vanishes, and in particular the integral over $\partial C_\mathrm{top}^{(N)}$ vanishes, which is what we wanted to show. So it remains to prove \eqref{intdiv0}. 

In Lemmas 13 and 14 of \cite{pt:2013a}, the ($3M+3n+1$-dimensional) divergence of $\bj$ vanished everywhere in $C$. In our situation, this is not the case, due to creation and annihilation terms in the Hamiltonian. However, as we will show, the integral of the divergence of $\bj$ over $C$ still vanishes because the creation and annihilation terms transfer probability to other places in $C$ but not outside of $C$. 
Indeed,
\be
\begin{aligned}
	\partial_0 \bj^0 &= \partial_0 \bigl( \Psi^* \Psi \bigr) =  2\, \Im( \Psi^*  \bH \Psi) =\\
	&= \sum_{k = 1}^{M} 2 \, \Im( \Psi^* \bH^{\free}_{x_k} \Psi) + \sum_{\ell = 1}^{n} 2\, \Im( \Psi^* \bH^{\free}_{y_\ell} \Psi) + \sum_{k = 1}^{M} 2 \,\Im( \Psi^* \bH^{\inter}_{x_k} \Psi)\\
	&= -\sum_{\mu=1}^{3M+3n} \partial_\mu \bj^\mu +\sum_{k = 1}^{M} 2 \,\Im( \Psi^* \bH^{\inter}_{x_k} \Psi)\,.
\end{aligned}
\label{eq:temporaldiv}
\ee
So it suffices to show that
\be
	\sum_{n=0}^\infty \int_{C^{(n)}_t} \hspace{-4mm} dq^3 \; \Im(\Psi^* \bH^{\inter}_{x_k} \Psi) = 0
\label{eq:jHintzero}
\ee
for every $k \in \{1,\ldots,M\}$, with $C_t:= \{q^3:(t,q^3)\in C\}$.
Let
\be
C_{yt}^{(n)}:=\Bigl[\Gr\bigl(\bX^{3M},T-t+\delta+\varepsilon\bigr)\cup \Gr\bigl(\bY^{3N},T-t+\varepsilon\bigr) \Bigr]^{n}\,.
\ee
Consider any $\bx^{3M}\subset G:=\Gr\bigl(\bX^{3M},T-t+\varepsilon \bigr)$. From the definition of the annihilation operator $\ba_s(\bx)$, cf.~\eqref{eq:a}, and the fact that $B_\delta(\bx_k)\subset \Gr\bigl(\bX^{3M},T-t+\delta+\varepsilon\bigr)$, we obtain that
\begin{align}
c_n&:=\int\limits_{C^{(n)}_{yt}} \hspace{-1mm} d\by^{3n} \Biggl[ \sum_{s=1}^4 g^s\,  \Psi^*(\bx^{3M},\by^{3n}) \Bigl[\ba_s(\bx_k) \Psi\Bigr](\bx^{3M},\by^{3n})\Biggr]^* \nonumber\\
&\, = \int\limits_{C^{(n+1)}_{yt}} \hspace{-2mm} d\tilde\by^{3(n+1)} \sum_{s=1}^4 (g^s)^*\,  \Psi^*(\bx^{3M},\tilde\by^{3(n+1)}) \Bigl[\ba_s^{\dagger}(\bx_k) \Psi\Bigr](\bx^{3M},\tilde\by^{3(n+1)})\,.
\end{align}
Thus, by the definition \eqref{eq:Hint} of $\bH_{x_k}^{\inter}$ in terms of $\ba$ and $\ba^\dagger$, the left-hand side of \eqref{eq:jHintzero} equals
\be
\int_{G^M} \hspace{-2mm} d\bx^{3M} \, \Im(c_0^*) +\sum_{n=1}^\infty \int_{G^M} \hspace{-2mm} d\bx^{3M} \, \Im(c_n^* + c_{n-1})\,,
\ee
a telescopic sum with partial sum up to $N_0$ given by $\int_{G^M} d\bx^{3M} \, \Im \, c_{N_0}^*$, whose modulus is
\begin{align}
&\leq \int_{G^M} d\bx^{3M} \, \bigl| c_{N_0}^*\bigr|\\
&\leq \int_{\RRR^{3M}} d\bx^{3M} \, \Biggl|\int_{C^{(N_0)}_{yt}} \hspace{-1mm} d\by^{3N_0} \sum_{s=1}^4 g^s\,  \Psi^*(\bx^{3M},\by^{3N_0}) \Bigl[\ba_s(\bx_k) \Psi\Bigr](\bx^{3M},\by^{3N_0})\Biggr|\\
&\leq \int_{\RRR^{3M}} d\bx^{3M} \int_{\RRR^{3N_0}} \hspace{-1mm} d\by^{3N_0} \sum_{s=1}^4 |g^s|   \Biggl|\Psi^*(\bx^{3M},\by^{3N_0}) \Bigl[\ba_s(\bx_k) \Psi\Bigr](\bx^{3M},\by^{3N_0})\Biggr|\\
&\leq \sum_{s=1}^4 |g^s| \, \bigl\|\Psi^{(N_0)}\bigr\| \, \bigl\| \ba_s(\bx_k^{op}) \Psi^{(N_0+1)}\bigr\|\\
&\leq 2\|g\| \, \bigl\|\Psi^{(N_0)}\bigr\| \, \|\cutoff\| \, \bigl\| (N_0+1)^{1/2} \Psi^{(N_0+1)}\bigr\| \stackrel{N_0 \to \infty}{\longrightarrow} 0
\end{align}
using the Cauchy-Schwarz inequality, \eqref{eq:ineq1b}, \eqref{eq:nelsonboundcutoff}, and the hypotheses $\|\Psi_t\|<\infty$ and $\|N^{1/2}\Psi_t\|<\infty$. This shows that \eqref{eq:jHintzero} indeed converges to zero (in fact absolutely). This concludes the proof.
\end{proof}

\begin{proof}[Proof of Lemma~\ref{lem:1timeunique}]
This is a simple corollary of Lemma~\ref{lem:smoothsupp}. By linearity, it suffices to show that the only solution with $\Psi_0=0$ is $\Psi=0$. But if $\Psi_0=0$, then $\supp_{3x}\Psi_0=\emptyset = \supp_{3y}\Psi_0$, and by \eqref{eq:growxy}, $\supp_{3x}\Psi_t=\emptyset =\supp_{3y}\Psi_t$ for every $t\in\RRR$.
\end{proof}

\section{Proofs: Multi-Time Evolution}	
\label{sec:multitime}

In this section, we prove Theorem~\ref{thm:1}; in particular, we prove existence and uniqueness of solutions of the multi-time equations \eqref{multi34} with given initial data \eqref{IVP}.

\subsection{Uniqueness}	
\label{subsec:unique}

\begin{lemma}[Uniqueness of solutions]\label{lem:unique}
For every $\Psi_0\in C^\infty(\cQ^{s3})$, there is at most one solution $\Phi\in C^\infty(\sS_\delta^s)$ satisfying \eqref{summability} to the multi-time equations \eqref{multi34} with initial data $\Psi_0$ as in \eqref{IVP}.
\end{lemma}

\begin{proof}
For any $q^4\in\cQ^4$, let $J(q^4)$ denote the number of different values of time variables that occur in $q^4$; for $q^4\in\sS_\delta$, this is the minimal number of families. Let
\be\label{SJdef}
\sS_{\delta J}=\{q^4\in\sS_\delta:J(q^4)\leq J\}\,.
\ee
We proceed by induction along $J$. 

\x{Everywhere on $\sS_{\delta 1}$, $\Phi$ is a function of just 1 distinct time variable, and the restriction of $\Phi$ to $\sS_{\delta 1}$ obeys the single-time equation \eqref{1time}.} By \eqref{summability} for $J=1$, $\|\Psi_t\|<\infty$ and $\|N^{1/2}\Psi_t\|<\infty$, and by Lemma~\ref{lem:1timeunique}, $\Phi$ is unique on $\sS_{\delta 1}$.

The induction assumption asserts that $\Phi$ is unique on $\sS_{\delta J-1}$, and we need to prove uniqueness on $\sS_{\delta J}$. Fix $Q^4\in\sS_{\delta J}$ and sort the families in increasing order of the time variables, $T_1<\ldots<T_{J-1}<T_{J}$. As in \eqref{qjdef}, we write $Q^4=(T_1,Q_1,\ldots,T_J,Q_J)$ with \x{$Q_j=(X_j,Y_j)$} the configuration of all particles with time variable $T_j$. Let $\phi$ be the function obtained from $\Phi$ on $\sS_{\delta J}$ by inserting $t_1=T_1,\ldots, t_{J-1}=T_{J-1}$ while keeping $t_J$ variable, and inserting the configurations of the families 1 through $J-1$, $q_1=Q_1,\ldots,q_{J-1}=Q_{J-1}$ while keeping $q_J$ variable; \x{we write $q_J=(x_J,y_J)=(\bx_{J1},\ldots,\bx_{JM_J},\by_{J1},\ldots,\by_{JN_J})$ with $M_J= M-\sum_{j=1}^{J-1}\#X_j$. Using the abbreviations
\be
\begin{aligned}
\sX(t_J)&= \RRR^3\setminus \bigcup_{j=1}^{J-1}
\bigl[\Gr(X_j,t_J-t_j+2\delta)\cup\Gr(Y_j,t_J-t_j+\delta)\bigr]\\
\sY(t_J)&= \RRR^3\setminus \bigcup_{j=1}^{J-1}
\bigl[\Gr(X_j,t_J-t_j+\delta)\cup\Gr(Y_j,t_J-t_j)\bigr]\,,
\end{aligned}
\label{XYdef}
\ee
the function $(t_J,x_J,y_J)\mapsto \phi(t_J,x_J,y_J)$ is defined on the set
\be\label{phidef}
\bigcup_{t_J>t_{J-1}} \biggl(\{t_J\} \times \sX(t_J)^{M_J} \times \bigcup_{N_J=0}^\infty \sY(t_J)^{N_J} \biggr)\,.
\ee
This set includes $(T_J,Q_J)$ because $Q^4\in\sS_{\delta}$. Since $\Phi$ is a solution of \eqref{multi34} for $j=J$, $\phi$ is a solution of \eqref{1time} with time variable $t=t_J$, $M_J$ rather than $M$ fermions, and initial data given by $\Phi$ where $t_J=t_{J-1}$. 
By \eqref{summability}, $\|\phi_t\|<\infty$ and $\|N^{1/2}\phi_t\|<\infty$ with the norm $\|\cdot\|$ taken over the configurations in the bracket in \eqref{phidef}. By Lemma~\ref{lem:1timeunique}, $\phi$ is uniquely fixed on its domain \eqref{phidef}}, in particular at $(T_J,Q_J)$, as claimed.
\end{proof}

\begin{lemma}[Growth of 4-support]\label{lem:4supp}
Every solution $\Phi\in C^\infty(\sS_\delta^s)$ of \eqref{multi34} satisfying \eqref{summability} obeys propagation locality up to $\delta$ as in \eqref{4supp}.
\end{lemma}

\begin{proof}
It now plays a role that we evolve negative $t_j$ only towards the past and positive ones only towards the future. So let $J_+(q^4)$ denote the number of positive time values, $J_-(q^4)$ that of negative ones, and let $J_0(q^4)$ be 1 or 0 depending on whether 0 occurs as a time value. Let $\sS_{\delta J_+ J_-}=\{q^4\in\sS_\delta:J_+(q^4)\leq J_+, J_-(q^4)\leq J_-\}$. We proceed by induction, first along $J_+$, then along $J_-$.

On $\sS_{\delta 10}$ and $\sS_{\delta 01}$, the statement is provided by Lemma~\ref{lem:smoothsupp}. If true on $\sS_{\delta, J_+-1,0}$, it follows on $\sS_{\delta, J_+,0}$ in the same way as in the previous proof. For $J_->0$, order the negative time variables so that the least comes last. Then the induction step from $\sS_{\delta,J_+,J_--1}$ to $\sS_{\delta,J_+,J_-}$ works in the same way as the previous proof but in the opposite time direction.
\end{proof}

\subsection{Commutator Conditions}
\label{subsec:consist}

The commutator condition \eqref{eq:consist} arises heuristically as the consistency condition. It will also play a role in our proof of the existence of solutions, specifically for proving that the function obtained by solving one of the multi-time equations also solves the others. To this end, we verify the commutator condition in this section. The appropriate condition concerns the Hamiltonians $H_j^P$ corresponding to a family with a common time variable $t_j$ as in \eqref{multi34}.

To check it explicitly, we begin with other partial Hamiltonians, associated with individual particles and defined by $\bH_{x_k} = \bH_{x_k}^\free + \bH_{x_k}^\inter$ and $\bH_{y_{\ell}} = \bH_{y_\ell}^\free$.
In the following, we will understand these Hamiltonians as acting on functions of $q^4$ instead of $q^3$ and then write them as $H_{z}$ instead of $\bH_z$, $z\in\{x_1\ldots x_M,y_1\ldots y_N\}$. In fact, we will assume for $H_z$ the slightly more general form corresponding to \eqref{multi56} instead of \eqref{multi34} (using cut-off Green functions and applicable also outside of $\sS_\delta$),
\begin{align}
H_{x_k}\Phi(q^4) &= H^{\free}_{x_k} \Phi(q^4) \nonumber\\[3mm]
& + \sqrt{N+1} \sum_{r_k',s_{N+1}} g^*_{r'_k r_k s_{N+1}} \int_{B_\delta(\bx_k)} \hspace{-7mm} d^3\tilde\by\:\: \cutoff(\tilde\by-\bx_k) \:\: \Phi^{(N+1)}_{r_k',s_{N+1}}\Bigl(x^{4M},\bigl(y^{4N}, (x_k^0, \tilde\by)\bigr)\Bigr) \nonumber\\
& + \frac{1}{\sqrt{N}} \sum_{\ell=1}^N \sum_{r'_k} \G_{r_k r'_k s_\ell}(y_\ell - x_k) \: \Phi_{r'_{k}\widehat{s_\ell}}^{(N-1)}\bigl(x^{4M}, y^{4N} \backslash y_\ell\bigr)  \label{multi7}\\[4mm]
H_{y_{\ell}}\Phi(q^4) &= H_{y_\ell}^\free \Phi(q^4)\,.\label{multi8}
\end{align}
\x{In order to convince oneself that the $H_z$ can be applied to smooth functions on $\cQ^{s4}$ or on $\widehat{\sS^s_\delta}$ or on $\sS^s_\delta$, one needs to verify the following three properties for each of these sets: (i)~It contains with every $q^{s4}$ also a spacelike neighborhood of $q^{s4}$. (ii)~It contains with every $q^{s4}$ also all other spin components of configurations with the same space-time positions. It now follows that the derivatives in $ H^{\free} $ are well-defined. We still need certain configurations to describe creation and annihilation: (iii)~It contains with every $q^{s4}$ also all configurations with a $y$-particle added or removed in the 3d $\delta$-neighborhood of any $x$-particle. Thus, all configurations needed for defining $H_z$ are included.}

Likewise, corresponding to a partition of $q^4\in\sS_\delta$ into families as in Section~\ref{sec:multieq}, we regard the Hamiltonian $H_j^P$ of family $j$ defined in \eqref{HjPdef} as acting on functions of $q^4$. We set
\be
K_j^P = i\partial_{t_j} - H_j^P
\ee
(adopting notation from \cite[Sec.~5.3]{pt:2013c}).
This operator can also be defined as acting on $C^\infty(\cQ^{s4})$ by
\be
K_j^P=i\sum_{x_k\in P_j}\partial_{x_k^0}+i\sum_{y_\ell\in P_j} \partial_{y^0_\ell}-\sum_{z\in P_j} H_z.
\ee
Since $K_j^P\Phi(q^4)$ depends only on values of $\Phi$ in a neighborhood of $q^4$ and with a $y$-particle added or removed, the action of $K_j^P$ at $q^4\in \sS_\delta$ does not depend on whether we regard it as an operator on $C^\infty(\cQ^{s4})$ or $C^\infty(\sS^s_\delta)$.

\begin{lemma}\label{lem:commutator}
The commutator condition
\be\label{jconsist}
\bigl[K_j^P, K_{j'}^P \bigr] =0
\ee
holds in $C^\infty(\widehat{\sS^s_\delta})$ and $C^\infty(\sS^s_\delta)$ at $q^4\in\widehat{\sS_\delta^P}$, respectively $q^4\in \sS_\delta^P$, for all $j,j'\in\{1,\ldots,J(P)\}$.
\end{lemma}

This follows from

\begin{lemma}\label{lem:commutator2}
On $C^\infty(\cQ^{s4})$, the commutator $\bigl[K_j^P, K_{j'}^P \bigr]$ vanishes at every $q^4\in\widehat{\sS_\delta^P}$. Even more,
\begin{subequations}
\begin{align}
\bigl[ i\partial_{y^0_\ell}-H_{y_\ell}, i\partial_{y^0_{\ell'}}-H_{y_{\ell'}} \bigr]&=0 \label{comyy}\\
\bigl[ i\partial_{x^0_k}-H_{x_k}, i\partial_{y^0_\ell}-H_{y_\ell} \bigr]&=0 \label{comxy}\\
\intertext{at every $q^4\in\cQ^4$, and}
\bigl[ i\partial_{x^0_k}-H_{x_k}, i\partial_{x^0_{k'}}-H_{x_{k'}} \bigr]&=0 \label{comxx}
\end{align}
whenever $x_k$ and $x_{k'}$ keep their safety distance.
\end{subequations}
\end{lemma}

\begin{proof}
The $yy$ commutator \eqref{comyy} can easily be seen to vanish, as the free Hamiltonians are time-independent and commute everywhere.

The $xy$ commutator \eqref{comxy} yields
\begin{multline}
\bigl[ i\partial_{x^0_k}-H_{x_k}, i\partial_{y^0_\ell}-H_{y_\ell} \bigr] =\\
- \frac{1}{\sqrt{N}}\sum_{r'_k} \Bigl((i\partial_{y^0_\ell}-H_{y_\ell}^\free)\G (y_\ell-x_k)\Bigr)_{r_k r'_k s_\ell} \Phi_{r'_k\widehat{s_\ell}}\bigl(x^{4M}, y^{4N} \backslash y_\ell\bigr)\,.
\end{multline}
It vanishes for all $\Phi$ if and only if the $\G$'s are chosen to be solutions of the free Dirac equation \eqref{Gdeltadef1}.

The $xx$ commutator \eqref{comxx} can be computed to be
\be
\begin{aligned}
	&\Bigl[ i\partial_{x^0_k}-H_{x_k}, i\partial_{x^0_{k'}}-H_{x_{k'}} \Bigr] = \Bigl[ H_{x_k}^{\inter}, H_{x_{k'}}^{\inter} \Bigr] \\
&= \sum_s\int d^3\by \, \cutoff(\by-\bx_k) \Bigl[g^*_{r'_k r_k s} \, \G_{r_{k'} r'_{k'} s}\bigl(x^0_k-x^0_{k'},\by-\bx_{k'}\bigr)\\
& \hspace{40mm} -g_{r_k r'_k s} \, \G^*_{r'_{k'} r_{k'} s}\bigl(x^0_k-x^0_{k'},\by-\bx_{k'}\bigr) \Bigr]
\end{aligned}
\label{eq:xxcomm_tilde}
\ee
(with $r'$ indices acting on $\Phi$). 
Since $\supp_3 \, \cutoff \subseteq \overline{B_\delta(\bzero)}$ and $\supp_3 \, \G_{rr's}(t,\cdot)\subseteq \overline{B_{\delta+|t|}(\bzero)}$, the last expression will vanish if the first safety distance condition in \eqref{eq:SdeltaN} holds true, i.e., if
\be
	\Vert \bx_k - \bx_{k'} \Vert > \vert x^0_k - x^0_{k'} \vert + 2 \delta.
\label{eq:safetydistance1}
\ee
(In the special cases that $g_{rr's}=\delta_{rr'} \, g_s$ or $g_{rr's}=h_{rr'} g_s$ with self-adjoint matrix $h$, \eqref{eq:xxcomm_tilde} also vanishes when $x^0_k=x^0_{k'}$, but we do not use this fact.)
\end{proof}

\subsection{Existence}
\label{subsec:comp}
\label{subsec:exist}

We now construct a solution $\Phi\in C^\infty(\sS^s_\delta)$ to the multi-time equations \eqref{multi34} from initial data $\Psi_0\in\sH_c^\infty$. 
The construction proceeds in a way similar to the proof of Lemma~\ref{lem:unique} (and to the construction in Sections 5.3 and 5.5 of \cite{pt:2013c}). We define $\Phi$ on $\sS_{\delta J}$ by induction over $J$. Put briefly, to obtain $\Phi$ at a configuration with $J$ time values $t_1<t_2<\ldots<t_J$, we consider the families defined by a common time value, first evolve all particles to time $t_1$, then those belonging to families 2 and up to time $t_2$, and so on; see Figure~\ref{fig:IVP1}. Since the number $M_J$ of $x$-particles in family $J$ may be less than $M$, we will have to use the appropriate versions of the Hilbert space $\sH$ and the Lemmas~\ref{lem:selfadjoint} through \ref{lem:smoothsingletime}.

\begin{figure}[hbt]
    \centering
    \def\coordsys at (#1,#2){
\draw[thick,->] (#1-0.2,#2) -- ++(4.3,0) node[anchor=south west] {$ x $};
\draw[thick,->] (#1,#2-1) -- ++(0,4) node[anchor=south east] {$ t $};
\node at (#1-0.5,#2) {0};
}
\begin{tikzpicture}
	\coordsys at (1,7);
	\draw[blue, thick] (0.8,7) -- ++(4.2,0);
	\draw[blue, dotted] (0.9,8) -- ++(2.1,0) node[anchor = south] {$ P_1 $};
	\draw[blue, dotted] (0.9,8.5) -- ++(0.6,0) node[anchor = south] {$ P_2 $};
	\draw[blue, dotted] (0.9,9.25) -- ++(3.85,0)node[anchor = south] {$ P_3 $};
	\node[blue] at (0.5,8) {$ t_1 $};
	\node[blue] at (0.5,8.5) {$ t_2 $};
	\node[blue] at (0.5,9.25) {$ t_3 $};
	\node[blue] at (5,6.7) {$ \Psi_0 $};
	\node[blue] at (3,6.5) {desired};
	\filldraw[fill=blue,draw=blue] (2.5,8) circle (0.1);
	\filldraw[fill=blue,draw=blue] (2.75,8) circle (0.1);
	\filldraw[fill=blue,draw=blue] (3,8) circle (0.1);
	\filldraw[fill=blue,draw=blue] (1.25,8.5) circle (0.1);
	\filldraw[fill=blue,draw=blue] (1.5,8.5) circle (0.1);
	\filldraw[fill=blue,draw=blue] (4.75,9.25) circle (0.1);

	\coordsys at (7,7);
	\draw[blue, thick] (6.8,8) -- ++(4.2,0);
	\node[blue] at (6.5,8) {$ t_1 $};
	\node[blue] at (11.5,8) {$ \Psi_{t_1} $};
	\node[blue] at (9,6.5) {$ \boldsymbol{U}_{\ge 1}(t_1) $};
	\filldraw[fill=blue,draw=blue] (8.5,8) circle (0.1);
	\filldraw[fill=blue,draw=blue] (8.75,8) circle (0.1);
	\filldraw[fill=blue,draw=blue] (9,8) circle (0.1);
	\filldraw[fill=blue,draw=blue] (7.25,8.5) circle (0.1);
	\filldraw[fill=blue,draw=blue] (7.5,8.5) circle (0.1);
	\filldraw[fill=blue,draw=blue] (10.75,9.25) circle (0.1);
	\draw[->,line width=1,blue] (8.75,7.02) -- node[right] {$\boldsymbol{H}_1$} ++(0,0.9) ;
	\draw[->,line width=1,blue] (7.385,7.02) -- node[right] {$\boldsymbol{H}_2$} ++(0,0.9) ;
	\draw[->,line width=1,blue] (10.75,7.02) -- node[left] {$\boldsymbol{H}_3$} ++(0,0.9) ;

	\coordsys at (1,1.5);
	\filldraw[fill=gray,fill opacity = 0.5,dotted] (2.5,2.5) -- (3,2.5) -- (3.5,3) -- (2,3) -- cycle;
	\draw[blue] (6.8,2.5) -- ++(4.2,0);
	\draw[blue] (6.8,3) -- ++(4.2,0);
	\draw[thick,blue] (6.8,3.75) -- ++(4.2,0);
	\node[blue] at (0.5,3) {$ t_2 $};
	\node[blue] at (5.3,3) {$ \Psi_{t_2} $};
	\node[blue] at (3,1) {$ \boldsymbol{U}_{\ge 2}(t_2 - t_1) $};
	\filldraw[fill=gray!50!white,draw=blue] (2.5,2.5) circle (0.1);
	\filldraw[fill=gray!50!white,draw=blue] (2.75,2.5) circle (0.1);
	\filldraw[fill=gray!50!white,draw=blue] (3,2.5) circle (0.1);
	\filldraw[fill=blue,draw=blue] (1.25,3) circle (0.1);
	\filldraw[fill=blue,draw=blue] (1.5,3) circle (0.1);
	\filldraw[fill=blue,draw=blue] (4.75,3.75) circle (0.1);
	\draw[->,line width=1,blue] (10.75,3.02) -- node[left] {$\boldsymbol{H}_3$} ++(0,0.6) ;
	
	\coordsys at (7,1.5);
	\filldraw[fill=gray,fill opacity = 0.5,dotted] (7.3,3.75) -- (8.5,2.5) -- (9,2.5) -- (10.25,3.75) --cycle;
	\filldraw[fill=gray,fill opacity = 0.5,dotted] (6.8,3.75) -- (6.8,3.5) -- (7.25,3) -- (7.5,3) -- (8.25,3.75) --cycle;
	\draw[blue] (0.8,2.5) -- ++(4.2,0);
	\draw[thick,blue] (0.8,3) -- ++(4.2,0);
	\node[blue] at (6.5,3.75) {$ t_3 $};
	\node[blue] at (11.3,3.75) {$ \Psi_{t_3} $};
	\node[blue] at (9,1) {$ \boldsymbol{U}_{\ge 3}(t_3 - t_2) $};
	\filldraw[fill=gray!50!white,draw=blue] (8.5,2.5) circle (0.1);
	\filldraw[fill=gray!50!white,draw=blue] (8.75,2.5) circle (0.1);
	\filldraw[fill=gray!50!white,draw=blue] (9,2.5) circle (0.1);
	\filldraw[fill=gray!50!white,draw=blue] (7.25,3) circle (0.1);
	\filldraw[fill=gray!50!white,draw=blue] (7.5,3) circle (0.1);
	\filldraw[fill=blue,draw=blue] (10.75,3.75) circle (0.1);
	\draw[->,line width=1,blue] (1.385,2.52) -- node[right] {$\boldsymbol{H}_2$} ++(0,0.4) ;
	\draw[->,line width=1,blue] (4.75,2.52) -- node[left] {$\boldsymbol{H}_3$} ++(0,0.4) ;

\end{tikzpicture}
    \caption{Time evolution for a partition $P=\{ P_1,P_2,P_3 \}$ into three sets. At each $ t_j $, the time evolution for particles in $ P_j $ generated by $ \bH_j $ is switched off. \x{The shaded regions indicate where the demand of being spacelike forbids particle coordinates of the multi-time configuration.}}
    \label{fig:IVP1}
\end{figure}
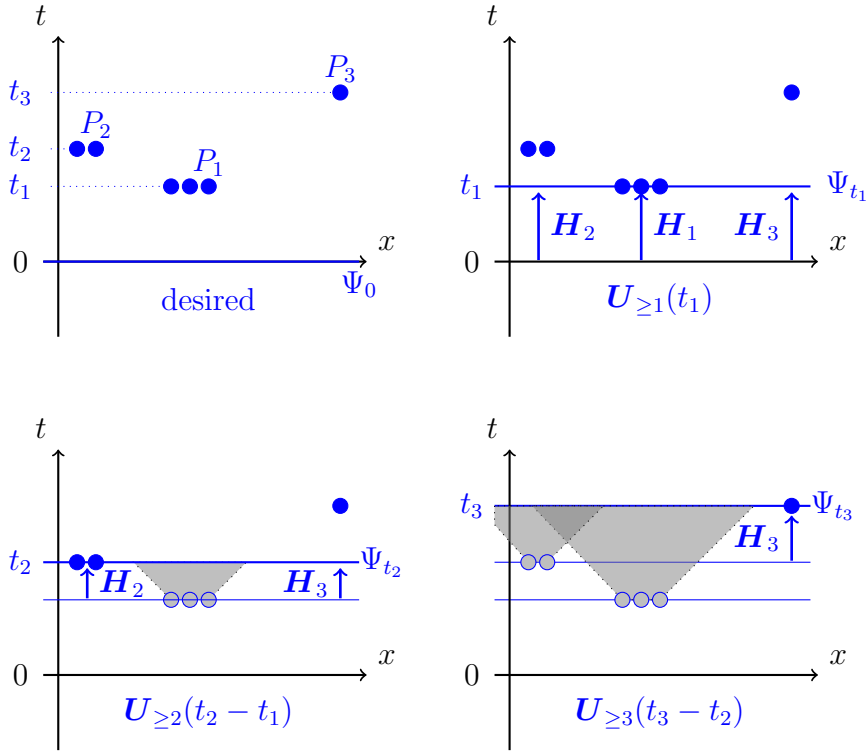

For $J=1$, we obtain $\Phi\equiv \Psi$ on $\sS^s_{\delta 1}$ from the single-time evolution; by Lemmas~\ref{lem:diffable} and \ref{lem:smootht}, it is smooth on $\sS_{\delta 1}\cong \RRR\times \cQ^3$ and belongs to $\sH_c^\infty$ for every fixed $t\in\RRR$.

The cases $J=2$ and $J=3$ allow a particularly simple construction that we want to describe first as it will play a role also for $J>3$; we describe it for $J=2$. (The same strategy can be applied if the number $M$ of $x$-particles is $\leq 3$.) Suppose that two regions $G_1,G_2\subset \RRR^3$ and times $t_1\leq t_2\in\RRR$ are such that $\{t_j\}\times G_j$ are $2\delta$-spacelike separated, i.e., for every $\bx_1\in G_1$ and $\bx_2\in G_2$, $\|\bx_1-\bx_2\|>2\delta+|t_1-t_2|$. Then for any $t\in [t_1,t_2]$, the sets $G'_j:= \Gr(G_j,|t_j-t|+\delta)$ are disjoint, and $\Phi$ on configurations concentrated in $\{t_1\}\times G_1 \cup \{t_2\}\times G_2$ (say with $M_1$ $x$-particles in $\{t_1\}\times G_1$ and $M_2$ in $\{t_2\}\times G_2$) is determined by initial data $\Psi_t\in \sH_{M_1}(G'_1) \otimes \sH_{M_2}(G'_2)$ and in fact given by
\be\label{PhiJ2}
\Phi(t_1,\cdot,t_2,\cdot) = W_{1,t_1-t} \otimes W_{2,t_2-t} \: \Psi_t
\ee
with $W_t$ as in \eqref{Wdef} and $W_{j,t}$ acting on $\sH_{M_j}(G'_j)$. Since, as explained before \eqref{Cdef}, for every configuration in $\sS_{\delta 2}$ there is $\varepsilon>0$ so that the union of the $\varepsilon$-balls around each particle are still $2\delta$-spacelike separated, $\Phi$ can be determined from $\Psi$ via \eqref{PhiJ2} everywhere on $\sS_{\delta 2}$. For $J>3$ time values, this strategy cannot be directly applied because for some $q^4\in\sS_{\delta J}$, there is no $t$ at which the $G'_j=\Gr(q_j,|t_j-t|+\delta)$ would all be mutually disjoint; for example, $q^4=(x_1...x_4)$ with $x_1=(0,0,0,0), x_2=(1,2,0,0), x_3=(3,5,0,0), x_4=(4,7,0,0)$. 

\bigskip

We now turn to the strategy for general $J$ by induction with anchor $J=1$. The induction hypothesis asserts that $\Phi\in C^\infty(\sS^s_{\delta J})$ is well defined and satisfies \eqref{multi34} and \eqref{summability} as well as a further condition that we will formulate in \eqref{cond} below. We assume it for $J$ and prove it for $J+1$. 

On $\sS_{\delta J+1}\setminus \sS_{\delta J}$, let us label the time variables so that $t_1<t_2<\ldots < t_{J+1}$. The strategy is to solve the multi-time equation \eqref{multi34} for $j=J+1$ in the variable $t_{J+1}$ (while keeping $t_1,\ldots, t_{J}$ unchanged) from initial data given by $\Phi$ on $\sS_{\delta J}$, i.e., for $t_{J+1}=t_{J}$. In fact, we solve it for $4^{M_1+N_1+...+M_{J}+N_{J}}$ functions, the components of $\Phi$ for different values of the spin indices for all particles in the families $1,\ldots, J$, where $(M_j,N_j)$ are the particle numbers in family $j$. Since the indices of other families are not acted upon in \eqref{multi34}, \eqref{multi34} can be solved separately for each choice of values for those indices. By propagation locality, the solution on $\sS(t_1...t_{J+1})=\cup_{N=0}^\infty \sS^{(N)}(t_1...t_{J+1})$ as in \eqref{St1tJ} is determined by initial data on $\sS(t_1...t_{J})$. We want to use Lemma~\ref{lem:smoothsingletime} to conclude that $\Phi$ exists and is smooth where $t_1<t_2<\ldots <t_{J+1}$. To this end, we regard $t_{J+1}$ as the variable $t$ of Lemma~\ref{lem:smoothsingletime}; we consider the evolution separately for every fixed choice of $M_1,N_1,\ldots, M_{J},N_{J}$; we regard the $t_1,\ldots,t_{J}$ and the configurations $q_1,\ldots,q_{J}$ of the families 1 through $J$ as the parameters $\lambda$ in Lemma~\ref{lem:smoothsingletime}, so $d=J+3M_1+3N_1+\ldots+3M_{J}+3N_{J}$. We know that the time evolution preserves the compactness of the 3-support in all space variables; however, in order to be able to apply Lemma~\ref{lem:smoothsingletime}, we need compact support in all $\lambda$ variables, including $t_1,\ldots,t_J$. Since wave functions do not have compact support on the time axis, we need to cut off the time dependence; this does not cause any harm because the solution provided by Lemma~\ref{lem:smoothsingletime} is obtained by solving the 1-time evolution for every value of $\lambda$ separately. That is, we apply Lemma~\ref{lem:smoothsingletime} to the function $f(t_1)\cdots f(t_J)\, \Phi(t_1,q_1,\ldots, t_J,q_J,q_{J+1})$ instead of $\Phi(t_1,q_1,\ldots, t_J,q_J,q_{J+1})$, where $f:\RRR\to\RRR$ is a smooth function such that $f=1$ on $[-T,T]$ and $f=0$ outside $[-2T,2T]$; for the desired result at $q^4\in \sS_{\delta, J+1}$, the value of $T$ must be chosen larger than all absolute time values in $q^4$. To fulfill the hypotheses of Lemma~\ref{lem:smoothsingletime}, we need that the smoothly cut off wave function lies in $\sH_{cd}^\infty$. We will assume a little more as part of the induction hypothesis:
\be\label{cond}
\forall f\in C_c^\infty(\RRR):\: f(t_1)\cdots f(t_J)\, \Phi\bigl(t_1,q_1,\ldots, t_J,q_J,q_{J+1} \bigr) \text{ possesses an extension in } \sH_{cd}^\infty.
\ee
(We talk about extension because $\Phi(t_1,q_1^{s3},...)$ is not defined for \emph{every} $q_1^{s3}$ but only those $\delta$-spacelike from the other families.) Then Lemma~\ref{lem:smoothsingletime} applies, and we obtain the desired function $\Phi$ on $\sS_{\delta, J+1}$ where all $|t_j|<T$, and on all of $\sS_{\delta, J+1}$ by letting $T\to\infty$. 

It remains to verify four things: 
(i)~that $\Phi$ satisfies the summability condition \eqref{summability};
(ii)~that $\Phi$ satisfies the multi-time equations \eqref{multi34} for \emph{all} $j$, not just $j=J+1$;
(iii)~that $\Phi$ is smooth (the issue here is whether the transition from $t_J<t_{J+1}$ to $t_J>t_{J+1}$ is smooth); and
(iv)~that \eqref{cond} holds on $\sS_{\delta,J+1}$.

\begin{itemize}
\item[(i)] follows from the unitarity of the single-time evolution and Lemma~\ref{lem:diffable}.

\item[(ii)] Since $\Phi$ was constructed on $\sS_{\delta,J+1}$ using the unitary time evolution in $t_{J+1}$, it is a solution of \eqref{multi34} for $j=J+1$ where $t_1<\ldots<t_J<t_{J+1}$ (in the strong sense for the same reasons as discussed after \eqref{1time}, i.e., $K_{J+1}\Phi=0$). Now we show that it is also a solution of \eqref{multi34} for every $j\leq J$, i.e., that $K_j\Phi=0$. By Lemma~\ref{lem:commutator}, $K_{J+1}K_j=K_jK_{J+1}$ on $\sS_{\delta,J+1}$, so 
\be\label{KJ+1KjPhi}
K_{J+1}K_j\Phi=K_jK_{J+1}\Phi=K_j 0=0\,.
\ee
By induction hypothesis, $K_j\Phi=0$ on $\sS_{\delta J}$. By \eqref{cond}, \eqref{HcLambdainftynorm}, and Lemma~\ref{lem:diffable}, $K_j\Phi$ also satisfies \eqref{summability} with $J+1$ instead of $J$ on $\sS_{\delta,J+1}$. By Lemma~\ref{lem:1timeunique}, $K_j\Phi=0$ on $\sS_{\delta,J+1}$.

\item[(iii)] By Lemma~\ref{lem:smoothsingletime}, $\Phi$ is smooth where $t_1\leq \ldots\leq t_J<t_{J+1}$, but it is not obvious that $\Phi$ is smooth on the set where $t_J=t_{J+1}$. To be sure, Lemma~\ref{lem:smoothsingletime} provides a smooth function of $t$ for $t>0$, $t=0$, and $t<0$, but when we reduce the time variable of the family $P_{J+1}$ to values less than $t_J$ we reach configurations for which $P_J$ and $P_{J+1}$ would be labeled $P_{J+1}$ and $P_J$ according to our prescription that the times are labeled increasingly. Still, for a fixed choice of the partition (in particular $P_J$ and $P_{J+1}$) and in a neighborhood of a 4-configuration in which $t_J=t_{J+1}$ while $q_J$ and $q_{J+1}$ keep the safety distance, the constructed function $\Phi$ agrees with \eqref{PhiJ2} with $t_J,t_{J+1}$ (the time variables of families $P_J$, $P_{J+1}$) playing the roles of $t_1,t_2$, in fact for $t_J>t_{J+1}$ as well as for $t_J<t_{J+1}$. But by Lemma~\ref{lem:smoothsingletime}, $\Phi$ given by \eqref{PhiJ2} is a smooth function. Hence, the constructed function $\Phi$ is smooth everywhere in $\sS_{\delta, J+1}$.

\item[(iv)] Statement \eqref{cond} for $J+1$ involves splitting $q_{J+1}$ into two families, $q_{J+1}=(\tilde q_{J+1},\tilde q_{J+2})$, considering them at time $t_{J+1}$, regarding $t_{J+1}$ and $\tilde q_{J+1}$ as parameters, and introducing a further factor $f(t_{J+1})$. To check the definition of $\sH_{cd}^\infty$, we know that $f\cdots f\, \Phi$ is smooth, has compact support in $t_{J+1}$ because of $f$, compact support in $t_1,\ldots,q_J$ by assumption, and compact 3-support in $q_{J+1}$ because, by Lemma~\ref{lem:supp}, the 3-support can grow (from that of $\Phi|_{\sS_{\delta J}}$) by at most $2T+\delta$; so it remains to check \eqref{HcLambdainftynorm}: Indeed, for any given degree $n$ of differentiation, at most the first $n$ derivatives of $f$ can occur (independently of $N$), all of which have finite $L^2$ norm; moreover, $t_{J+1}$-derivatives in \eqref{HcLambdainftynorm} (without loss of generality the rightmost derivatives in $\partial^\alpha$) can be replaced by $-iH_{J+1}$ as in \eqref{multi56}; now each term in $H_{J+1}\Phi^{(N)}$ involves either a spatial derivative of $\Phi^{(N)}$ (which obeys \eqref{HcLambdainftynorm} by Lemma~\ref{lem:diffable}) or $\int \cutoff\,\Phi^{(N+1)}$ or $\cutoff\,\Phi^{(N-1)}$, whose contributions to \eqref{HcLambdainftynorm} involve derivatives of $\cutoff$ of up to $n$-th order and a factor of $\sqrt{N}\leq N$ (so $m\to m+1$), so they remain finite since $\Phi$ itself obeys \eqref{HcLambdainftynorm} at every $t=t_{J+1}$ by Lemma~\ref{lem:diffable}. Thus, \eqref{cond} holds on $\sS_{\delta, J+1}$.
\end{itemize}
This completes the construction of $\Phi$ on $\sS_{\delta}$ and the proof that $\Phi$ is smooth and solves the multi-time equations \eqref{multi34}.

\bigskip

\noindent{\bf Remarks.}
\begin{enumerate}
\setcounter{enumi}{\theremarks}
\item One might be tempted to think that the following procedure yields an alternative construction of the solution. The idea is to construct a function $\Phi$ on $\cQ^4$ which agrees with the desired solution on $\sS_\delta$ (or $\widehat{\sS_\delta}$) by writing the multi-time equations in the form
\begin{align}
i\partial_{x^0_k}\Phi(q^4) &= H^{\free}_{x_k} \Phi(q^4) 
+ \sqrt{N+1} \sum_{r_k',s_{N+1}} \int_{B_{\delta+|x^0_k|}(\bx_k)} \hspace{-12mm} d^3\tilde\by\:\: \G^*_{r'_k r_k s_{N+1}}\bigl(-x^0_k,\tilde\by-\bx_k\bigr) \:\times \nonumber\\
&\qquad \times \:
\Phi^{(N+1)}_{r_k',s_{N+1}}\Bigl(x^{4M},\bigl(y^{4N}, (0, \tilde\by)\bigr)\Bigr) \nonumber\\
& + \frac{1}{\sqrt{N}} \sum_{\ell=1}^N \sum_{r'_k} \G_{r_k r'_k s_\ell}\bigl(y_\ell - x_k\bigr) \: \Phi_{r'_k\widehat{s_\ell}}^{(N-1)}\bigl(x^{4M}, y^{4N} \backslash y_\ell\bigr)  \label{multi9}\\[4mm]
i\partial_{y^0_{\ell}}\Phi(q^4) &= H_{y_\ell}^\free \Phi(q^4)\label{multi10}
\end{align}
and to integrate them in a particular order. Specifically, solve first the equation for $x_1$ from 0 to any desired $x_1^0$, keeping all other $x_k^0$ and $y_\ell^0$ at 0; then, solve the equation for $x_2$, and then for $x_3,\ldots,x_M$, keeping all $y_\ell^0$ at 0; finally, solve \eqref{multi10} for all $y$s. However, the function $\Phi$ thus obtained will not be a solution of our multi-time equations \eqref{multi34} on $\sS_\delta$; indeed, $\Phi$ does not even agree with $\Psi$ on the set $\sS_{\delta 1}$ of simultaneous configurations. That is because when all $x_k^0=t=y_\ell^0$, then
\begin{align}
\Phi &= e^{-i\bH_{y}t} \, \mathcal{T}e^{-i\int_0^t\bH_{x_M}(s)\, ds} \cdots \mathcal{T}e^{-i\int_0^t\bH_{x_1}(s)\, ds} \: \Psi_0 ~~\text{whereas}\\
\Psi &= e^{-i(\bH_{y}+\bH'_{x_1}+...+\bH'_{x_M})t} \: \Psi_0 
\end{align}
with $\mathcal{T}e$ the time ordered exponential, $\bH_{y}=\bH_{y_1}+...+\bH_{y_N}$ on the $N$-sector and $\bH_{x_k}(x_k^0)$, $\bH'_{x_k}$, and $\bH_{y_\ell}$ the right-hand sides of \eqref{multi9}, \eqref{multi7}, and \eqref{multi10}. That the expressions are not the same is strongly suggested by the facts that the $\bH_{x_k}(s)$ do not commute with each other (and $e^A e^B\neq e^{A+B}$ when $AB\neq BA$) and $\bH'_{x_k}\neq \bH_{x_k}$.
\end{enumerate}
\setcounter{remarks}{\theenumi}

\label{subsec:char}

\begin{proof}[Proof of Theorem~\ref{thm:1}]
Theorem~\ref{thm:1} follows by putting together the statements obtained about the existence of solutions with Lemmas~\ref{lem:unique} and \ref{lem:4supp}. It remains to prove the permutation (anti-)symmetry of $\Phi$. Since all Hamiltonians and time evolution procedures are invariant under relabeling, $\Phi$ with permuted labels is the solution with the correspondingly permuted initial condition. Thus, if permuting two $x$'s changes the sign of the initial data, then the corresponding permutation in $\Phi$ changes the sign of $\Phi$; likewise, if permuting two $y$'s leaves $\Psi_0$ invariant, then the corresponding permutation in $\Phi$ leaves $\Phi$ invariant, as claimed.
\end{proof}

\subsection{Proof of Remark~\ref{rem:Shat}}
\label{sec:Shat}

\begin{proof}
Given $\Phi$ on $\sS_{\delta}$, it can be extended to $\widehat{\Phi}$ on $\widehat{\sS_{\delta}}$ by freely varying all $y_\ell$, i.e.,
\be
\widehat\Phi^{(N)}(x^{4M},\hat{y}^{4N})= e^{-i\bH_{y_1}(\hat y_1^0-y_1^0)} \cdots e^{-i\bH_{y_N}(\hat y_N^0-y_N^0)} \: \Phi^{(N)}(x^{4M},y^{4N})\,.
\ee
The resulting $\widehat{\Phi}$ is smooth, for example by Lemma~\ref{lem:smoothsingletime}. Since the $\bH_{y_\ell}$ (and the $H_{y_\ell}$) commute pairwise, $\widehat\Phi$ satisfies the multi-time equations \eqref{multi56}. Since the strong solution of the 1-particle Dirac equation is unique and $\Phi$ is unique by Theorem~\ref{thm:1} and \eqref{summability}, $\widehat\Phi$ is unique.
\end{proof}

\section{Conclusions}
\label{sec:conclusions}

In this paper, we have provided a rigorous study of a system of 
multi-time equations \eqref{multi34} for a model quantum field theory with UV cut-off. 
We have proved the existence and uniqueness of solutions on the set $\sS_\delta$ of 
$\delta$-spacelike configurations (and even on a larger set $\widehat{\sS}_\delta$ allowing arbitrary 
points for bosons), and thus the consistency of the multi-time 
equations. This result supports the viability of multi-time wave 
functions as a covariant expression of the quantum state in the 
particle-position representation, although the model considered here is 
not fully covariant, partly because of the UV cut-off.

Our proof is one of the first rigorous consistency proofs for multi-time 
formulations of quantum field theories. Only two other results of this 
kind are known to date: First, for a similar set of equations, first 
proposed by Dirac, Fock, and Podolsky \cite{dfp:1932} and involving a 
fixed number of time variables, consistency was recently proved in 
\cite{ND:2019}. And second, for a model with a variable number of time 
variables in 1+1 dimensions and with a cut-off in the particle number, 
consistency was recently proved in \cite{LN:2018}.

Since we considered solutions in the classical sense (i.e., 
differentiable functions, rather than weak derivatives), we also had to 
prove smoothness of the solutions. Furthermore, our proof establishes in 
particular that the non-rigorous arguments for consistency in 
\cite{pt:2013c} also apply rigorously. In fact, it turns out that all 
considerations of \cite{pt:2013c} are rigorously valid if formulated 
appropriately in view of the UV cut-off. In particular, the model 
satisfies, up to a tolerance of the size $\delta$ of the UV cut-off, the 
conditions ``propagation locality'' and ``interaction locality'' that 
played important roles for the derivation of Born's rule on arbitrary 
Cauchy surfaces in \cite{LT:2017}.

For the future, it would be of interest to move towards more realistic 
models of quantum field theory and to obtain a multi-time formulation of 
quantum electrodynamics (QED). For example, we believe that a consistent 
multi-time formulation of the Landau-Peierls model of QED \cite{LP:1930} 
is possible and will have the advantage of allowing to switch easily 
between electromagnetic field tensors and vector potentials, which is 
not possible in the single-time formulation used by Landau and Peierls.

\bigskip

\noindent\textit{Acknowledgments.}
We are grateful to Dirk-Andr\'e Deckert, Matthias Lienert, S\"oren Petrat, and Stefan Teufel 
for helpful discussions.

\end{document}